\documentclass{article}
  
\usepackage[letterpaper, margin=1.1in]{geometry}

\usepackage{graphicx,xspace,xcolor} 
\newif\iflong
\newif\ifshort

 \longtrue

\iflong
\else
\shorttrue
\fi

\newcommand{\m}[1]{V(#1)}
\newcommand{\md}[1]{mid(#1)}
\newcommand{\emptytype}{empty type}
\newcommand{\fulltype}{full type}
\newcommand{\enref}[1]{\emph{(#1)}}

\usepackage{todonotes}
\presetkeys     {todonotes}     {backgroundcolor=yellow}{}

\usepackage{amsthm}
\usepackage{amsmath,amssymb}
\usepackage{enumerate,verbatim}

\usepackage{thmtools,hyperref}
\usepackage[list=off]{caption}
\usepackage{cleveref}

\usepackage{tikz,tikz-cd}
\usetikzlibrary{arrows,cd,positioning,shapes,patterns}
\usetikzlibrary {decorations.pathmorphing, decorations.pathreplacing, decorations.shapes}
\usepackage{subcaption}
\usetikzlibrary{calc}
\tikzset{cross/.style={cross out, draw=black, minimum size=2*(#1-\pgflinewidth), inner sep=0pt, outer sep=0pt},
 cross/.default={1pt}}

\newcommand{\SB}{\{\,}
\newcommand{\SM}{\;{|}\;}
\newcommand{\SE}{\,\}}

\newcommand{\RRR}{\mathcal{R}}

\newcommand{\ARC}{\mathcal{A}}

\newcommand{\XXX}{\mathcal{X}}
\newcommand{\BBB}{\mathcal{B}}
\newcommand{\PPP}{\mathcal{P}}

\newcommand{\xor}{\ensuremath{\oplus}\xspace}

\newcommand{\type}{\Gamma}

\newcommand{\bigoh}{\mathcal{O}}

\newcommand{\INVP}{\textup{IN}}

\newcommand{\cc}[1]{{\mbox{\textnormal{\textsf{#1}}}}\xspace}      

\renewcommand{\P}{\cc{P}}
\newcommand{\NP}{\cc{NP}}

\newcommand{\XP}{\cc{XP}}

\newcommand{\ie}{\textsl{i.e.}\xspace}
\newcommand{\eg}{\textsl{e.g.}\xspace}

\newcommand{\noprop}[1]{}

\usepackage{boxedminipage}

\newcommand{\problemdef}[3]{
	\begin{center}
		\begin{boxedminipage}{.99\textwidth}
			\textsc{{#1}}\\[2pt]
			\begin{tabular}{ r p{0.8\textwidth}}
				\textit{~~~~Instance:} & {#2}\\
				\textit{Question:} & {#3}
			\end{tabular}
		\end{boxedminipage}
	\end{center}
}

\newcommand{\CCC}{\mathcal{C}}

\let\phi=\varphi
\let\epsilon=\varepsilon

\newcommand{\specialfont}[1]{{\normalfont\slshape #1}}

\newcommand{\bw}{\text{\specialfont{bw}}}
\newcommand{\tw}{\text{\specialfont{tw}}}

\newcommand{\SH}{\textsc{\textup{SUBHAM}}\xspace}

 \newcommand{\pe}{\textsc{Pe}}
\newcommand{\expa}{\textsc{exp}}
\newcommand{\sk}{\textsc{Sk}}
\newcommand{\SPQR}{\textsc{SPQR}}

\newcommand{\concat}{\circ}

\theoremstyle{plain}
\newtheorem{theorem}{Theorem}
\newtheorem{corollary}[theorem]{Corollary}
\newtheorem{lemma}[theorem]{Lemma}

\newtheorem{observation}[theorem]{Observation}
\newtheorem{fact}[theorem]{Fact}

\title{A Tight Subexponential-time Algorithm for Two-Page Book Embedding}

\author{Robert Ganian\thanks{Algorithms and Complexity Group, TU Wien,
    Vienna, Austria, \texttt{rganian@gmail.com}} \and
  Haiko M\"uller\thanks{School of Computing, University of Leeds, UK,
    \texttt{h.muller@leeds.ac.uk}} \and
  Sebastian Ordyniak\thanks{School of Computing, University of Leeds, UK,
  \texttt{sordyniak@gmail.com}} \and
Giacomo Paesani\thanks{School of Computing, University of Leeds, UK, \texttt{g.paesani@leeds.ac.uk}} \and
Mateusz Rychlicki\thanks{School of Computing, University of Leeds, UK,
  \texttt{mkrychlicki@gmail.com}}}

\begin{document}
\maketitle

\begin{abstract}
A book embedding of a graph is a drawing that maps vertices onto a
line and edges to simple pairwise non-crossing curves drawn into
``pages'', which are half-planes bounded by that line. Two-page book
embeddings, i.e., book embeddings into 2 pages, are of special importance as they are both \NP-hard to compute and have specific applications. We obtain a $2^{\bigoh(\sqrt{n})}$ algorithm for computing a book embedding of an $n$-vertex graph on two pages---a result which is asymptotically tight under the Exponential Time Hypothesis. As a key tool in our approach, we obtain a single-exponential fixed-parameter algorithm for the same problem when parameterized by the treewidth of the input graph. We conclude by establishing the fixed-parameter tractability of computing minimum-page book embeddings when parameterized by the feedback edge number, settling an open question arising from previous work on the problem.
  \end{abstract}

\section{Introduction}

Book embeddings of graphs are drawings centered around a line, called the \emph{spine}, and half-planes bounded by the spine, called \emph{pages}. 
In particular, a $k$-page book embedding of a graph $G$ is a drawing
which maps vertices to distinct points on the spine and edges to
simple curves on one of the $k$ pages such that no two edges on the same page cross~\cite{DBLP:journals/jct/BernhartK79}. 
These embeddings have been the focus of extensive study to date~\cite{FraysseixMP95,DBLP:journals/dmtcs/DujmovicW04,DujmovicW07,ENDO199787,DBLP:journals/dam/GanleyH01,DBLP:journals/jal/Malitz94a,Yan89}, among others due to their classical applications in VLSI, bio-informatics, and parallel computing~\cite{doi:10.1137/0608002,DBLP:journals/dmtcs/DujmovicW04,Haslinger1999}. 

Every $n$-vertex graph is known to admit an $\lceil \frac{n}{2}\rceil$-page book embedding~\cite{DBLP:journals/jct/BernhartK79,doi:10.1137/0608002,GyarfasL85}, but in many cases it is possible to obtain book embeddings with much fewer pages.
Particular attention has been paid to two-page embeddings, which have specifically been used, e.g., to represent RNA pseudoknots~\cite{Haslinger1999,nowicka2023automated}. 
The class of graphs that can be embedded on two pages was studied by Di Giacomo and Liotta~\cite{GiacomoL10}, Heath~\cite{0608018} as well as by other authors~\cite{AbregoAFRS13}, and was shown to be a superclass of planar graphs with maximum degree at most $4$~\cite{DBLP:journals/algorithmica/BekosGR16}.

While two-page book embeddings are a special class of planar embeddings, they are not polynomial-time computable unless $\P=\NP$. Indeed, a graph admits a two-page book embedding if and only if it is \emph{subhamiltonian} (\ie, is a subgraph of a planar Hamiltonian graph)~\cite{DBLP:journals/jct/BernhartK79} and testing subhamiltonicity is an \NP-hard problem~\cite{doi:10.1137/0608002}. On the other hand, the aforementioned problem of constructing a two-page book embedding (or determining that none exists)---which we hereinafter call \textsc{Two-Page Book Embedding}---becomes linear-time solvable if one is provided with a specific ordering of the $n$ vertices of the input graph along the spine~\cite{Haslinger1999}. While \textsc{Two-Page Book Embedding} can be seen to admit a trivial brute-force $2^{\bigoh(n\cdot \log n)}$ algorithm, it has also been shown to be solvable in $2^{\bigoh(n)}$ time---in particular, one can branch to determine the allocation of edges into the two pages and then solve the problem via dynamic programming on \SPQR{} trees~\cite{AngeliniBB12,hong2009two,HongN18}. 

\subparagraph{Contribution.}
As our main contribution, we break the single-exponential barrier for \textsc{Two-Page Book Embedding} by providing an algorithm that solves the problem in $2^{\bigoh(\sqrt{n})}$ time. Our algorithm is exact and deterministic, and avoids the single-exponential overhead of branching over edge allocations to pages by instead attacking the equivalent subhamiltonicity testing formulation of the problem. It is also asymptotically optimal under the Exponential Time Hypothesis~\cite{ImpagliazzoPZ01}: there is a well-known quadratic reduction that excludes any $2^{o(\sqrt{n})}$ algorithm for \textsc{Hamiltonian Cycle} on cubic planar graphs~\cite{Hamcyclereduction}, and a linear reduction from that problem (under the same restrictions) to subhamiltonicity testing~\cite{Subhamcyclereduction} then excludes any $2^{o(\sqrt{n})}$ algorithm for our problem of interest.

The central component of our result is a non-trivial dynamic programming procedure that solves \textsc{Two-Page Book Embedding} in time $2^{\bigoh(\tw)}\cdot n$, where $\tw$ is the treewidth of the input graph. The desired subexponential algorithm then follows by the well-known fact that $n$-vertex planar graphs have treewidth at most $\bigoh(\sqrt{{n}})$~\cite{GuT12,Marx20,RobertsonST94}. But in addition to that, we believe our single-exponential treewidth-based algorithm to be of independent interest also in the context of parameterized algorithmics~\cite{DBLP:books/sp/CyganFKLMPPS15,DBLP:series/txcs/DowneyF13}.

Indeed, while \textsc{Two-Page Book Embedding} was already shown to be fixed-parameter tractable w.r.t.\ treewidth (i.e., to admit an algorithm running in time $f(\tw)\cdot n$) by Bannister and Eppstein~\cite{DBLP:journals/jgaa/BannisterE18}, that result crucially relied on Courcelle's Theorem~\cite{DBLP:journals/iandc/Courcelle90}. More specifically, they showed that the required property can be encoded via a constant-size sentence in Monadic Second Order logic, which suffices for fixed-parameter tractability---but unfortunately not for a single-exponential algorithm, and a direct dynamic programming algorithm based on the characterization employed there seems to necessitate a parameter dependency that is more than single-exponential. Moreover, it is not at all obvious how one could employ convolution-based tools---which have successfully led to $2^{\bigoh(tw)}\cdot n$ algorithms for, e.g., \textsc{Hamiltonian Cycle}~\cite{BodlaenderCKN15,CyganKN18,CyganNPPRW22}---for our problem of interest here.

Instead, we obtain our results by employing dynamic programming along a \emph{sphere-cut decomposition}---a type of branch decomposition specifically designed for planar graphs of small treewidth~\cite{DornPBF10}. However, unlike in previous applications of sphere-cut decompositions~\cite{JacobP22,MarxPP22}, our algorithm requires the nooses delimiting the bags in the sphere-cut decomposition to admit a fixed drawing since our arguments rely on constructing a hypothetical solution (a subhamiltonian curve) that is ``well-behaved'' w.r.t.\ a fixed set of curves. While this would typically lead to extensive case analysis to compute the records of a parent noose from the records of the children, we introduce a generic framework that allows us to transfer records from child to parent nooses via XOR operations. We believe that this may be of broader interest, especially when working with problems which require one to enhance the embedding or drawing of an input graph.

In the final part of the article, we turn our attention to the parameterized complexity of computing book embeddings. While \textsc{Two-Page Book Embedding} is fixed-parameter tractable when parameterized by the treewidth of the input graph, the only graph parameter which has been shown to yield fixed-parameter algorithms for computing $\ell$-page book embeddings for $\ell>2$ is the \emph{vertex cover number}\footnote{The vertex cover number is the minimum size of a vertex cover, and represents a much stronger restriction on the structure of the input graphs than, e.g., treewidth.}~\cite{BhoreGMN20b}. Whether this tractability result also holds for other structural graph parameters such as treewidth, \emph{treedepth}~\cite{DBLP:books/daglib/0030491} or the \emph{feedback edge number}~\cite{UhlmannW13} has been stated as an open question in the field\footnote{E.g., at \textbf{Advances in Parameterized Graph Algorithms} (Spain, May 2--7 2022) and also at Dagstuhl seminar 21293 \textbf{Parameterized Complexity in Graph Drawing}~\cite{GanianMNZ21}.}.
We conclude by providing a novel fixed-parameter algorithm for computing $\ell$-page book embeddings (or determining that one does not exist) under the third parameterization mentioned above---the feedback edge number, i.e., the edge deletion distance to acyclicity. This result is complementary to the known vertex-cover based fixed-parameter algorithm, and can be seen as a necessary stepping stone towards eventually settling the complexity of computing $\ell$-page book embeddings parameterized by treewidth. Moreover, since the obtained kernel is linear in the case of $\ell=2$, the obtained kernel allows us to generalize our main algorithmic result to a run-time of $2^{\bigoh(\sqrt{k})}\cdot n^{\bigoh(1)}$ where $k$ is the feedback edge number of the input graph.

\section{Preliminaries}\label{sec:pre}
\newcommand{\fen}{\cc{fen}}
\newcommand{\nin}{\cc{min}}
\newcommand{\short}{\psi}
\newcommand{\thi}{\cc{th}}
\newcommand{\seqb}{\sigma}
 \subparagraph{Basic Notions.}
We use basic terminology for graphs and multi-graphs~\cite{Graphs}, and
assume familiarity with the basic notions of parameterized complexity and fixed-parameter tractability~\cite{DBLP:books/sp/CyganFKLMPPS15,DBLP:series/txcs/DowneyF13}. 
The \emph{feedback edge number} of $G$, denoted by
$\fen(G)$, is the minimum size of any \emph{feedback edge set} of $G$,
\ie, a set $F \subseteq E(G)$ such that $G-F=(V(G),E(G)\setminus F)$ is acyclic. 
\iflong
\begin{fact}\label{fact:comp-fes}
  Let $G$ be a graph. Then, a minimum feedback edge set of $G$
  can be computed in time $\bigoh(|V(G)|+|E(G)|)$.
\end{fact}
\begin{proof}
  The theorem follows because any minimum feedback edge set is equal
  to $E(G)-E(F)$, where $F$ is a spanning forest of $G$, together with
  the fact that we can compute a spanning forest of $G$ in time $\bigoh(|V(G)|+|E(G)|)$.
\end{proof}
\fi
\ifshort
For a face $f$ of a plane graph, we use $\seqb(f)$ to denote the cyclic sequence of the vertices obtained by traversing the closed curve representing the border of $f$ in a clock-wise manner.
 \fi
\iflong
A {\it cut vertex} of a multi-graph is a vertex whose removal increases the number of connected components. A connected multi-graph that has no cut vertex is called {\it biconnected}. 

We say that a multi-graph $G$ is \emph{planar} if it
admits a \emph{planar drawing}, \ie, a drawing in the plane
in such a way that its edges are drawn as simple curves which pairwise intersect only at their endpoints. Let
$D$ be a planar drawing of $G$ and $f$ a face of $D$. We denote by
$V(f)$ ($E(f)$) all vertices (edges) of $G$ incident with $f$.

Every
face of a connected planar graph equipped with a drawing induces a cyclic sequence $\seqb(f)$ of the vertices in $V(f)$, \ie, the cyclic sequence is obtained by traversing the closed
curve representing the border of $f$ in a clock-wise manner. Note that
while $\seqb(f)$ can repeat vertices, this is no longer the case if
the graph is biconnected, in which case $\seqb(f)$ induces a cyclic
order of the vertices in $V(f)$. For convenience, we will represent
cyclic orders by sequences; note that each cyclic order on $n$
elements can be equivalently represented by one of $n$ sequences (one for each
starting element). For instance, the cyclic orders represented by the
sequences $(a_1,\dotsc,a_\ell)$ and
$(a_i,\dotsc,a_\ell,a_1,\dotsc,a_{i-1})$ are the same for every $i$
with $1 \leq i \leq \ell$.

The following basic observations about planar graphs will be useful later:

\begin{observation}\label{obs:face}
  Let $G$ be a graph with planar drawing $D$ and let $f$ be a face of $D$.
  Then, we can draw a simple curve inside $f$ between any two distinct
  points in $f$ or its border.
  Moreover, if $G$ is connected and $\seqb(f)=(v_1,\dotsc,v_\ell)$, then drawing a curve
  inside $f$ between
  $v_i$ and $v_j$ with $i<j$ splits $f$ into two faces $f_1$ and $f_2$
  such that $\seqb(f_1)=(v_i,\dotsc,v_j)$, $\seqb(f_2)=(v_j,\dotsc,v_i)$. 
\end{observation}  

\begin{observation}\label{obs:face-nopl}
  Let $f$ be a face of a planar drawing $D$ of a connected graph $G$ and let
  $(v_1,v_2,v_3,v_4)$ be a subsequence of $\seqb(f)$ such that
  $v_i\neq v_j$ for every distinct $i$ and $j$.
  Then every $v_1$-$v_3$ path must intersect every $v_2$-$v_4$ path in $G$ in at least one vertex.
      \end{observation}

\fi

\subparagraph{Book Embeddings and Subhamiltonicity.}
An $\ell$-page book embedding of a multi-graph $G=(V,E)$ will be denoted by
a pair $\langle \prec, \sigma \rangle$, where $\prec$ is a linear
order of $V$, and $\sigma \colon E \rightarrow [\ell]$ is a function
that maps each edge of $E$ to one of $\ell$ pages $[\ell] = \{1, 2,
\dots, \ell\}$. In an $\ell$-page book embedding $\langle \prec,
\sigma \rangle$ it is required that for no pair of edges $uv, wx \in
E$ with $\sigma(uv) = \sigma(wx)$ the vertices are ordered as $u \prec
w \prec v \prec x$, i.e., each page must be crossing-free. The
\emph{page number} of a graph $G$ is the minimum number $\ell$ such
that $G$ admits an $\ell$-page book embedding. The general problem of computing the page number of an input graph is thus:

\newcommand{\bt}{{\sc Book Thickness}}
\problemdef{\bt}
{A multi-graph $G$ with $n$ vertices and a positive integer $\ell$.}
{Does $G$ admit a $\ell$-page book embedding?}

\ifshort
It is known that a multi-graph admits a 2-page book embedding if and only if it is \emph{subhamiltonian}, i.e., if it has a planar Hamiltonian
supergraph on the same vertex
set~\cite{DBLP:journals/jct/BernhartK79}; an illustration is provided in \Cref{fig:2page}.
Hence, the problem of deciding whether a multi-graph has page number $2$ can be equivalently stated as:
\fi
\iflong
It is known that a multi-graph admits a 2-page book embedding if and only if it is \emph{subhamiltonian}, i.e., if it has a planar Hamiltonian
supergraph~\cite{DBLP:journals/jct/BernhartK79}; see \Cref{fig:2page}
for an illustration. It is also known (and also easy to observe) that if $G$ is subhamiltonian, then it
has a planar Hamiltonian supergraph $G'$ with $V(G')=V(G)$ and
$E(G')\setminus E(G)=E(H)$, where $H$ is a Hamiltonian cycle
in $G'$~(see, \eg,~\cite{HK19}). Hence, the problem of deciding whether a graph has page number $2$ can be equivalently stated as:
\fi

\problemdef{{\sc Subhamiltonicity (\SH)}}
{A multi-graph $G$ with $n$ vertices.}
{Is $G$ subhamiltonian?} 

Since the transformation between 2-page book embeddings and
Hamiltonian cycles of supergraphs is constructive in both directions, a constructive algorithm for \SH\ (such as the one presented here) allows us to also output a 2-page book embedding for the graph.

Let $G$ be subhamiltonian. For a Hamiltonian cycle $H$ on $V(G)$ (where $H$ is not necessarily a subgraph of $G$), we
denote by $G_H$ the graph obtained from $G$ after adding the edges of
$H$ and say that $H$ is a \emph{witness} for $G$ if $G_H$ is planar.
A drawing $D$ of $G$ \emph{respects} $H$ if $D$ can be
completed to a planar drawing of $G_H$ by only adding the edges of $H$. We extend the notion of ``witness'' to include all the information defining the solution as follows: a tuple $(D,D_H,G_H,H)$
is a \emph{witness} for $G$ if $G_H$ is a planar
supergraph of $G$ containing the Hamiltonian cycle $H$, $D_H$ is a
planar drawing of $G_H$, and $D$ is the restriction of $D_H$ to
$G$; note that $D_H$ witnesses that $D$ respects $H$. 

\iflong
The following basic observations will be useful when dealing with subhamiltonian graphs in Section~\ref{sec:oneconn}.
\begin{observation}\label{obs:subham}
  Let $G$ be a subhamiltonian graph with witness $(D,D_H,G_H,H)$. Then:
  \begin{enumerate}[(1)]
  \item Every subgraph of $G$ is also subhamiltonian.
  \item If $uv \in E(H)$, then the graph obtained from $G$ by
    adding a new vertex $x$ together with the edges $xu$ and $xv$ is
    subhamiltonian.
              \item If $uv \in E(H)$, then the graph obtained by contracting
    $uv$ is also subhamiltonian.
    \end{enumerate}
\end{observation}
\fi

\begin{figure}
  \centering
\begin{subfigure}{.4\textwidth}
\begin{tikzpicture}[xscale=1,yscale=1]
\coordinate (A) at (0,4);
\coordinate (B) at (-2,2);
\coordinate (C) at (-3,1);
\coordinate (D) at (-1,1);
\coordinate (E) at (-2,0);
\coordinate (F) at (0,2);
\coordinate (G) at (0,1);
\coordinate (H) at (0,0);
\coordinate (I) at (2.5,1);
\coordinate (L) at (2,0);
\coordinate (M) at (3,0);
\coordinate (N) at (2,-1.5);
\coordinate (O) at (3,-1.5);
\coordinate (P) at (-1,-4);
\coordinate (Q) at (0,-5);
\coordinate (R) at (1,-4);
\coordinate (S) at (2,-4);
\coordinate (T) at (0,-3);
\coordinate (U) at (2.5,-0.75);
\draw[very thick,color=green!50!black] (E)--(T)  (O)--(M)--(U)--(L) (S)--(Q) (L)--(N);
\draw[very thick,color=red]
(H)--(G)  (F)-- (A) (L)--(I) (N)--(T)  (R)--(T)--(Q) 
(N)--(O) (O)--(T)--(S) (E)--(B)--(D);
\draw[very thick,color=blue]
(H)--(T)--(P) (Q)--(R)--(S) (N)--(U)--(O) (G)--(F) (A)--(I)--(M)--(L) (A)--(B)--(C)--(E)--(D);
\draw[very thick, dashed, color=blue] (H)--(D) (G)--(N) (F)--(L) (O)--(S) (P)--(Q);
\draw[fill=black] 
(A) circle [radius=2pt]
(B) circle [radius=2pt]
(C) circle [radius=2pt]
(D) circle [radius=2pt]
(E) circle [radius=2pt]
(F) circle [radius=2pt]
(G) circle [radius=2pt]
(H) circle [radius=2pt]
(I) circle [radius=2pt]
(L) circle [radius=2pt]
(M) circle [radius=2pt]
(N) circle [radius=2pt]
(O) circle [radius=2pt]
(P) circle [radius=2pt]
(Q) circle [radius=2pt]
(R) circle [radius=2pt]
(S) circle [radius=2pt]
(T) circle [radius=2pt]
(U) circle [radius=2pt];
\node[right] at (A) {$A$};
\node[right] at (B) {$S$};
\node[right] at (C) {$R$};
\node[right] at (D) {$P$};
\node[below left] at (E) {$Q$};
\node[right] at (F) {$E$};
\node[right] at (G) {$F$};
\node[right] at (H) {$O$};
\node[right] at (I) {$B$};
\node[above right] at (L) {$D$};
\node[above right] at (M) {$C$};
\node[left] at (N) {$G$};
\node[above right] at (O) {$I$};
\node[left] at (P) {$M$};
\node[below] at (Q) {$L$};
\node[left] at (R) {$K$};
\node[above right] at (S) {$J$};
\node[left] at (T) {$N$};
\node[right] at (U) {$H$};
\end{tikzpicture}
\end{subfigure} \begin{subfigure}{.4\textwidth}
  \centering
\begin{tikzpicture}[xscale=0.5,yscale=0.5]
\coordinate (A) at (0,4);
\coordinate (I) at (0,3);
\coordinate (M) at (0,2);
\coordinate (L) at (0,1);
\coordinate (F) at (0,0);
\coordinate (G) at (0,-1);
\coordinate (N) at (0,-2);
\coordinate (U) at (0,-3);
\coordinate (O) at (0,-4);
\coordinate (S) at (0,-5);
\coordinate (R) at (0,-6);
\coordinate (Q) at (0,-7);
\coordinate (P) at (0,-8);
\coordinate (T) at (0,-9);
\coordinate (H) at (0,-10);
\coordinate (D) at (0,-11);
\coordinate (E) at (0,-12);
\coordinate (C) at (0,-13);
\coordinate (B) at (0,-14);
\draw[very thick, color=red] 
(H) to[out=180,in=180] (G)
(F) to[out=180,in=180] (A)(E) to[out=180,in=180] (B) to[out=180,in=180] (D) to[out=180,in=180] (E) 
(L) to[out=180,in=180] (I) 
(N) to[out=180,in=180] (T) (N) to[out=180,in=180] (O) (T) to[out=180,in=180] (Q) 
(R) to[out=180,in=180] (T) to[out=180,in=180] (S);
\draw[very thick, color=green!50!black]
(O) to[out=0,in=0] (M) to[out=0,in=0] (U) to[out=0,in=0] (L)
(L) to[out=0,in=0] (N)
(E) to[out=0,in=0] (T)
(S) to[out=0,in=0] (Q);
\draw[very thick, dashed,color=blue] (H) to[out=180,in=180] (D) (G) to[out=180,in=180] (N) (F) to[out=180,in=180] (L) (O) to[out=180,in=180] (S) (P) to[out=180,in=180] (Q);
\draw[very thick,color=blue]
(H) to[out=180,in=180] (T) to[out=180,in=180] (P) (Q) to[out=180,in=180] (R) to[out=180,in=180] (S) (N) to[out=180,in=180] (U) to[out=180,in=180] (O) (G) to[out=180,in=180] (F) (A) to[out=180,in=180] (I) to[out=180,in=180] (M) to[out=180,in=180] (L) (A) to[out=180,in=180] (B) to[out=180,in=180] (C) to[out=180,in=180] (E) to[out=180,in=180] (D);
\draw[fill=black] 
(A) circle [radius=2pt]
(B) circle [radius=2pt]
(C) circle [radius=2pt]
(D) circle [radius=2pt]
(E) circle [radius=2pt]
(F) circle [radius=2pt]
(G) circle [radius=2pt]
(H) circle [radius=2pt]
(I) circle [radius=2pt]
(L) circle [radius=2pt]
(M) circle [radius=2pt]
(N) circle [radius=2pt]
(O) circle [radius=2pt]
(P) circle [radius=2pt]
(Q) circle [radius=2pt]
(R) circle [radius=2pt]
(S) circle [radius=2pt]
(T) circle [radius=2pt]
(U) circle [radius=2pt];
\node[right] at (A) {$A$};
\node[right] at (B) {$S$};
\node[right] at (C) {$R$};
\node[right] at (D) {$P$};
\node[right] at (E) {$Q$};
\node[right] at (F) {$E$};
\node[right] at (G) {$F$};
\node[right] at (H) {$O$};
\node[right] at (I) {$B$};
\node[right] at (L) {$D$};
\node[right] at (M) {$C$};
\node[right] at (N) {$G$};
\node[right] at (O) {$I$};
\node[right] at (P) {$M$};
\node[right] at (Q) {$L$};
\node[right] at (R) {$K$};
\node[right] at (S) {$J$};
\node[right] at (T) {$N$};
\node[right] at (U) {$H$};
\end{tikzpicture}
\end{subfigure}
\caption{A drawing of a subhamiltonian graph $G$, made of the
  full-edges, which is completed by the dashed edges to one of its
  Hamiltonian supergraphs $G_H$ (left) and the same graph drawn as a
  two-page book embedding (right). In both drawings the Hamiltonian
  cycle $H$ is colored in blue and the edges belonging to page 1 and 2
  are colored with green and red, respectively.\iflong Note that the
    partition of the edges into the pages can be obtained from a
    planar drawing of $G_H$ by partitioning the edges according to the
    two regions given by $H$.\fi}
\label{fig:2page}
\end{figure}
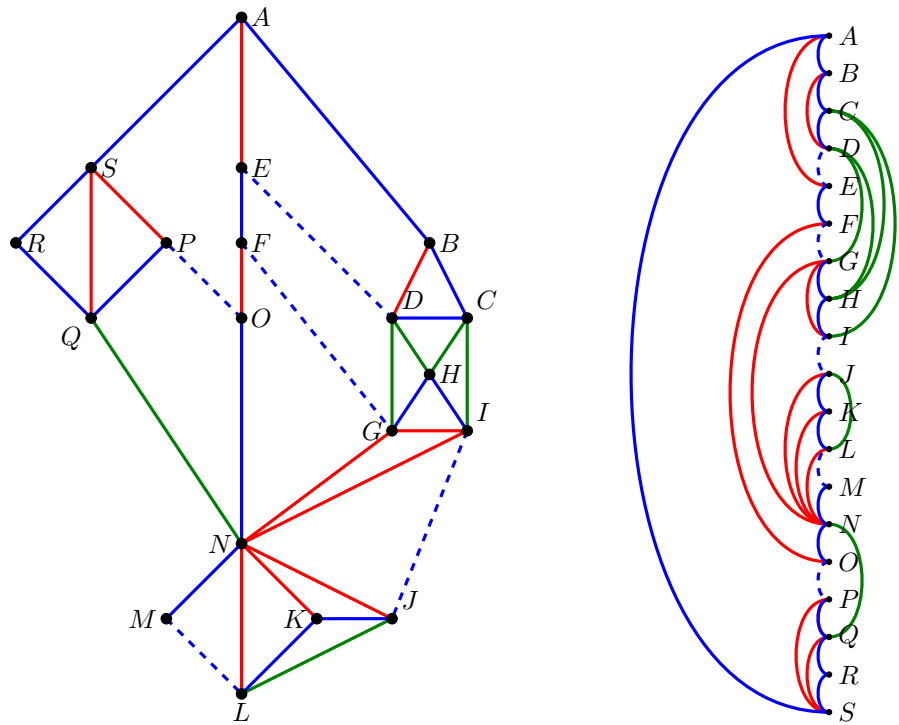

\subparagraph{\SPQR-Trees.}
\ifshort
We assume familiarity with the \SPQR-tree data structure for
biconnected multi-graphs which decomposes a graph into (S)eries,
(P)arallel, (R)igid and (Q) nodes (leaf nodes and root node), following the formalism used by Gutwenger et al.~\cite{GutwengerMW05}, see also~\cite{BattistaT89,BienstockM89,BienstockM90}. For a node $b$ in an \SPQR-tree, we use $\sk(b)$ and $\pe(B)$ to denote the \emph{skeleton} and \emph{pertinent graph} of $b$, respectively. \SPQR-trees can be computed in linear time, and an illustration of the data structure is provided in Figure~\ref{fig:spqr}. 
\fi
\iflong
We give a brief introduction to the \SPQR-tree data structure for biconnected multi-graphs, following the formalism used by Gutwenger et al.~\cite{GutwengerMW05}. 
 Let $G=(V,E)$ be a biconnected multi-graph and $a,b\in V$. We can partition $E$ into equivalence classes $E_1,\ldots,E_k$ in the following way: for any two edges $e,e'\in E$, $e$ and $e'$ belong to the same equivalence class if and only if there exists a path $P$ in $G$ which contains both $e$ and $e'$ as edges and no internal vertex of $P$ is in $\{a,b\}$. The classes $E_i$ are called the {\it separation classes} of $G$ with respect to $\{a,b\}$ and if $k\geq 2$ then $\{a,b\}$ is called a {\it separation pair} unless $(i)$ $k=2$ and one of the separation classes only contains a single edge, or $(ii)$ $k=3$ and each separation class is made of a single edge.
A biconnected multi-graph without a separation pair is called {\it triconnected}. A \emph{split pair} is a pair of vertices which are either adjacent to each other, or form a separation pair.

\SPQR-trees were introduced by Di Battista and Tamassia~\cite{BattistaT89}, based on the ideas of Bienstock and Monma~\cite{BienstockM89,BienstockM90}, and since then have been used in various graph drawing applications, for a survey we refer to the work of Mutzel~\cite{Mutzel03}.

\SPQR-trees represent the decomposition of a biconnected multi-graph $G$ based on split pairs and their ``split components''. 
A {\it split component} of a split pair $\{u,v\}$ is either the edge $(u,v)$ or a maximal subgraph $C$ of $G$ such that $\{u,v\}$ is not a split pair of $C$. Let $\{s,t\}$ be a split pair of $G$. A {\it maximal split pair} $\{u,v\}$ of $G$ with respect to $\{s,t\}$ is such that, for any other split pair $\{u',v'\}$, vertices $u$, $v$, $s$ and $t$ are in the same split component. 

Let $e=(s,t)$ be an edge of $G$, called the {\it reference edge}. The \SPQR-tree $\mathcal{B}$ of $G$ with respect to $e$ is a rooted ordered tree whose nodes are of four types: $S$, $P$, $Q$, and $R$. Each node $b$ of $\mathcal{B}$ has an associated biconnected multi-graph $\sk(b)$, called the {\it skeleton} of $b$. The tree $\mathcal{B}$ is recursively defined as follows:

\begin{itemize}
\item {\it Trivial Case}. If $G$ consists of exactly two parallel edges between $s$ and $t$, then $\mathcal{B}$ consists of a single $Q$-node whose skeleton is $G$ itself.
\item {\it Parallel Case}. If the split pair $\{s,t\}$ has $k$ split components $G_1,\ldots, G_k$ with $k\geq 3$, the root of $\mathcal{B}$ is a $P$-node $b$, whose skeleton consists of $k$ parallel edges $e=e_1,\ldots, e_k$ between $s$ and $t$.
\item {\it Series Case}: Otherwise, the split pair $\{s,t\}$ has exactly two split components, one of them is $e$, and the other one is denoted with $G'$. If $G'$ has cutvertices $c_1,\ldots,c_{k-1}$ ($k\geq 2$) that partition $G$ into its blocks $G_1,\ldots,G_k$, in this order from $s$ to $t$, the root of $\mathcal{B}$ is an S-node $b$, whose skeleton is the cycle $e_0, e_1,\ldots, e_k$, where $e_0=e$, $c_0=s$, $c_k=t$, and $e_i=(c_{i-1},c_1$) ($i=1,\ldots, k$).
\item {\it Rigid Case}: If none of the above cases applies, let $\{s_1,t_1\},\ldots, \{s_k,t_k\}$ be the maximal split pairs of $G$ with respect to $\{s,t\}$ ($k\geq 1$), and, for $i=1,\ldots ,k$, let $G_i$ be the union of all the split components of $\{s_i,t_i\}$ but the one containing $e$. The root of $\mathcal{B}$ is an R-node, whose skeleton is obtained from $G$ by replacing each subgraph $G_i$ with the edge $e_i=(s_i,t_i)$.
\end{itemize}

Except for the trivial case, $b$ has children $b_1,\ldots, b_k$, such that $b_i$ is the root of the \SPQR-tree of $G_i\cup e_i$ with respect to $e_i$ ($i=1,\ldots,k$). The endpoints of the edge $e_i$ are called the {\it poles} of node $b_i$. 
Edge $e_i$ is said to be the {\it virtual edge} of node $b_i$ in the skeleton of $b$ and of node $b$ in the skeleton of $b_i$. 
We call node $b$ the {\it pertinent node} of $e_i$ in the skeleton of $b_i$, and $b_i$ the {\it pertinent node} of $e_i$ in the skeleton of $b$. 
The virtual edge of $b$ in the skeleton of $b_i$ is called the reference edge of $b_i$.

Let $b_r$ be the root of $\mathcal{B}$ in the decomposition given above. 
We add a $Q$-node representing the reference edge $e$ and make it the parent of $b_r$ so that it becomes the new root.

Let $e$ be an edge in $\sk(b)$ and let $b'$ be the pertinent node of $e$. Deleting edge $\{b,b'\}$ in $\mathcal{B}$ splits $\mathcal{B}$ into two connected components. Let $\mathcal{B}_{b'}$ be the connected component containing $b'$. The {\it expansion graph} of $e$ (denoted with $\expa(e)$) is the graph induced by the edges of $G$ contained in the skeletons of the $Q$-nodes in $\mathcal{B}_{b'}$. We further introduce the notation $\expa^+(e)$ for the graph $\expa(e)\cup e$. 
The {\it pertinent graph} $\pe(b)$ of a tree node $b$ is obtained from $\sk(b)$ minus the reference edge by replacing each skeleton edge with its expansion graph. An illustration of an \SPQR-tree is provided in Figure~\ref{fig:spqr}.
\fi

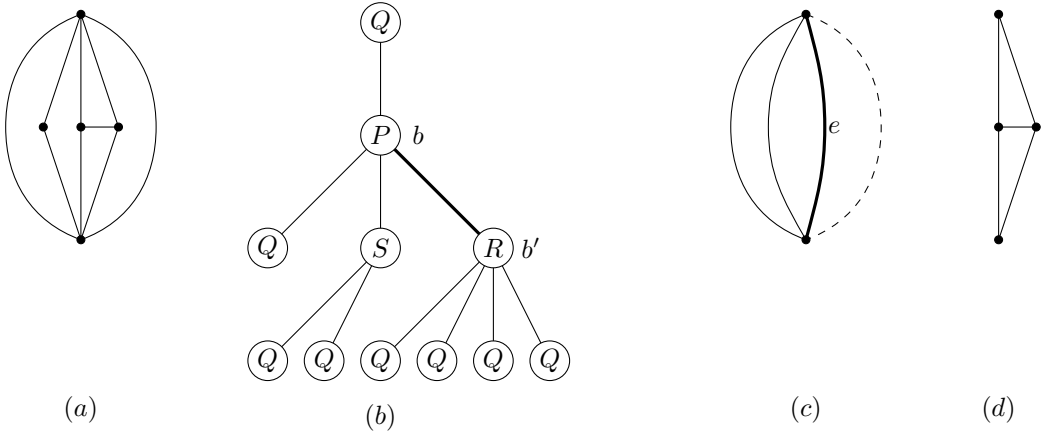
\begin{figure}
  \centering
\begin{minipage}[c]{0.15\textwidth}
\begin{tikzpicture}[xscale=0.5, yscale=0.5]
\draw 
(0,3) to[out=340,in=90] (2,0) to[out=270,in=20] (0,-3)
(0,3) to[out=200,in=90] (-2,0) to[out=270,in=160] (0,-3)

(0,3)--(-1,0)--(0,-3)--(0,0)--(0,3)--(1,0)--(0,-3) (0,0)--(1,0);
 (0,3) to[out=340,in=90] (2,0) to[out=270,in=20] (0,-3);
\draw[fill=black] (-1,0) circle [radius=3pt] (1,0) circle [radius=3pt] (0,0) circle [radius=3pt] 
(0,3) circle [radius=3pt] (0,-3) circle [radius=3pt];
\node at (0,-7.5) {$(a)$};
\end{tikzpicture}
\end{minipage}
\qquad
\begin{minipage}[c]{0.35\textwidth}
\begin{tikzpicture}[xscale=0.5, yscale=0.5]
\draw (0,4.5)--(0,1.5)--(-3,-1.5) (0,1.5)--(3,-1.5)
(0,-4.5)--(3,-1.5)--(1.5,-4.5)(3,-4.5)--(3,-1.5)--(4.5,-4.5)(-3,-4.5)--(0,-1.5)--(-1.5,-4.5)(0,1.5)--(0,-1.5);
\draw[very thick] (0,1.5)--(3,-1.5);
\draw[fill=white] (0,4.5) circle [radius=15pt](0,1.5) circle [radius=15pt](-3,-1.5) circle [radius=15pt](0,-1.5) circle [radius=15pt](3,-1.5) circle [radius=15pt](-3,-4.5) circle [radius=15pt](-1.5,-4.5) circle [radius=15pt](0,-4.5) circle [radius=15pt](1.5,-4.5) circle [radius=15pt](3,-4.5) circle [radius=15pt](4.5,-4.5) circle [radius=15pt];
\node at (0,4.5) {$Q$};\node at (-3,-4.5) {$Q$};\node at (-1.5,-4.5) {$Q$};\node at (0,-4.5) {$Q$};\node at (1.5,-4.5) {$Q$};\node at (3,-4.5) {$Q$};\node at (4.5,-4.5) {$Q$};\node at (0,1.5) {$P$};\node at (-3,-1.5) {$Q$};\node at (0,-1.5) {$S$};\node at (3,-1.5) {$R$};\node at (1,1.5) {$b$};\node at (4,-1.5) {$b'$};
\node at (0,-6) {$(b)$};
\end{tikzpicture}
\end{minipage}
\qquad
\begin{minipage}[c]{0.15\textwidth}
\begin{tikzpicture}[xscale=0.5, yscale=0.5]

\draw[dashed]
(0,3) to[out=340,in=90] (2,0) to[out=270,in=20] (0,-3)
;
\draw[] 
(0,3) to[out=200,in=90] (-2,0) to[out=270,in=160] (0,-3);
\draw[very thick] (0,3) to[out=285,in=90] (0.5,0) to[out=270,in=75] (0,-3);
\draw[fill=black]
(0,3) circle [radius=3pt] (0,-3) circle [radius=3pt]
;
\draw[]
(0,3) to[out=240,in=90] (-1,0) to[out=270,in=120] (0,-3)
;

\node at (0.75,0) {$e$};
\node at (0,-7.5) {$(c)$};
\end{tikzpicture}
\end{minipage}
\qquad
\begin{minipage}[c]{0.1\textwidth}
\begin{tikzpicture}[xscale=0.5, yscale=0.5]
\draw 
(0,-3)--(0,0)--(0,3)--(1,0)--(0,-3) (0,0)--(1,0);

(0,3) to[out=340,in=90] (2,0) to[out=270,in=20] (0,-3);
\draw[fill=black]  (1,0) circle [radius=3pt] (0,0) circle [radius=3pt] 
(0,3) circle [radius=3pt] (0,-3) circle [radius=3pt];
\node at (0,-7.5) {$(d)$};
\end{tikzpicture}
\end{minipage}
\caption{$(a)$ shows a biconnected multi-graph $G$. $(b)$ shows the
  \SPQR{}-tree $\mathcal{B}$ of $G$. $(c)$ shows the skeleton of $b$,
  $\sk(b)$, where the edge $e$ that corresponds to the child (with
  pertinent node) $b'$ is
  in bold and the dashed edge represents the reference edge. 
        Finally, $(d)$ shows $\pe(b')$.}
\label{fig:spqr}  
\end{figure}

\iflong
\SPQR-trees can be computed efficiently, and this also implicitly
bounds the their size.

 \begin{lemma}[\cite{GutwengerM00}] 
\label{lem:computeSPQR}
  Let $G$ be biconnected multi-graph with $n$ vertices and $m$ edges.
  An \SPQR{}-tree of $G$ with $\bigoh(m)$ nodes and edges inside skeletons
  can be constructed in $\bigoh(n+m)$ time.
\end{lemma}

Choosing a different reference edge $e'$ is equivalent to rooting the tree $\mathcal{B}$ at the $Q$-node whose skeleton contains $e'$. In particular, the unrooted version of the \SPQR-tree of a biconnected multi-graph (including the skeleton graphs) is unique.

We will later also need the following well-known fact about
\SPQR{}-trees, which will need to define the types of nodes in an \SPQR{}-tree.
\begin{fact}\label{fact:noose-for-SPQRtree-node}
  Let $G$ be a biconnected planar multi-graph with planar
  drawing $D$ and let $\mathcal{B}$ be the \SPQR-tree of $G$. Then,
  for every node $b$ of $\mathcal{B}$ with reference edge $(s_b,t_b)$,
  there is a noose $N_b$ such that:
  \begin{itemize}
  \item $N_b$ intersects with $D$ only at $s_b$ and $t_b$.
  \item $N_b$ separates $\pe(b)$ from $G\setminus \pe(b)$ in $D$.
  \end{itemize}
  Moreover,
        if $N_b$ and $N_{b'}$ for
  two nodes $b$ and $b'$ of $\BBB$ have the same reference edge $(s,t)$ and
  contain a subcurve between $s$ and $t$ in the same face of $D$, then we
  can and will assume that the two subcurves are identical.
\end{fact}
\fi
 
\subparagraph{Sphere-Cut Decompositions.}
A branch decomposition $\langle T, \lambda \rangle $ of a graph $G$ consists of an
unrooted ternary tree $T$ (meaning that each node of $T$ has degree one or three) and of a bijection $\lambda:\mathcal{L}(T)\leftrightarrow E(G)$ from
the leaf set $\mathcal{L}(T)$ of $T$ to the edge set $E(G)$ of $G$; to
distinguish $E(T)$ from $E(G)$, we call the elements of the former
\emph{arcs} (as was also done in previous work~\cite{DornPBF10}). For each arc $a$ of $T$, let $T_1$ and $T_2$ be the two connected
components of $T - a$, and, for $i = 1, 2$, let $G_i$ be the subgraph of $G$ that consists of the edges corresponding
to the leaves of $T_i$, i.e., the edge set $\{\lambda(\mu) : \mu \in \mathcal{L}(T) \cap V (T_i)\}$. The middle set $\md{a} \subseteq V (G)$ is the
intersection of the vertex sets of $G_1$ and $G_2$, i.e., $\md{a} := V (G_1) \cap V (G_2)$. The width $\beta(\langle T, \lambda \rangle )$ of $\langle T, \lambda \rangle$
is the maximum size of the middle sets over all arcs of $T$, i.e.,
$\beta(\langle T, \lambda \rangle ):= max\{| \md{a}| : a \in E(T)
\}$. An
optimal branch decomposition of $G$ is a branch decomposition with minimum width; this width is called
the branchwidth $\beta(G)$ of $G$.
We will need the following well-known relation between treewidth and
branchwidth.
\begin{lemma}[{\cite[Theorem 5.1]{DBLP:journals/jct/RobertsonS91}}]
\label{lem:bw-tw}
  Let $G$ be a graph. Then, $\bw(G)-1 \leq \tw(G)\leq \frac{3}{2}\bw(G)-1$,
  where $\bw(G)$ is the branchwidth and $\tw(G)$ is the treewidth of $G$.
\end{lemma}

Let $D$ be a plane drawing of a connected planar graph $G$. A noose of $D$ is a closed simple curve that (i) intersects $D$ only at vertices
and (ii) traverses each face at most once, i.e., its intersection with
the region of each face forms a connected curve. The length of a noose
is the number of vertices it intersects, and every noose $O$ separates
the plane into two regions $\delta_1$ and $\delta_2$. A
\emph{sphere-cut decomposition} $\langle T, \lambda, \Pi = \SB\pi_a\SM
a\in E(T)\SE \rangle $ of $(G,D)$ is a branch decomposition$\langle T, \lambda \rangle$ of $G$
together with a set $\Pi$ of circular orders $\pi_a$ of
$\md{a}$---one for each arc $a$ of $T$---such that there exists a noose
$O_a$ whose closed discs $\delta_1$ and $\delta_2$ enclose the drawing of $G_1$ and of $G_2$, respectively. Observe that $O_a$ intersect $G$ exactly at $\md{a}$ and its length is $| \md{a}|$. 
Note that the fact that $G$ is connected together with Conditions~(i) and~(ii) of the
definition of a noose implies that the graphs $G_1$ and $G_2$ are both connected and that the set of nooses forms a
laminar set family, that is, any two nooses are either disjoint or nested. 
A clockwise traversal of $O_a$ in
the drawing of $G$ defines the cyclic ordering $\pi_a$ of $\md{a}$. We always assume that the vertices of every
middle set $\md{a}$ are enumerated according to $\pi_a$. 
A sphere-cut decomposition of a given planar graph with $n$ vertices can be constructed in $\bigoh(n^3)$ time \cite{DornPBF10}.

\iflong
\begin{lemma}[{\cite[Theorem 1]{DornPBF10}}] 
\label{lem:comp-spcut}
  Let $G$ be a biconnected planar multi-graph on $n$ vertices and
  branchwidth $\omega$. Then, a sphere-cut
  decomposition of $G$ of width $\omega$ can be computed in time $\bigoh(n^3)$.   
\end{lemma}
Note that~\cite[Theorem 1]{DornPBF10} requires that $G$ has not
vertices of degree at most one, which is the case for biconnected multi-graphs.

We will only consider sphere-cut decompositions of $\sk(b)$ for
some R-node or S-node $b$ in an \SPQR{}-tree, which implies that the
underlying graph will admit a unique planar embedding. Due to this
fact, we sometimes abuse the notation by
referring to sphere-cut decompositions as purely combinatorial objects
(i.e., without an explicit drawing of the individual nooses).
Suppose that
$b$ is an $R$-node in some \SPQR{}-tree and let $\langle
T_b, \lambda_b, \Pi_b \rangle$ be a sphere-cut decomposition for graph
$\sk(b)$ with the reference edge $\{(s_{b},t_{b})\}$.
Let ${\lambda_b}^{(-1)}((s_{b},t_{b}))$ be the root of $T_b$.
Each arc $a$ of $T_b$ is associated with the subgraph $\sk(b,a)$ of
$\sk(b)\cup \{(s_{b},t_{b})\}$ in the inside region, i.e., the region not containing the reference edge, of the noose $O_a$ of $a$.
The {\it pertinent graph} $\pe(b,a)$ of an R-node $b$ is obtained from
$\sk(b,a)$ by replacing each skeleton edge with its expansion graph.

Every noose $O_a$ can be divided into subcurves by splitting the noose
at the vertices in $\md{a}$. Each such subcurve can be characterized by
a pair $(\{u,v\},f)$, where $u,v \in \md{a}$ are two consecutive nodes
in $\pi_a$ and $f$ is a face of $\sk(b)\cup \{(s_{b},t_{b})\}$.   Due to the properties of sphere-cut decompositions, we can assume that
whenever two nooses contain two subcurves that are characterized by
the same pair, then the subcurves are identical.
For convenience, we can identify any noose of the sphere-cut
decomposition with the set of subcurves that it contains, e.g.,
we often view the noose $O_a$ as the set of pairs $(\{u,v\},f)$ that
correspond to the subcurves contained in $O_a$.

Below, we note that the notion of nooses defined above can also be
assumed to be well-behaved when dealing with sphere-cut decompositions
of an R-node or an S-node in an \SPQR-tree of $G$.

 \begin{observation}\label{observation:rnode_drawing}
  Let $G$ be a biconnected planar multi-graph with planar
  drawing $D$, let $\mathcal{B}$ be the \SPQR-tree of $G$ and let $b$
  be an R-node or an S-node of $\mathcal{B}$ with sphere-cut decomposition
  $\langle T_b,\lambda_b, \Pi_b \rangle$. Then, $D$ can be extended to a planar
  drawing of $G$ together with the nooses $\{O_a~|~a \in E(T_b)\}$ where each of the nooses lies inside $N_b$.
   \end{observation}
\fi

We say that a biconnected planar multi-graph $G$ equipped
with an \SPQR{}-tree $\mathcal{B}$ is \emph{associated} with a set
$\mathcal{T}$ of sphere-cut decompositions if $\mathcal{T}$ contains
a sphere-cut decomposition of $\sk(b)$ for every R-node and every
S-node $b$ of $\mathcal{B}$.
\iflong
The following lemma now follows immediately
from~\Cref{fact:noose-for-SPQRtree-node}
and~\Cref{observation:rnode_drawing}.
\fi

\begin{lemma}
\label{cor:add-nooses-to-drawing}
  Let $G$ be biconnected planar multi-graph with planar drawing $D$ and
  \SPQR{}-tree $\mathcal{B}$ of $G$ together with the associated set
  $\mathcal{T}$ of sphere-cut decompositions. Then, $D$ can be
  extended to a planar drawing $D'$ of $G$ together with all nooses in
  $\SB O_a \SM a \in E(T_b) \land \langle T_b,\lambda_b,\Pi_b\rangle
  \in \mathcal{T}\SE$ as well as a noose $N_b$ for every node $b$ of
  $\mathcal{B}$ satisfying:
  \begin{itemize}
  \item $N_b$ intersects with $D$ only at $s_b$ and $t_b$.
  \item $N_b$ separates $\pe(b)$ from $G\setminus \pe(b)$ in $D$.
  \end{itemize}
  Moreover, if any of the subcurves of the nooses $O_a$ and the nooses
  $N_b$ connect the same two vertices in the same face of $D$, then
  the two subcurves are identical in $D'$.
\end{lemma}

\subparagraph{Non-Crossing Matchings.}
\iflong
We will use non-crossing matchings and the closely related Dyck words
for the definition and analysis of our types.
Let $K_n$ be the complete graph on vertices $\{1,\dotsc,n\}$ and let
$<$ be a cyclic ordering of the elements in $\{1,\dotsc,n\}$.
A \emph{non-crossing matching} is a matching $M$ in the graph $K_n$
such that for every two edges 
$\{a,b\},\{c,d\} \in M$
it is
not the case that $a < c < b < d$. A non-crossing matching can be
visualized by placing $n$ vertices on a cycle, and connecting matching
vertices with pairwise non-crossing curves all on one fixed side of
the cycle. The number of non-crossing matchings over $n$ vertices is
given by~\cite{Noncrossingpartoriginal,Noncrossingpart}:

$$M(n) = CN(\frac{n}{2}) \approx \frac{2^n}{\sqrt{\pi} (\frac{n}{2})^{\frac{3}{2}}} \approx 2^n$$

Here, $CN(n)$ is the $n$-th Catalan number, i.e.:

$$CN(n) = \frac{1}{n+1} \binom{2n}{n} \approx \frac{4^n}{\sqrt{\pi}
  n^{\frac{3}{2}}} \approx 4^n$$

A \textit{Dyck word} is a sequence composed of $\{"[","]"\}$ symbols, such that each prefix has an equal or greater number of $"["$s than $"]"$s, and the total number of $"["$s and $"]"$s are equal.

\begin{observation}\label{obs:dyck}
  There is a one-to-one correspondence between non-crossing matchings
  with $2n$ vertices and Dyck words of length $2n$. Moreover, one can
  be translated into the other after fixing a starting vertex and an
  orientation of the cycle.
\end{observation}
\fi \ifshort
Let $K_n$ be the complete graph on vertices $\{1,\dotsc,n\}$ and let
$<$ be a cyclic ordering of the elements in $\{1,\dotsc,n\}$.
A \emph{non-crossing matching} is a matching $M$ in the graph $K_n$
such that for every two edges 
$\{a,b\},\{c,d\} \in M$
it is
not the case that $a < c < b < d$.
\fi

\section{Solution Normal Form}

Our first order of business is to show that we can
assume that the solution (Hamiltonian cycle) to the \SH{} problem interacts with the drawing
in a restricted manner. In particular, we aim to show that every
subhamiltonian graph $G$ has a witness $(D,D_H,G_H,H)$ in \emph{normal
  form}, i.e., with the
following property: it is possible to draw a curve in $D_H$
between any two vertices occurring in a common face of $D$ such that
this curve only crosses the Hamiltonian cycle at most twice.
Note that this property will allow us to bound the number of possible
interactions of the Hamiltonian cycle with any subgraph corresponding
to either a node in the SQPR-tree or an arc in a sphere-cut
decomposition and is crucial to bound the number of types in our
dynamic programming algorithm.
\ifshort The following lemma is the main technical lemma behind our
  normal form. An illustration of the main ideas behind the proof is
  provided in  Figure~\ref{fig:hc-two}.
\fi  

\iflong We will need the following auxiliary lemmas.
  \begin{lemma}
\label{obs:cycle}
  Let $G$ be a subhamiltonian graph with witness $(D,D_H,G_H,H)$ and
  let $f$ be a face of $D_H$. Then the restriction of the cyclic order
  given by $H$ to the vertices in $V(f)$ is either equal to $\seqb(f)$
  or it is equal to the reverse of
  $\seqb(f)$.
\end{lemma}
\iflong \begin{proof}
           Suppose for a contradiction that the cyclic order of
  $\seqb(f)$ differs from the (reverse) cyclic order given by $H$. Then, there
  are three vertices $a$, $b$, and $c$ such that $b$ is
  between $a$ and $c$ in the cyclic order given by $\seqb(f)$, but
  $b$ is between $d$ and $e$ in the cyclic order given by $H$,
  where $\{d,e\}\neq \{a,c\}$. W.l.o.g. assume that $d\neq a$. Then,
  $H$ contains a path $P_{db}$ between $d$ and $b$ that does not
  contain any vertex from $V(f)\setminus \{d,b\}$ and moreover $d$ is
  neither between $a$ and $b$ nor between $b$ and $c$ in the cyclic
  order given by $\seqb(f)$. Let $A$ be the set of all vertices
  between $d$ and $b$ in the cyclic order defined by $\seqb(f)$. Then,
  because $H$ is a Hamiltonian cycle and $V(f)\setminus A\neq
  \emptyset$, we obtain that there is a vertex $x \in A$ and a vertex
  $y \in V(f)\setminus A$ such that $H$ contains a path $P_{xy}$ that
  does not contain any vertex in $V(f)\setminus \{x,y\}$. Since
  $P_{db}$ and $P_{xy}$ are disjoint, the statement of the lemma now
  follows from Observation~\ref{obs:face-nopl}.
\end{proof}\fi 

\begin{lemma} 
\label{lem:edge-given}
  Let $G$ be a subhamiltonian graph with witness $(D,D_H,G_H,H)$ and let $f$
  be a face of $D$. For any two vertices $u,v \in V(f)$,
  a $uv$-curve $c$ can be added to $D_H$ inside $f$ such that every
  edge from $E(H)$ crosses at most once with $c$.
\end{lemma}
\iflong \begin{proof}
  Let $D_H'$ be obtained from the restriction of $D_H$ to vertices and edges inside
  $f$ and let $f_u$ and $f_v$ be the two faces of $D_H'$ inside $f$ having $u$
  or $v$ on their border, respectively. If $f_u=f_v$, then the claim
  follows immediately from Observation~\ref{obs:face}. Otherwise,
  consider the dual graph $H$ of $D_H'$ together with its drawing $D_H^D$ inside $D_H'$.
  Then, $H$ contains a path 
  from $f_u$ to $f_v$ that uses only faces inside $f$ and that corresponds to
  a curve $P$ between $f_u$ and $f_v$ in $D_H^D$  that intersects every
  edge of $H$ at most once. Because of Observation~\ref{obs:face},
  we can draw a curve $c_u$ from $u$
  to $f_u$ and a curve $c_v$ from $f_v$ to $v$ inside $f_u$ and $f_v$, respectively,
  without adding any crossings. Then, the curve obtained from the
  concatination of $c_u$, $P$, and $c_v$ is the required $uv$-curve.
\end{proof}\fi

\begin{lemma} 
\label{lem:edge}
  Let $G$ be a subhamiltonian graph with witness $(D,D_H,G_H,H)$, let $f$
  be a face of $D$, and let $c$ be a curve drawn inside $f$ between
  two vertices $u,v \in V(f)$. Then, we can redraw the curves
  corresponding to the edges of $H$ inside $f$ such that every such curve
  crosses $c$ at most once, i.e., we can adapt $D_H$ inside $f$ into a
  drawing $D_H'$ such that $(D,D_H',G_H,H)$ is a witness for $G$ and
  every curve corresponding to an edge of $H$ inside $f$ crosses $c$
  at most once in $D_H'$.
\end{lemma}
\iflong \begin{proof}
  Because of Lemma~\ref{lem:edge-given} there is a $uv$-curve $c'$
  that can be added to $D_H$ inside $f$ such that every curve corresponding to a
  Hamiltonian cycle of $H$ crosses $c'$ at most once. Let
  $(p_1,\dotsc,p_\ell)$ be the sequence of all crossing points between $c'$ and curves
  corresponding to edges of $H$ given in the order of appearance when going along $c'$
  from $u$ to $v$ and suppose that $p_i$ is the crossing point of the
  edge $e_i$ on $H$ with $c'$.

  Now consider the drawing $D_H^-$
  obtained from $D_H$ after adding $c$ and removing all curves
  corresponding to edges of $H$ inside $f$. Moreover, let
  $p_1^c,\dotsc,p_\ell^c$ be an arbitrary set of pairwise distinct
  points on $c$ that occur in the order $(p_1^c,\dotsc,p_\ell^c)$ when
  going along $c$ from $u$ to $v$ and let $c^1$ and $c^2$ be the two
  subcurves in $D_H^-$ of the border of $f$ between $u$ and $v$.
  Note that every edge $e_i$ has one endpoint $v_i^1$ on $c^1$ and one
  endpoint $v_i^2$ on $c^2$; otherwise both endpoints of $e_i$ are
  either on $c^1$ or on $c^2$ and $c'$ could have been drawn in $D_H$
  without crossing the curve corresponding to $e_i$. Note furthermore
  that because of Observation~\ref{obs:face-nopl}, the vertices $v_1^j,\dotsc,v_\ell^j$ must appear in the order
  $(v_1^j,\dotsc,v_\ell^j)$, when going along $c^j$ from $u$ to $v$
  for every $j \in \{1,2\}$. Since $v_i^1$ ($v_i^2$) and $p_i$ are initially in
  the same face of $D_H^-$, we can use Observation~\ref{obs:face} to draw a curve between $v_i^1$ ($v_i^2$)
  and $p_i$ in this face for every $i \in [1,\ell]$. Moreover, using the same observation, we
  obtain that after drawing this curve, it still holds that $v_i^1$
  ($v_i^2$) are in the same face as $p_i$ for every $i \in [1,\ell]$. Therefore,
  we can repeatedly apply
  Observation~\ref{obs:face} to draw curves in $D_H^-$ between
  $v_i^j$ and $p_i'$ for every $i \in [1,\ell]$ and $j \in \{1,2\}$ to
  obtain the required drawing $D_H'$.   
                                                                        \end{proof}\fi 
 \fi 

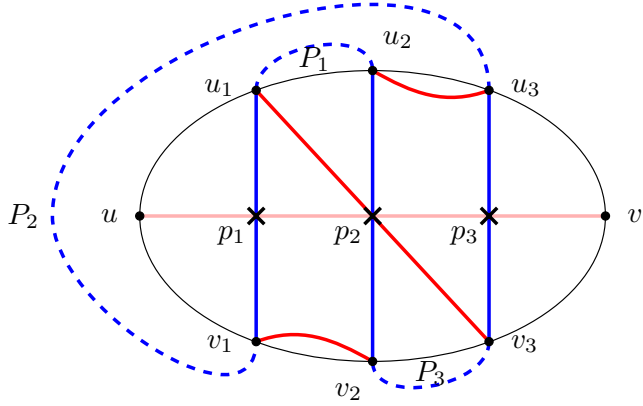
\begin{figure}
  \centering
  \resizebox{!}{7cm}{\begin{tikzpicture}[scale=0.7]
      \draw[fill=white] (0,0) ellipse (4cm and 2.5cm);
      \draw[very thick,color=red!30!white] (-4,0)--(4,0); 
      \draw[very thick,blue](-2,2.16)--(-2,-2.16)(0,2.5)--(0,-2.5)(2,2.16)--(2,-2.16);(-4,3) to[out=0,in=90] (-2,2.16) (-2,2.16)--(-2,0)(0,0)--(0,2.5)(4,-3) to[out=180,in=270] (2,-2.16)(2,0)--(2,-2.16);
      \draw[very thick,red](0,2.5) to[out=-30,in=200] (2,2.16) (0,-2.5) to[out=150,in=20] (-2,-2.16)(-2,2.16)--(2,-2.16);
      \draw[very thick, dashed,blue](-2,2.16)to[out=90,in=90] (0,2.5)(2,-2.16)to[out=270,in=270](0,-2.5)(2,2.16)to[out=90,in=90] (-5.5,0)to[out=270,in=270] (-2,-2.16);
      \draw[fill=black]
      (-4,0)circle[radius=2pt](4,0)circle[radius=2pt](-2,-2.16)circle
      [radius=2pt](2,-2.16)circle[radius=2pt](-2,2.16)circle
      [radius=2pt](2,2.16)circle [radius=2pt](0,2.5)circle
      [radius=2pt](0,-2.5)circle [radius=2pt](-2,0)node[very
      thick,cross=4pt]{} (0,0)node[very thick,cross=4pt] {} (2,0)
      node[very thick,cross=4pt] {};
      \node[left] at (-2.2,2.2) {$u_1$};\node[above right] at (0,2.7)
      {$u_2$};\node[right] at (2.2,2.2) {$u_3$};\node[below left] at (-2,0){$p_1$};
      \node[below left] at (0,0) {$p_2$};\node[below left] at (2,0) {$p_3$};\node[left] at (-2.2,-2.2) {$v_1$};\node[below left] at (0,-2.7) {$v_2$};\node[right] at (2.2,-2.2) {$v_3$};\node[left] at (-4.2,0) {$u$};\node[right] at (4.2,0) {$v$};\node at (-1,2.7) {$P_1$};\node at (-6,0) {$P_2$};\node at (1,-2.7) {$P_3$};\end{tikzpicture}}
  \caption{The cycle $H = (u_2,P_1,u_1,v_1,P_2,u_3,v_3,P_3,v_2,u_2)$
    represents a Hamiltonian cycle that crosses the $uv$-curve at least
    three times (in $p_1$, $p_2$ and $p_3$). Thanks to
    Lemma~\ref{lem:crossing}, we obtain a Hamiltonian cycle $H'=
    (u_2,P_1,u_1,v_3,P_3,v_2,v_1,P_2,u_3,u_2)$ that differs from $H$
    only inside the face $f=(u,u_1,u_2,u_3,v,v_3,v_2,v_1)$ and crosses the $uv$-curve two fewer times than $H$
    does. Finally, note that the vertices $u$ and $v$ are part of
    either $P_1$, $P_2$, or $P_3$.}
  \label{fig:hc-two}

\end{figure}
  
\iflong The following lemma is crucial to obtain our normal form.
An illustration of the main ideas behind the proof is provided in  Figure~\ref{fig:hc-two}.
\fi

 \begin{lemma}
\label{lem:crossing}
  Let $G$ be a subhamiltonian graph with witness $(D,D_H,G_H,H)$, let
  $f$ be a face of $D$ and  let $c$ be a curve drawn inside $f$ between
  two vertices $u,v \in V(f)$. Then, there is a witness
  $(D,D_{H'},G_{H'},H')$ for $G$ such that:
  \begin{enumerate}[(1)]
   \item $D_{H'}$ and $D_H$ differ only inside $f$.
  \item $c$ crosses at most two curves corresponding to the
    edges of $H'$.
  \item $c$ crosses each curve corresponding to an
    edge of $H'$ at most once.
  \end{enumerate}

          \end{lemma}
\iflong \begin{proof}
  By Lemma~\ref{lem:edge}, we can assume that the every curve
  corresponding to an edge
  of $H$ inside $f$ crosses $c$ at most once in $D_H$, which shows
  \enref{3}.
  If $c$ crosses at
  most two edges of $H$, then the statement of the lemma holds. So
  suppose that this is not the case and let $u_1v_1$, $u_2v_2$, and
  $u_3v_3$ be three distinct edges in $E(H) \setminus E(G)$ that cross
  $c$ at three successive points $p_1$, $p_2$, and $p_3$ such that no other
  edge of $E(H) \setminus E(G)$ crosses $c$ between $p_1$ and
  $p_3$. Assume furthermore that $u_1$, $u_2$, and $u_3$ are on the
  same face in $D+c$, where here and in the following $D+c$ denotes
  the drawing obtained from $D$ after adding $c$, and the same for $v_1$, $v_2$, and $v_3$.
  Then, there are faces $f^1_H$ and $f^3_H$ of
  $D_H$ such that $u_1, v_1, u_2, v_2 \in V(f^1_H)$ and $u_2, v_2,
  u_3, v_3 \in V(f^3_H)$.

  Now we will analyze the Hamiltonian cycle $H$. For every $i$ in
  $\{1,2,3\}$, the edge $u_iv_i$ is in $E(H)$.
  Let $P_H$ be the path
  from $u_2$ to $v_2$ obtained by deleting the edge $u_2v_2$ from
  $H$. From Observation \ref{obs:cycle} applied to $f^i_H$ we obtain
  that $u_i$ is between $u_2$ and $v_i$ in $P_H$, for $i$ in
  $\{1,3\}$. Let $P_1,P_2,P_3$ be paths created after deleting
  $u_1v_1$ , $u_2v_2$, $u_3v_3$ from $H$. Then, either:
  \begin{itemize}
  \item $P_H = (u_2,P_1,u_1,v_1,P_2,u_3,v_3,P_3,v_2)$ or
  \item $P_H = (u_2,P_1,u_3,v_3,P_2,u_1,v_1,P_3,v_2)$.
  \end{itemize}
  The proof for both cases is entirely analogous, so we will only make the first
  case explicit. An illustration of the current setting is provided in Figure~\ref{fig:hc-two}.

  \sloppypar
  $H' = (V(H), (E(H)\setminus
  \{u_1v_1, u_2v_2, u_3v_3\}) \cup \{u_1v_3, v_2v_1, u_3u_2\} $ is a
  Hamiltonian cycle, because it corresponds to the sequence
  $(u_2,P_1,u_1,v_3,P_3,v_2,v_1,P_2,u_3)$.    
  At this point, we have to prove that there exists a planar drawing $D_{H'}$ of
  $G_{H'}$ that satisfies \enref{1}--\enref{3}. To do so we will change $D_H$.

  Let $G_H^\star$ be the graph obtained from $G_H$ after subdividing the
  edges $u_iv_i$ with the new vertex $p_i$ for every $i$ in $\{1,2,3\}$
  and adding the edges $p_1p_2$ and $p_2p_3$. Note that
  $G_H^\star$ is planar, as witnessed by the drawing $D^\star=D_H+c'$, where $c'$ is the restriction of $c$ to the
  segment between $p_1$ and $p_3$.

  Because $u_2$ and $u_3$ lie on the same face as $p_2$ and $p_3$
  in $D^\star$, we obtain from Observation~\ref{obs:face}
  that we can add the curve between $u_2$ and $u_3$ inside this
  face without adding any crossings.
  Analogously, we will add the curves $v_1v_2$, $u_1p_2$ and
  $p_2v_3$. We can now obtain a new drawing $D'$ from $D^\star$ by
  removing the curves $u_iv_i$ and adding the curves $u_2u_3$,
  $v_1v_2$, and the curve $u_1v_3$ obtained as the concatenation of
  the curves $u_1p_2$ and $p_2v_3$.
  
  Observe that in this new drawing we reduced the number of crossings
  by $2$, \ie, instead of the crossings at $p_1$, $p_2$, and $p_3$,
  only the crossing at $p_2$ remains ($u_1v_3$-curve). Moreover, all
  changes happened inside $f$.
  
  Finally,
  let $D_{H'}$ be the drawing obtained from $D'$ after removing the
  curve between $u$ and $v$. Then, $D_{H'}$ shows \enref{1}.
  Moreover, by repeating the process as long as we have at least $3$
  crossings with the $uv$-curve, we obtain a drawing that also
  satisfies \enref{2}.
\end{proof}\fi
We are now ready to define our normal form for the Hamiltonian
cycle. Essentially, we show that if there is a Hamiltonian cycle, then
there is one which crosses each subcurve that is either part of the
border of a node in the \SPQR{}-tree or that is a subcurve of some noose in a
sphere-cut decomposition of an R-node or an S-node at most twice.

Let $G$ be a biconnected subhamiltonian multi-graph with \SPQR{}-tree
$\mathcal{B}$ and the associated set $\mathcal{T}$ of sphere-cut decompositions $\langle
T_b,\lambda_b, \Pi_b \rangle$ of $\sk(b)$ for every R-node and S-node $b$ of $\mathcal{B}$.
We say that a witness $W=(D,D_H,G_H,H)$ for $G$ \emph{respects} 
the sphere-cut decompositions in $\mathcal{T}$, if there is a planar drawing of all nooses in the
sphere-cut decompositions of $\mathcal{T}$ into $D$ such that every
subcurve $c$ in $\bigcup_{a\in E(T_b)} O_a$ crosses the curves
corresponding to the edges of $H$ at most twice in $D_H$.
We say that the witness $W$ for $G$ \emph{respects} $\mathcal{B}$ if
it respects the sphere-cut decompositions in $\mathcal{T}$ and for
every node $b$ of $\mathcal{B}$ with reference edge $(s_b,t_b)$, it holds
that there is a noose $N_b$ that can be drawn into $D_H$ such that:
\begin{itemize}
\item $N_b$ touches $D$ only at $s_b$ and $t_b$.
\item $N_b$ separates $\pe(b)$ from $G\setminus \pe(b)$ in $D$.
\item Each of the two subcurves $L_b$ and $R_b$ obtained from $N_b$ by splitting
  $N_b$ at $s_b$ and $t_b$ crosses the curves corresponding to the edges
  of $H$ at most twice.
\item Moreover, if any of the subcurves of the nooses $O_a$ and the nooses
  $N_b$ connect the same two vertices in the same face of $D$, then
  the two subcurves are identical.
\end{itemize}
The following lemma allows us to assume our normal form and follows
easily
from \iflong \Cref{cor:add-nooses-to-drawing} together with \fi a repeated application of \Cref{lem:crossing}.
\begin{lemma}
\label{lem:nicedrawing}
  Let $G$ be a biconnected subhamiltonian multi-graph with \SPQR-tree
  $\mathcal{B}$ and the associated set $\mathcal{T}$ of sphere-cut decompositions. Then,
  there is a witness $W=(D,D_H,G_H,H)$ for $G$ that respects $\mathcal{B}$.
\end{lemma}
\iflong \begin{proof}
  Let $W=(D,D_H,G_H,H)$ be any witness for $G$, which exists because
  $G$ is subhamiltonian. Let $D'$ be the planar drawing obtained from
  $D$ using Corollary~\ref{cor:add-nooses-to-drawing}. That is, $D'$
  is a planar drawing of $G$ together with all nooses in
  $\SB O_a \SM a \in E(T_b) \land \langle T_b,\lambda_b,\Pi_b\rangle
  \in \mathcal{T}\SE$ as well as a noose $N_b$ for every node $b$ of
  $\mathcal{B}$ satisfying:
  \begin{itemize}
  \item $N_b$ intersects with $D$ only at $s_b$ and $t_b$.
  \item $N_b$ separates $\pe(b)$ from $G\setminus \pe(b)$ in $D$.
  \end{itemize}
  Moreover, if any of the subcurves of the nooses $O_a$ and the nooses
  $N_b$ connect the same two vertices in the same face of $D$, then
  the two subcurves are identical in $D'$. Note that this implies that
  every face of $D$ contains at most one subcurve from the
  nooses $O_a$ and $N_b$.

  Similarly, let $D_H'$ be obtained in the same manner from $D_H$. If
  $W$ already respects $\mathcal{B}$, then there is nothing to
  show. Otherwise, it holds that either:
  \begin{itemize}
  \item there is a noose $N_b$ such that $L_b$ or $R_b$ are crossed by
    the curves corresponding to the edges of $H$ more than twice or
  \item there is a subcurve $c \in O_a$ for some $a \in E(T_b)$ and
    $\langle T_b,\lambda_b,\Pi_b\rangle\in \mathcal{T}$ that is crossed by
    the curves corresponding to the edges of $H$ more than twice
  \end{itemize}
  Since every face of $D$ contains at most one subcurve from the
  nooses $O_a$ and $N_b$, it follows that in both cases, we can apply Lemma~\ref{lem:crossing} to obtain a
  witness $W'$ for $G$ that crosses $L_b$, $R_b$, or $c$, respectively, at most twice
  and does not introduce any additional crossings. This implies that a
  repeated application of Lemma~\ref{lem:crossing} allows us to obtain
  the desired witness that crosses each of the subcurves added to $D_H$
  in $D_H'$ at most twice. 
\end{proof}\fi

\section{Setting Up the Framework}

In this section we provide the foundations for our algorithm. That is,
in Subsection~\ref{sec:oneconn}, we show that it suffices to consider
biconnected graphs allowing us to employ \SPQR{}-trees. We then define
the types for nodes in the \SPQR{}-tree, which we compute in our
dynamic programming algorithm on \SPQR{}-trees, in
Subsection~\ref{ssec:typesSPQR}. Finally, in
Subsection~\ref{ssec:SPCD-framework} we introduce our general
framework for simplifying dynamic programming algorithms on sphere-cut
decompositions and introduce the types for nodes of a sphere-cut
decomposition\ifshort.\fi\iflong that we compute as part of our dynamic programming
algorithm on sphere-cut decompositions.\fi

\subsection{Reducing to the Biconnected Case}\label{sec:oneconn}

We begin by showing that any instance of \SH\ can be easily reduced to
solving the same problem on the biconnected components of the same
instance. It is well-known that \SH{} can be solved independently on
each connected component of the input graph, the following theorem now
also shows that the same holds for the biconnected components of the
graph and allows us to employ \SPQR{}-trees for our algorithm.

\begin{theorem}
\label{the:biconnected-comp}
  Let $G$ be a graph and let $C \subseteq V(G)$ such that $N(C) =
  \{n\}$, where $N(C)=\SB v \in V(G)\setminus C \SM \exists c \in C\ 
  \{v,c\} \in E(G)\SE$ is the set of neighbors of any vertex of $C$ in
  $V(G)\setminus C$. 
  Then $G$ is subhamiltonian if and only if both $G^-=G-C$ and
  $G^C=G[C \cup \{n\}]$ are subhamiltonian.
\end{theorem}
\iflong \begin{proof}
  If $G$ is subhamiltonian, then because $G^-$ and $G^C$ are both
  subgraphs of $G$ we obtain from Observation~\ref{obs:subham} \enref{1}
  that $G^-$ and $G^C$ are also subhamiltonian.

  Towards showing the reverse direction, suppose that $G^-$ and $G^C$
  are subhamiltonian. Therefore, $G^-$ and $G^C$ have witnesses $(D^-,D_{H^-},G_{H^-},H^-)$
  and $(D^C,D_{H^C},G_{H^C},H^C)$, respectively. Let $e^-=n^-v^-$ and
  $e^C=n^Cv^C$ be one of the two edges incident to
  $n$ in $H^-$ and $H^C$, respectively. Because any face can
  be drawn as the outer face of a planar graph, we can assume
  w.l.o.g. that the edges
  $e^-$ and $e^C$ are incident to the outer faces of the drawings
  $D_{H^-}$ and $D_{H^C}$, respectively.
  
  Let $G'$ be the graph obtained via the disjoint union of $G^-$ and
  $G^C$. Then, $G'$ is subhamiltonian because the cycle
  $H' = (V(G'), (E(H^-) \cup E(H^C) \cup \{n^-n^C, v^-v^C\}) \setminus
  \{e^-,e^C\})$ is a Hamiltonian cycle of $G'$ that has a planar drawing $D_{H'}$ which is obtained
  from the disjoint union of the drawing $D_{H^-}$ and $D_{H^C}$ after
  adding the edges
  $n^-n^C$ and $v^-v^C$ using Observation~\ref{obs:face}.
  Then from Observation~\ref{obs:subham} \enref{3} applied to $G'$ for the edge $n^-n^C
  \in E(H')$, we conclude that $G$ is also subhamiltonian, as desired.
\end{proof}\fi

\subsection{Defining the Types  for Nodes in the \SPQR{}-tree}\label{ssec:typesSPQR}

Here, we define the types for nodes in the
\SPQR{}-tree  that we
will later compute using dynamic programming. In the following, we
assume that $G$ is a biconnected multi-graph with \SPQR-tree
$\mathcal{B}$ and the associated set $\mathcal{T}$ of sphere-cut decompositions.
Let $b$ be a node of $\mathcal{B}$ with pertinent graph $\pe(b)$
and reference edge $e=(s,t)$. A \emph{type} of $b$ is a triple $(\psi,M,S)$
such that (please refer also to Figure~\ref{fig:types-SPQR} for an
illustration of some types):
\begin{itemize}
      \item $\psi$ is a function from $\{L,R\}$
  to subsets of $\{l,l',r,r'\}$ such that $\psi(L) \in \{\emptyset, \{l\},\{l,l'\}\}$
  and $\psi(R) \in \{\emptyset, \{r\},\{r,r'\}\}$. We denote by
  $V(\psi)$ the set $\psi(L)\cup \psi(R)$. Informally, $\psi$
  captures how many times the Hamiltonian cycle enters and exits the graph
  $\pe(b)$ from the left ($L$) and from the right ($R$).
\item $M \subseteq \SB \{u,v\} \SM u,v \in \{s,t\} \cup V(\psi) \land
  u\neq v\SE$ and $M$ is a
  non-crossing matching w.r.t. the circular ordering $(s,r,r',t,l',l)$
  that matches all vertices in $V(\psi)$ (i.e. $V(\psi)\subseteq V(M)$), where $V(M)=\bigcup_{e \in
    M}e$.
  Informally, $M$ captures the maximal path
  segments of the Hamiltonian cycle inside $\pe(b)\cup V(\psi)$ with endpoints
  in $\{s,t\} \cup V(\psi)$.
\item $S \subseteq \{s,t\}\setminus V(M)$. Informally, $S$ captures
  whether $s$ or $t$ are contained as inner vertices on path segments
  corresponding to $M$.
\end{itemize}

We now provide the formal semantics of types; see
\Cref{fig:types-SPQR} for an illustration.
Let $\XXX$ be the set of all types and $\pe^*(b)$ be the graph
obtained from $\pe(b)$ after adding the dummy
vertices $l$, $l'$, $r$, and $r'$ together with the edges
$sl$, $ll'$, $l't$, $sr$, $rr'$, and $r't$. 
We say that $b$ has type $X=(\psi,M,S)$ if there is a set $\PPP$ of
vertex-disjoint paths or a single cycle
in the complete graph with vertex set $V(\pe^*(b))$ such that:
\begin{itemize}
\item $\PPP$ consists of exactly one path $P_e$ between $u$ and $v$
  for every $e=\{u,v\} \in M$ or $\PPP$ is a cycle and $M=\emptyset$.
\item $\SB \INVP(P)\SM P\in \PPP\SE$ is a partition of
  $(V(\pe(b))\setminus \{s,t\})\cup S$, where
  $\INVP(P)$ denotes the set of inner vertices of $P$.
\item there is a planar drawing $D(b,X)$ of $\pe^*(b)\cup
  \bigcup_{P \in \PPP}P$ with outer-face $f$ such that
  $\seqb(f)=\{s,r,r',t,l',l\}$.
\end{itemize}

\begin{figure}
\begin{minipage}[c]{0.33\textwidth}
\begin{tikzpicture}[xscale=0.7, yscale=0.45]
\draw[color=red!30!white, very thick] (0,0) ellipse (2cm and 3.5cm);
\draw[very thick, blue] (-2,0) to[out=20,in=270] (0,3.5)(1.75,1.7) to[out=200,in=90] (0.5,0) to[out=270,in=160] (1.75,-1.7);
\draw[fill=black] (0,3.5)circle[radius=3pt];
   \draw (-2,0) node[very thick,cross=3.5pt] {};
\draw (1.75,-1.7) node[very thick,cross=3.5pt] {};
\draw (1.75,1.7) node[very thick,cross=3.5pt] {};

\draw[fill=blue](0,-3.5)circle[radius=3pt];
\node[above] at (0,3.6) {$s$};
\node[below] at (0,-3.6) {$t$};
\node[left] at (-2,0) {$l$};
\node[right] at (1.75,1.7) {$r$};
\node[right] at (1.75,-1.7) {$r'$};
\end{tikzpicture}
 \end{minipage}\hfill
\begin{minipage}[c]{0.33\textwidth}
\begin{tikzpicture}[xscale=0.7,yscale=0.45]
\draw[color=red!30!white, very thick] (0,0) ellipse (2cm and 3.5cm);
\draw[very thick, blue] (-1.75,1.7) to[out=20,in=270] (0,3.5)(-1.75,-1.7) to[out=30,in=260] (0,0) to[out=80,in=210] (1.75,1.7) (0,-3.5) to[out=90,in=200] (1.75,-1.7);
\draw[fill=black] (0,3.5)circle[radius=3pt](0,-3.5)circle[radius=3pt];
 \draw (-1.75,-1.7) node[very thick,cross=3.5pt] {};
\draw (-1.75,1.7) node[very thick,cross=3.5pt] {};
\draw (1.75,-1.7) node[very thick,cross=3.5pt] {};
\draw (1.75,1.7) node[very thick,cross=3.5pt] {};

\node[above] at (0,3.6) {$s$};
\node[below] at (0,-3.6) {$t$};
\node[left] at (-1.75,1.7) {$l$};
\node[left] at (-1.75,-1.7) {$l'$};
\node[right] at (1.75,1.7) {$r$};
\node[right] at (1.75,-1.7) {$r'$};
\end{tikzpicture}
 \end{minipage}\hfill
\begin{minipage}[c]{0.33\textwidth}
\begin{tikzpicture}[xscale=0.7,yscale=0.45]
\draw[color=red!30!white, very thick] (0,0) ellipse (2cm and 3.5cm);
\draw[very thick, blue] (-2,0)--(2,0);
\draw[fill=blue] 
(0,-3.5)circle[radius=3pt];
\draw[fill=white]
(0,3.5)circle[radius=3pt]
;

\draw (-2,0) node[very thick,cross=3.5pt] {};
\draw (2,0) node[very thick,cross=3.5pt] {};

\node[above] at (0,3.6) {$s$};
\node[below] at (0,-3.6) {$t$};
\node[left] at (-2,0) {$l$};
\node[right] at (2,0) {$r$};
\end{tikzpicture}
 \end{minipage}
\caption{The figure shows three different types of a node in an
  \SPQR-tree with reference edge $(s,t)$, i.e., the types shown are (from left to
  right): $(\{\{L\rightarrow \{l\}\},\{R\rightarrow
  \{r,r'\}\},\{\{l,s\},\{r,r'\}\},\{t\})$, $(\{\{L\rightarrow
  \{l,l'\}\},\{R\rightarrow
  \{r,r'\}\},\{\{l,s\},\{l',r\},\{t,r'\}\},\emptyset)$, and $(\{\{L\rightarrow \{l\}\},\{R\rightarrow \{r\}\}, \{\{l,r\}\},\{t\})$.
  The subset of $\{l,l'\}$ and $\{r,r'\}$ that appears corresponds to
  $\psi(L)$ and $\psi(R)$ respectively. The blue edges correspond to the
  matching $M$ and the blue vertices corresponds to $S$.}
\label{fig:types-SPQR}
\end{figure}
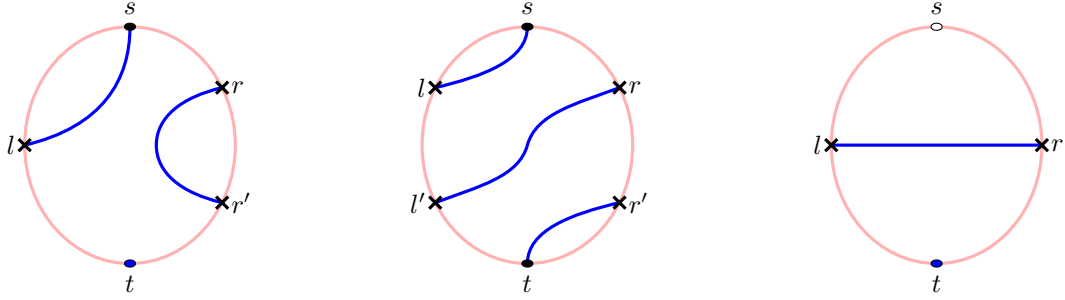

\ifshort 
The way we define the types $X=(\psi,M,S)$ of a node $b$ allows us to associate each witness $W=(D,D_H,G_H,H)$ with a type, denoted by $\type_W(b)$, based on the restriction of the witness to the respective pertinent graph. 
\fi
\iflong
The way we define the types $X=(\psi,M,S)$ of a node $b$ allows us to associate each witness with a type based on the restriction of the witness to the respective pertinent graph. 

Formally, let $W=(D,D_H,G_H,H)$ be a witness for $G$ that
respects $\mathcal{B}$. Then, the \emph{type} of a node $b$ in the \SPQR{}-tree of $G$ \emph{w.r.t.}\ 
$W$, denoted by $\type_W(b)$, is obtained as follows. Let $D_H^b$ be the drawing
$D_H$ restricted to the region with border $N_b$ containing
$\pe(b)$. Let $\CCC$ be the set of all curves in $D_H^b$
corresponding to path segments of $H$. Let $P_L$ ($P_R$) be the set of all
endpoints of curves in $\CCC$ in $L_b$ ($R_b$). Note that $|P_L|\leq
2$ and $|P_R|\leq 2$, the
endpoints of every curve in $\CCC$ are from the set $\{s,t\}\cup
P_L\cup P_R$, and the inner vertices of the curves represent a partition of
$V(\pe(b))\setminus \{s,t\}$. W.l.o.g., we assume that
$P_L\subseteq \{l,l'\}$ and $P_R \subseteq \{r,r'\}$. We are now ready
to define the type $X=(\psi,M,S)$ for $b$.
Let $M$ be the set containing the set of endpoints for
every curve in $\CCC$, let $\psi(L)=P_L$ and $\psi(R)=P_R$, and let $S \subseteq \{s,t\}$ be the set that contains $s$ (or $t$) if $s$ ($t$) occurs as an inner vertex of some curve in $\CCC$. Then, the type of $b$ w.r.t. $W$ is equal to $(\psi,M,S)$. Note also that the drawing $D_H^b$ witnesses that $b$ has type $(\psi,M,S)$.
\fi

\subsection{Framework for Sphere-cut Decomposition}\label{ssec:SPCD-framework}
Here, we introduce our framework to simplify the computation of
records via bottom-up dynamic programming along a sphere-cut
decomposition. Since the framework is independent of the type of
records one aims to compute, we believe that the framework is widely
applicable and therefore interesting in its own right.
In particular, we introduce a simplified framework for computing the types of arcs
(or, equivalently, nooses) in sphere-cut decompositions.
      
  Indeed, the central ingredient of any dynamic programming
  algorithm on sphere-cut decompositions is a procedure that given an
  inner node with parent arc $a_P$ and child arcs $a_L$ and $a_R$
  computes the set of types for the noose $O_{a_P}$ from the set of
  types for the nooses $O_{a_L}$ and $O_{a_R}$.
  Unfortunately, there is no simple way to obtain
  $O_{a_P}$ from $O_{a_L}$ and $O_{a_R}$ and this is why computing the
  set of types for $O_{a_P}$ from the set of types for $O_{a_L}$ and
  $O_{a_R}$ usually involves a technical and cumbersome case
  distinction~\cite{DornPBF10}. To circumvent this issue, we introduce a simple operation,
  i.e., the $\xor$ (\textbf{XOR}) operation defined below, and show that
  the noose $O_{a_p}$ can be obtained from the nooses $O_{a_L}$ and
  $O_{a_R}$ using merely a short sequence---one of length at most four---of $\xor$
  operations. 
        
Central to our framework is
the notion of \emph{weak nooses}, which are defined below and can be
seen as intermediate results in the above-mentioned sequence of simple
operations from the child nooses to the parent noose; in particular,
weak nooses are made up of subcurves of the nooses in the sphere-cut decomposition.
                  Let $G$ be a biconnected multi-graph and let $\mathcal{B}$ be an \SPQR-tree of $G$.
Let $b$ be an R-node or S-node of $\mathcal{B}$ with pertinent graph $\pe(b)$.
 Let $\langle T_b, \lambda_b, \Pi_b\rangle$ be a sphere-cut decomposition of $\sk(b)$
and $a$ be an arc of $T_b$ with pertinent graph $\pe(b,a)$.
Let $C(T_b)$ be the set of all subcurves of all nooses occurring in $T_b$, 
i.e., $C(T_b) = \bigcup_{a \in E(T_b)} O_a$ where $O_a$ is seen as a set of subcurves. We say $O$ is a \emph{weak
  noose} if $O$ is a noose consisting only of subcurves from
$C(T_b)$.
   \ifshort For each $O \subseteq C(T_b)$, let $\m{O}$ be equal to
the vertices of $G$ touched by the noose $O$.  
 \fi
   
\iflong
We now create the equivalent definition of $\md{a}$, which is defined
on subsets of subcurves in $C(T_b)$ (and therefore also weak nooses)
instead of on arcs in the sphere-cut decomposition, as follows.
For $O \subseteq C(T_b)$, the set $\m{O}$  is equal to
$\bigcup_{(\{u,v\},f) \in O}\{u,v\}$; note that this definition
is equivalent to the definition given for arcs, i.e., $\md{a}=\m{O_a}$
holds for every arc $a$ of $T_b$.
We also define $\pe(b,O)$ and $\sk(b,O)$ as the subgraph of $\pe(b)$
and $\sk(b)$, respectively, that is contained inside the weak noose
$O$; note that in particular $\pe(b,O_a)=\pe(b,a)$ and
$\sk(b,O_a)=\sk(b,a)$ for any arc $a$ of $T_b$.

While the above definitions are general, in our setting it will be
sufficient to restrict our attention to ``local'' weak nooses
consisting of $O_{a_P} \cup O_{a_L} \cup O_{a_R}$, where $a_P$ is a
parent arc for arcs $a_L$ and $a_R$.
Moreover, every weak noose $O$ in our setting will either separate an
edge-less graph with three nodes from $V(\sk(b))$, or separate the
graph $\sk(b,a)$ (where $a \in E(T_b)$) with at most one extra node from the rest of the graph. 
\fi

Having defined weak nooses, we will now define our simplified operation.
Let $A \xor B$ be an exclusive or for two sets $A$ and $B$, i.e. $A
\xor B = (A\cup B) \setminus (A\cap B)$. We will apply the
$\xor$-operation to weak nooses, whose $\xor$ is again a weak noose.
The following lemma, whose setting is illustrated in
\Cref{fig:parentchildnooses}, is central to our framework as it shows that
we can always obtain the noose for the parent arc $a_P$ from
the nooses of the child arcs $a_L$ and $a_R$ using a short sequence of
$\xor$-operations such that every intermediate result is a weak noose.
\iflong
Therefore, using our framework it is now sufficient to
show how to compute the set of types for a weak noose $O$ from the set
of types of two weak nooses $O_1$ and $O_2$ given that $O$ can be
obtained as $O_1 \xor O_2$. This greatly simplifies the computation of
types and can potentially be also applied to simplify dynamic
programming algorithms on sphere-cut decompositions for other problems
in the future.
\fi
\begin{lemma}
\label{lem:triangles}
  Let $a_P$ be a parent arc with two child arcs $a_L$ and $a_R$ in a
  sphere-cut decomposition $\langle T,\lambda,\Pi \rangle$ of a
  biconnected multi-graph $G$ with the drawing $D$.
  There exists a sequence $Q$ of at most $3$ $\xor$-operations such that:
  \begin{itemize}
       \item each step generates a weak noose $O$ with $|O|\leq 1+\max
         \{|\md{a_P}|,|\md{a_L}|,|\md{a_R}|\}$ as the \xor-operation of two
         weak nooses $O_1$ and $O_2$, whose inside region contains all
         subcurves in $(O_1\cap O_2)$,
      \item the last step generates the noose $O_{a_P}$,
      \item $Q$ contains $O_{a_L}$ and $O_{a_R}$ and at most two new
        weak nooses, each of them bounds the edge-less graph of size $3$. 
  \end{itemize}
      \end{lemma}
\iflong \begin{proof}
  Using Condition~(ii) in the definition of sphere-cut
  decompositions together with certain planar properties, it is
  straightforward to show that, the intersection of any
  two of the nooses $O_{a_P}$, $O_{a_L}$, and $O_{a_R}$, interpreted
  as curves, is a single segment and that the intersection of all
  three contains at most two points. Therefore, we obtain the
  following fact:

  \smallskip
  \textbf{(*)} The difference between the curve $O_{a_P}$ and the
  curve $O_{a_L} \cup O_{a_R}$ corresponds to at most two segments
  $s_1$ and $s_2$.

  \smallskip
  Let $O'$ be the set of subcurves $O_{a_P} \xor O_{a_L} \xor
  O_{a_R}$. Note that $O'$ is a union of simple closed
  curves (weak nooses) that do not share any subcurve, since it is
  obtained from $\xor$-operations of simple closed curves and
  because $O_{a_L}$ and $O_{a_R}$ are inside $O_{a_P}$ and the insides
  of $O_{a_L}$ and $O_{a_R}$ are disjoint; please also refer to
  Figure~\ref{fig:parentchildnooses} for an illustration of the
  current setting. Since $\sk(b,a_P) =
  \sk(b,a_L) \cup \sk(b,a_R)$ and $O'$ is outside of $O_{a_L}$ and
  $O_{a_R}$ but inside $O_{a_P}$, all weak nooses
  $O$ from $O'$ bound an edge-less graph. Therefore, all of the subcurves
  from $O$ are inside the same face of $D$. This together with
  Condition~(ii), implies that every weak noose $O$ in $O'$
  contains exactly one subcurve from each $O_{a_P}$, $O_{a_L}$ and
  $O_{a_R}$. Therefore, using the fact that the weak nooses in
  $O'$ do not share any subcurve together with \textbf{(*)}, we obtain
  that $O'$ contains
  at most two weak nooses (one for each of the at most two segments
  $s_1$ and $s_2$).
  In summary, $O'$ is a union of at most two weak nooses and each of
  them bounds an edge-less graph of size $3$.
  
  Now we are going to define a sequence of operations $Q$ that
  satisfies the conditions set out in the Lemma. In order to achieve
  this goal, we will make use of the following observation:

  \smallskip
  \textbf{(**)} If
  $|\m{O_a} \cap \m{O}| = 2$, then $O_a \xor O$ is a weak noose, where
  $O$ is a weak noose from $O'$, and $a\in \{a_L,a_R\}$.

  \smallskip
  Below we distinguish between three cases, depending on the  number
  of weak nooses $O'$ is made of.
  \begin{itemize}
  \item If $O' = \emptyset$ then $Q=O_{a_L} \xor O_{a_R}$.
    
  \item If $O'$ consists of only one weak noose $O_1$, then we do the following. 
                   If $|\m{O_{a_L}} \cap \m{O_1}| = 2$, then we can use the sequence
    $Q=(O_{a_L} \xor O_1) \xor O_{a_R}$. Similarly, if $|\m{O_{a_R}}
    \cap \m{O_1}| = 2$, then we can use the sequence
    $Q=(O_{a_R} \xor O_1) \xor O_{a_L}$.
    
    Otherwise, $|\m{O_{a_L}} \cap \m{O_1}| = |\m{O_{a_R}} \cap \m{O_1}|
    = 3$. Let $c=(\{u,v\}, f)$ be the common subcurve of $O_{a_P} \cap
    O_1$. Then, $c \notin O_{a_L} \cup O_{a_R}$ and $u,v
    \in \m{O_{a_L}} \cap \m{O_{a_R}}$.
    
    If $|O_{a_P}| = 1$, then $G$ is not biconnected which is against
    our assumption, otherwise $|O_{a_P}| > 1$, but then $O_{a_P}\cap
    O_{a_L} \cap O_{a_R}$ is not empty, which contradicts our previous
    observation that the intersection of all three nooses $O_{a_P}$,
    $O_{a_L}$, and $O_{a_R}$ (seen as curves) is at most two points.

  \item If $O'$ consists of two weak nooses $O_1$ and $O_2$, then for every 
    $a \in \{a_L,a_R\}$ and $i \in \{1,2\}$, it holds that $O_a \xor
    O_i$ is a weak noose, which can be seen as follows. Because of
    \textbf{(*)}, we obtain that $|\m{O_a} \cap \m{O_i}| = 2$ which
    together with \textbf{(**)} implies that $O_a \xor
    O_i$ is a weak noose. Therefore, we can use the sequence $Q = (O_{a_L}\xor O_1)
    \xor (O_{a_R} \xor O_2)$.
  \end{itemize}
  
  Note that, for each $\xor$-operation between $O_1$ and $O_2$ in the solution $Q$, 
  the region $O_1\xor O_2$ contains subcurves $O_1\cap O_2$, because
  intersection between any two regions made from nooses from $Q$ is
  only part of their boundaries. Moreover, to show that
  $|O|\leq 1+\max \{|\md{a_P}|,|\md{a_L}|,|\md{a_R}|\}$ for every weak
  noose $O$ obtained as an intermediate step, it suffices to consider
  case that $O$ is obtained from one of the at most parts $O''$ of $O'$ and
  $O_{a}$ for $a \in \{a_L,a_R\}$. But this follows because $\m{O''}$
  contains only one vertex that is not in $\m{O}$.
\end{proof}\fi

\begin{figure}
  \centering
\resizebox{!}{5cm}{\begin{tikzpicture}[scale=0.92]

\coordinate (L1) at ($(0.1,0.05)$);
\coordinate (L2) at ($(0.1,2)$);
\coordinate (L3) at ($(2.6,2.9)$);
\coordinate (L4) at ($(5.1, 2)$);
\coordinate (L5) at ($(5.1,0.05)$);

\coordinate (R1) at ($(0.1, -0.05)$);
\coordinate (R2) at ($(0.1,-2)$);
\coordinate (R3) at ($(2.6,-2.9)$);
\coordinate (R4) at ($(5.1, -2)$);
\coordinate (R5) at ($(5.1, -0.05)$);

\coordinate (P1) at ($(0.05,2.08)$);
\coordinate (P2) at ($(2.6,3)$);
\coordinate (P3) at ($(5.15, 2.08)$);
\coordinate (P35) at ($(6.4, 0)$);
\coordinate (P4) at ($(5.15, -2.08)$);
\coordinate (P5) at ($(2.6, -3)$);
\coordinate (P6) at ($(0.05,-2.08)$);
\coordinate (P65) at ($(-1.2, 0)$);

\coordinate (T1) at ($(0,-2.01)$);
\coordinate (T2) at ($(0,0)$);
\coordinate (T3) at ($(0,2.01)$);
\coordinate (T35) at ($(-1.1, 0)$);
\coordinate (T4) at ($(5.2, 2.01)$);
\coordinate (T5) at ($(5.2, 0)$);
\coordinate (T6) at ($(5.2,-2.01)$);

\coordinate (T65) at ($(6.3, 0)$);

\coordinate (U1) at ($(0.05, 0)$);
\coordinate (U2) at ($(0.05, 2.035)$);
\coordinate (U3) at ($(2.6, 2.95)$);
\coordinate (U4) at ($(5.15, 2.035)$);
\coordinate (U5) at ($(5.15, 0)$);
\coordinate (U6) at ($(5.15,-2.035)$);
\coordinate (U7) at ($(2.6, -2.95)$);
\coordinate (U8) at ($(0.05,-2.035)$);

\draw[very thick, green!50!black]
 (L1)to[out=135,in=225](L2)
 (L2)to[out=45,in=180](L3)
 (L3)to[out=0,in=135](L4)
 (L4)to[out=315,in=45](L5)
(L5)--(L1)
;

\draw[very thick, red]
 (R1)to[out=225,in=135](R2)
 (R2)to[out=315,in=180](R3)
 (R3)to[out=0,in=225](R4)
 (R4)to[out=45,in=315](R5)
 (R5)--(R1)
;

\draw[very thick, black]
 (P1)to[out=45,in=180](P2)
 (P2)to[out=0,in=135](P3)
 (P3)to[out=0,in=90](P35)
 (P35)to[out=270,in=0](P4)
 (P5)to[out=0,in=225](P4)
 (P6)to[out=315,in=180](P5)
 (P1)to[out=180,in=90](P65)
 (P65)to[out=270,in=180](P6)
;
\draw[very thick, blue]
 (T35)to[out=90,in=190](T3)
 (T35)to[out=270,in=170](T1) 
 (T1)to[out=135,in=225](T2)
 (T2)to[out=135,in=225](T3)
 (T65)to[out=90,in=355](T4)
 (T65)to[out=270,in=5](T6)
 (T5)to[out=45,in=315](T4) 
 (T6)to[out=45,in=315](T5)
;

    \draw[thick, fill=white]
(U1) circle (3pt)
(U2) circle (3pt)
(U3) circle (3pt)
(U4) circle (3pt)
(U5) circle (3pt)
(U6) circle (3pt)
(U7) circle (3pt)
(U8) circle (3pt)
;

\node[above right] at  (U1) {$u_1$};
\node[below right] at (U2) {$u_2$};
\node[below] at (U3) {$u_3$};
\node[below left] at (U4) {$u_4$};
\node[above left] at (U5) {$u_5$};
\node[above left] at (U6) {$u_6$};
\node[above] at (U7) {$u_7$};
\node[above right] at (U8) {$u_8$};

\node at (2.6,1.5) {\textcolor{green!50!black}{$O_{a_{L}}$}};
\node at (2.6,-1.5) {\textcolor{red}{$O_{a_{R}}$}};
\node at (-0.7,0) {\textcolor{blue}{$O_1$}};
\node at (5.9,0.0) {\textcolor{blue}{$O_2$}};
\node at (2.6,3.5) {$O_{a_{P}}$};
\end{tikzpicture}}
                              \iflong
\caption{
An illustration of the relationship of the parent noose
  $O_{a_{P}}$ and the child nooses $O_{a_{L}}$ and $O_{a_{R}}$. The
  illustration represents the case of Lemma~\ref{lem:triangles} where 
  $O'=O_{a_P} \xor O_{a_L} \xor O_{a_R}$ consists of two disjoint weak
  nooses (triangles) $O_1$ and $O_2$. 
  Let $c_{i,j}$ be a curve between $u_i$, $u_j$ then all nooses are defined as follows: 
  $O_{a_P} = \{c_{2,3}, c_{3,4}, c_{4,6}, c_{6,7}, c_{7,8}, c_{2,8}\}$,
  $O_{a_L} = \{c_{1,2},c_{2,3}, c_{3,4},c_{4,5}, c_{1,5}\}$,
  $O_{a_R} = \{c_{1,5},c_{5,6}, c_{6,7}, c_{7,8}, c_{1,8}\}$,
  $O_{1} = \{c_{1,2},c_{2,8}, c_{1,8}\}$,
  $O_{2} = \{c_{4,5},c_{5,6}, c_{4,6}\}$,
  $O'= O_1 \cup O_2$.
  }
\fi
\ifshort
\caption{
An illustration of the relationship of the parent noose
  $O_{a_{P}}$ and the child nooses $O_{a_{L}}$ and $O_{a_{R}}$. The
  illustration represents the case of Lemma~\ref{lem:triangles} where 
  $O'=O_{a_P} \xor O_{a_L} \xor O_{a_R}$ consists of two disjoint weak
  nooses (triangles) $O_1$ and $O_2$.}
\fi
\label{fig:parentchildnooses}
\end{figure}

We are now ready to define the types of weak nooses, which informally
can be seen as a generalization of the types of nodes in an
\SPQR{}-tree introduced in Subsection~\ref{ssec:typesSPQR}. An illustration of the types is also provided in Figure~\ref{fig:typesSPCD}. In the following we fix an arbitrary order $\pi_G$ of the vertices in $G$.
A type of a weak noose $O$ is a triple $(\psi, M, S)$ such that:
\iflong
\begin{itemize}
\item $\psi$ is a function that for each subcurve $c=(\{u,v\}, f)$ in $O$
  returns a sequence of at most two new nodes. Informally,
  these two nodes are on the subcurve $c$ (in the order given by
  $\psi(c)$) and if $\pi_G(u) < \pi_G(v)$ we assume that
  $[u,\psi(c),v]$ is the sequence of nodes on the subcurve $c$.
\item $S$ is a subset of $\m{O}$.
\item $M \subseteq \SB \{u,v\} \SM u,v \in V(\psi) \cup
  (\m{O}\setminus S) \land u\neq v\SE$, $V(\psi)\subseteq
  V(M)$, and $M$ is a non-crossing matching w.r.t. the circular order $\pi^\circ(\psi)$
  defined as follows. $\pi^\circ(\psi)$ is the circular order obtained from
  the circular order $\pi^\circ(O)$ of $\m{O}$ after adding $\psi(c)$
  between $u$ and $v$, for every $c=(\{u,v\},f) \in O$ assuming that
  $\pi_G(u) < \pi_G(v)$.
\end{itemize}
\fi
\ifshort
(1) $\psi$ is a function that for each subcurve $c=(\{u,v\}, f)$ in
$O$\ifshort, i.e., the subcurve of $O$ between $u$ and $v$ in face
  $f$, \fi
  returns a sequence of at most two new nodes,
(2) $S$ is a subset of $\m{O}$, and
(3) $M \subseteq \SB \{u,v\} \SM u,v \in V(\psi) \cup
  (\m{O}\setminus S) \land u\neq v\SE$, $V(\psi)\subseteq
  V(M)$, and $M$ is a non-crossing matching w.r.t. the circular order $\pi^\circ(\psi)$
  defined as follows. $\pi^\circ(\psi)$ is the circular order obtained from
  the circular order $\pi^\circ(O)$ of $\m{O}$ after adding $\psi(c)$
  between $u$ and $v$, for every $c=(\{u,v\},f) \in O$ assuming that
  $\pi_G(u) < \pi_G(v)$.

The semantics for the types as well as the definition of a type given
a witness are now defined in a similar way as in the case of types for
\SPQR{}-tree nodes.
\fi

\iflong
Let $O$ be a weak noose and $X=(\psi,M,S)$ be a type. We define the
graph $\pe_X(b,O)$ as the graph obtained from $\pe(b,O)$ after adding
the vertices in $V(\psi)$ and all edges on the cycle
$\pi^\circ(\psi)$. We say that $O$ has type $X$ if there is a set 
$\PPP$ of vertex-disjoint paths or a cycle in the
complete graph with vertex set $V(\pe_X(b,O))$ such that:
\begin{itemize}
\item $\PPP$ consists of exactly one path $P_e$ between $u$ and $v$
  for every $e=\{u,v\} \in M$ or $\PPP$ is a cycle and $M=\emptyset$.
\item $\SB \INVP(P)\SM P\in \PPP\SE$ is a partition of
  $\bigl(V(\pe(b,O)) \setminus \m{O}\bigr) \cup S$, where
  $\INVP(P)$ denotes the set of inner vertices of $P$.
\item there is a planar drawing $D_X$ of $\pe_X(b,O)\cup
  \bigcup \PPP$ with outer-face $f$ such that
  $\seqb(f)=\pi^\circ(\psi)$.
\end{itemize}
We also say that the fact that $O$ has type $X$ is \emph{witnessed} by the pair $(\PPP,D_X)$.
We say that type $X=(\psi,M,S)$ of $O$ is the
  \textit{\fulltype}, if $M=\emptyset$ and $S=\m{O}$, which informally
  means that $\PPP$ is a Hamiltonian cycle. Moreover we say
  that type $X$ is the
  \textit{\emptytype}, if $M=S=\emptyset$, which may only occur when
  $\pe(b,O)$ is merely an edge.
  \fi

\iflong
\begin{lemma}
\label{lem:enum_arc_types}
  Let $O$ be weak noose. Then, the number of types defined on $O$ is
  at most $28^{|O|}$ and all possible types for $O$ can be
  enumerated in time $\bigoh(28^{|O|}|O|)$.
\end{lemma}
\begin{proof}
    Let $X=(\psi,M,S)$ be a type that can be defined on a weak noose
    $O$. Let $DW$ be the Dyck word corresponding to the matching $M$
    from Observation~\ref{obs:dyck}.
    For each $v \in \m{O}$ there are $4$ possibilities of the role of
    $v$ in type $X$, i.e., $v \in S$, $v\notin S \cup V(M)$ or $v \in
    V(M)$ and $v$ corresponds to either $"["$ or $"]"$ in $DW$.
    Note that due to the type definition  we get that $V(\psi) \subseteq V(M)$. Therefore for each $v \in V(\psi)$ there are $2$ possibilities of the role of $v$ in $X$, i.e., $v$ corresponds to $"["$ or $"]"$ in $DW$.
    For each subcurve $c \in O$, there are $3$ possible values
    $\{\emptyset, [x], [x, x']\}$ for $\psi(c)$, and therefore there
    are $1+2+4=7$ possibilities, i.e., $1$, $2$, and $4$ possibilities
    in case that $\psi(c)=\emptyset$, $\psi(c)=[x]$, and
    $\psi(c)=[x,x']$, respectively, of the role of $c$ in type $X$.
    Furthermore, since $|O| = |\m{O}|$, there are at most
    $4^{|O|}7^{|O|} = 28^{|O|}$ types that can be defined
    on $O$.

    We can generate all types, by choosing a starting vertex on $\m{O}$
    together with a direction. We can then assign a role to each
    vertex in $\m{O}$ and every subcurve of $O$ and verify that the
    corresponding word is a Dyck word in time $\bigoh(|O|)$ and if
    so translate it into a type description using
    Observation~\ref{obs:dyck}. Since there are at most
    $28^{|O|}$ possibilities to check and each can be checked in
    time $\bigoh(|O|)$, we obtain $\bigoh(28^{|O|}|O|)$ as
    the total run-time to enumerate all possible types for $O$.
\end{proof}

Finally, we now defined the type of weak nooses for a given witness.
Let $W=(D,D_H,G_H,H)$ be a witness for $G$ that
\emph{respects} $\mathcal{T}$ and let $b$ be an R-node or an S-node with sphere-cut decomposition
$\langle T_b,\lambda_b,\Pi_b \rangle \in \mathcal{T}$. Then, the \emph{type} of a weak noose $O\subseteq C(T_b)$ w.r.t. $W$, denoted by $\type_W(b,O)$, is obtained as follows. Let $D_H'$ be the drawing obtained from $D_H$
after adding the noose $O$. Then, because $W$ respects $\mathcal{T}$, it holds that every subcurve $c \in O$ is crossed at most twice in $D_H'$. In the following we will assume that we replaced every such crossing with a new vertex in $D_H'$ and that these vertices are also introduced into $G_H$ and $H$. Moreover, we let $\psi(c)$ be the sequence of (the at most two) new vertices introduced in this way for the subcurve $c=(\{u,v\},f) \in O$ such that 
$[u,\psi(c),v]$ is the ordering of the vertices on $c$ assuming that $\pi_G(u)<\pi_G(v)$. Let $D_H^O$ be the drawing $D_H'$ restricted to $\pe(b,O)\cup V(\psi)$ and let $G_H^O$ and $H^O$ be obtained in the same way from $G_H$ and $H$, respectively. Let $f_O$ be the face in $D_H^O$ such that $V(f_O)=\m{O} \cup V(\psi)$.
Let $\PPP$ be a set of all maximal paths in $H^O$ each of size at least $2$. Then, $S$ is the set of all vertices in $\m{O}$ that have degree two in $\PPP$ and 
the matching $M$ contains edge between the endpoints of every path in $\PPP$.
Then, the type $\type_W(b,O)$ is equal to the triple $X=(\psi,M,S)$. Note that $X$
satisfies all properties of a type because of the following.
First  every node $v$ in $V(\psi)$ has degree $1$ in $\PPP$ and therefore $V(\psi) \subseteq V(M)$. Moreover, $M$ is a non-crossing matching w.r.t. 
$\pi^\circ(\psi)=\seqb(f_O)$ because $D_H^O$ is a planar drawing of $\PPP$
with face $f_O$. Therefore, the weak noose $O$ has type $X$ and this is witnessed
by the pair $(\PPP,D_H^O)$. 
\fi

\section{An FPT-algorithm for \SH{} using Treewidth}
\label{sec:subhamtw}

In this section we show that \SH{} admits a constructive single-exponential fixed-parameter algorithm
parameterized by treewidth.
\begin{theorem}\label{the:sh-fpt-tw}
       \SH{} can be solved in time $2^{\bigoh(\tw)}\cdot n^{\bigoh(1)}$, where $\tw$ is the treewidth of the input graph.
\end{theorem}

Since the treewidth of an $n$-vertex planar graph is upper-bounded by
$\bigoh(\sqrt{n})$~\cite{GuT12,Marx20,RobertsonST94} and there are
single-exponential constant-factor approximation algorithms for
treewidth~\cite{Korhonen21}, \Cref{the:sh-fpt-tw} immediately implies
the following corollary.  \begin{corollary}
\label{cor:main}
  \SH{} can be solved in time $2^{\bigoh(\sqrt{n})}$.
\end{corollary}

The main component used towards proving~\Cref{the:sh-fpt-tw} is the
following lemma\ifshort, from which \Cref{the:sh-fpt-tw} follows as an
  easy consequence \fi.

\begin{restatable}{lemma}{lemshfpt}\label{lem:sh-fpt-bw-given}
  Let $G$ be a biconnected multi-graph with $n$ vertices and $m$ edges
  and \SPQR{}-tree $\mathcal{B}$. Then, we can decide in time 
  $\bigoh(315^{\omega}n+n^3)$ whether $G$ is subhamiltonian, where
  $\omega$ is the maximum branchwidth of $\sk(b)$ over all R-nodes and
  S-nodes $b$ of $\BBB$.
\end{restatable}

\iflong
With the help of~\Cref{lem:sh-fpt-bw-given}, \Cref{the:sh-fpt-tw} can
now be easily shown as follows.
 \begin{proof}[Proof of~\Cref{the:sh-fpt-tw}]
  Let $G$ be the graph given as input to \SH{} having $n$
  vertices and $m$ edges. Because of
  \Cref{the:biconnected-comp}, we can assume that $G$ is biconnected
  since we can otherwise solve every biconnected component of $G$
  independently. We first test
  whether $G$ is planar, which is well-known do be achievable in
  linear-time~\cite{DBLP:journals/jacm/HopcroftT74}. If this is not
  the case, the algorithm correctly outputs no. Otherwise, the
  algorithm uses~\Cref{lem:computeSPQR} to compute an \SPQR{}-tree
  $\mathcal{B}$ of
  $G$ with at most $\bigoh(m)$ nodes and edges inside skeletons in
  time at most $\bigoh(n+m)$; note that $\bigoh(m)=\bigoh(n)$ because $G$ is planar. We then employ
  Lemma~\ref{lem:sh-fpt-bw-given} to solve \SH{} in time
  $\bigoh(315^{\omega}n+n^3)$, where $\omega$ is the maximum
  branchwidth of $\sk(b)$ over all R-nodes and S-nodes $b$ of $\BBB$. Since
  $\omega$ is an upper bound on the branchwidth of $G$, we obtain
  from~\Cref{lem:bw-tw} that the branchwidth of $G$ is at most the
  treewidth $\tw(G)$ of $G$ plus $1$, which
  implies that \SH{} can be solved in time $\bigoh(315^{\tw(G)}n+n^3)$,
  as required.
\end{proof}\fi

The remainder of this section is therefore devoted to a proof
of~\Cref{lem:sh-fpt-bw-given}, which we show by
providing a bottom-up dynamic programming algorithm along the
\SPQR{}-tree of the graph. That is, let $G$ be a biconnected multi-graph,
$\mathcal{B}$ be an \SPQR-tree of $G$ with associated set
$\mathcal{T}$ of sphere-cut decompositions for every R-node and S-node
of $\BBB$. Using a dynamic programming algorithm starting at
the leaves of $\mathcal{B}$, we will
compute a set $\RRR(b)$ of all types $X$ satisfying the following two conditions:
\begin{enumerate}[(R1)]
\item If $X \in \RRR(b)$, then $b$ has type $X$.
\item If there is a witness $W=(D,D_H,G_H,H)$ for $G$ that
  \emph{respects} $\mathcal{B}$ such that $b$ has type 
  $X=\type_W(b)$,
     then $X \in \RRR(b)$.
\end{enumerate}
Interestingly, we do not know whether it is possible to compute the
set of all types $X$ such that $b$ has type $X$ as one would
usually expect to be able to do when looking at similar algorithms
based on dynamic programming. 
That is, we do not know whether one can
compute the set of types that also satisfies the reverse
direction of \enref{R1}. While we do not know, we suspect that this is not the
case because $b$ might have a type that can only be
achieved by crossing some sub-curves of nooses inside of $\pe(b)$ more than
twice. Indeed~\Cref{lem:crossing}, which allows us to avoid more than
two crossings per sub-curve, requires the property that the type of $b$ can be
extended to a Hamiltonian cycle of the whole graph, which is clearly
not necessarily the case for every possible type of $b$. 

\iflong
This section is organized as follows. First in
Subsection~\ref{ssec:pnodes}, we show how to compute $\RRR(b)$ for every P-node $b$
of $\BBB$. This is probably the most challenging part of the
algorithm and we show that instead of having to enumerate all possible
orderings among the children of $b$ in $\BBB$, we merely have to
consider a constant number children and their orderings. This allows
us to compute  $\RRR(b)$ very efficiently in time
$\bigoh(\ell)$, where $\ell$ is the number of
children of $b$ in $\BBB$. Then, in Subsection~\ref{ssec:rsnodes}, we
show how to compute $\RRR(b)$ for any R-node and S-node $b$ of $\BBB$ using a
dynamic programming algorithm on a sphere-cut decomposition of
$\sk(b)$.  We then put everything together and
show~\Cref{lem:sh-fpt-bw-given} in Subsection~\ref{ssec:puttogether}.
\fi

 \subsection{Handling P-nodes}\label{ssec:pnodes}
    In this part, we show how to compute the set of types for any
$P$-node in the given \SPQR{}-tree by establishing the following lemma. 
\begin{lemma}
\label{lem:Pnode}
  Let $b$ be a P-node of $\mathcal{B}$ such that $\RRR(c)$
  has already been computed for every child
  $c$ of $b$ in $\mathcal{B}$. Then, we can compute
  $\RRR(b)$ in time $\bigoh(\ell)$,
  where $\ell$ is the number of children of $b$ in $\mathcal{B}$.
\end{lemma}
In the following, let $b$ be a P-node of $\mathcal{B}$ with reference
edge $(s,t)$ and let $C$ with $|C|=\ell$ be the set of all children of $b$ in
$\mathcal{B}$. Informally, $\RRR(b)$ is the set of types $X$ such that there
is an ordering $\rho=(c_1,\dotsc,c_\ell)$ of the children in $C$ and
an assignment $\tau : C \rightarrow \XXX$ of children to types with
$\tau(c) \in \RRR(c)$ for every child $c \in C$ that ``realizes'' the
type $X$ for $b$. The main challenge is to compute $\RRR(b)$ efficiently, i.e.,
without having to enumerate all possible orderings $\rho$ and assignments $\tau$.
Below, we make this intuition more precise before proceeding.

For a type $X=(\psi,M,S)$ of $b$ and $A \in \{L,R\}$, we let
$\#_A(X)=|\psi(A)|$. Moreover, for every $A \in \{s,t\}$,
we set $\#_A(X)$ to be equal to $2$ if
$A \in S$, equal to $1$ if $A \in V(M)$ and equal to $0$ otherwise.
Next, let $\rho=(X_1,\dotsc,X_\ell)$ be a sequence of types, where
$X_i=(\psi_i,M_i,S_i)$ for every $i$ with $1 \leq i \leq \ell$.
We say that $\rho$ is \emph{weakly compatible} if the following holds:
\begin{enumerate}[(C1)]
\item for every $i$ with $1 \leq i < \ell$,
  $\#_R(X_i)=\#_L(X_{i+1})$, and
\item $\sum_{i=1}^\ell \#_s(X_i)\leq 2$ and
  $\sum_{i=1}^\ell \#_t(X_i)\leq 2$.
\end{enumerate}
Note that \enref{C1} corresponds to our assumption made 
in \iflong \Cref{cor:add-nooses-to-drawing} \fi\ifshort
  \Cref{lem:nicedrawing} \fi that we can add the nooses $N_b$
to any planar drawing $D$ of $G$ such that every face of $D$ contains
at most one subcurve of any $N_b$. This in particular means that
if $\pe(c)$ is drawn immediately to the left of $\pe(c')$ for two
children $c$ and $c'$ of $b$, then the subcurves $R_{c}$ and $L_{c'}$
are identical. Please also refer to Figure~\ref{fig:shtd} for an
illustration of these subcurves.

 \begin{figure}
\begin{center}
\centering
\begin{tikzpicture}[xscale=1.01,yscale=0.54]
\draw[thick] (0,3)--(1,0.9) (1,-0.9)--(0,-3)(0,3)--(3,0.7)
(3,-0.7)--(0,-3)(0,3)--(5,0.8) (5,-0.8)--(0,-3)(0,3)--(-1,0.9)
(-1,-0.9)--(0,-3)(0,3)--(-3,0.7) (-3,-0.7)--(0,-3)(0,3)--(-5,0.8)
(-5,-0.8)--(0,-3);\draw[fill=gray!20!white,thick] (1,0) ellipse (0.5cm
and 0.9cm)(3,0) ellipse (0.5cm and 0.9cm)(5,0) ellipse (0.5cm and
0.9cm)(-1,0) ellipse (0.5cm and 0.9cm)(-3,0) ellipse (0.5cm and
0.9cm)(-5,0) ellipse (0.5cm and 0.9cm);\draw[very
thick,blue](0,-3)--(-3,-0.3)--(-5,-0.3)--(-5,0.3)--(-3,0.3)--(4.75,-0.3)--(4.75,0.3)--(1,0.3)--(0,3)
(5.25,0.3)--(5.25,-0.3)(5.25,0.3)
to[out=22,in=210](5.45,0.39)(5.25,-0.3)
to[out=-22,in=150](5.45,-0.39);\draw[very
thick,dashed,blue](0,3)to[out=30,in=20]
(5.25,0.3)(0,-3)to[out=-30,in=-20] (5.25,-0.3);\draw[very
thick,red!30!white](0,3)--(0,-3)(0,-3) to[out=140,in=270](-4,0)(0,3)
to[out=220,in=90](-4,0)(0,-3) to[out=40,in=270](4,0)(0,3)
to[out=320,in=90](4,0)(0,-3) to[out=115,in=270](-2,0)(0,3)
to[out=245,in=90](-2,0)(0,-3) to[out=65,in=270](2,0)(0,3)
to[out=295,in=90](2,0)(0,-3) to[out=20,in=270](5.7,0)(0,3)
to[out=-20,in=90](5.7,0)(0,-3) to[out=160,in=270](-5.7,0)(0,3)
to[out=200,in=90](-5.7,0);\draw[fill=black] (0,3)circle
[radius=3pt](0,-3)circle [radius=3pt];

\draw[color=violet,fill=violet]
(-3.975,0.3)node[very thick,cross=3pt] {} [radius=2pt](-3.975,-0.3)node[very thick,cross=3pt] {}
[radius=2pt](3.975,0.3)node[very thick,cross=3pt] {} [radius=2pt](3.98,-0.25)node[very thick,cross=3pt] {}
[radius=2pt](-1.98,0.23)node[very thick,cross=3pt] {} [radius=2pt](1.99,-0.09)node[very thick,cross=3pt] {}
[radius=2pt](0,0.075)node[very thick,cross=3pt] {} [radius=2pt](5.65,0.55)node[very thick,cross=3pt] {}
[radius=2pt](5.65,-0.55)node[very thick,cross=3pt] {} [radius=2pt](1.99,0.3)node[very thick,cross=3pt] {}
[radius=2pt];\draw[color=blue,fill=blue](5.25,0.3)circle
[radius=1pt](5.25,-0.3)circle [radius=1pt](-5,0.3)circle
[radius=1pt](-5,-0.3)circle [radius=1pt](-3,-0.3)circle
[radius=1pt](4.75,0.3)circle [radius=1pt](4.75,-0.3)circle
[radius=1pt](1,0.3)circle [radius=1pt];\node[below left] at
(-0.2,-3.2) {$t$};\node[above left] at (-0.2,3.2) {$s$};
\end{tikzpicture}
\end{center}
\vspace{-0.3cm}
\iflong
\caption{An illustration of how a Hamiltonian Cycle in normal form can
  interact with a drawing of $\pe(b)$ for a P-node $b$ of $\BBB$. Here, the pertinent graphs
  $\pe(c)$ for all children $c$ of $b$ (without the nodes $s$ and $t$ of
  the common reference edge $(s,t)$) are represented by gray ellipses.
  The Hamiltonian cycle is given in blue with dashed segments
  representing path segments outside of $\pe(b)$. The red curves
  represent the subcurves of $N_c$ for every child $c$ of $b$. Note
  again that the subcurves $R_c$ and $L_{c'}$ are identical for every
  two children $c$ and $c'$ such that $\pe(c)$ is drawn immediately to
  the left of $\pe(c')$. The drawing is
  in normal form, \ie, is a drawing that respects the nooses $N_c$ for
  every child $c$ of $b$, because
  every red curve is intersected by $H$ at most twice. In this figure
  all but the types of the second and fourth pertinent graph are
  clean. Moreover, the type of the third
  and fifth pertinent graphs are $1$-good and $2$-good, respectively,
  and the types of all other pertinent graphs are bad.}
\fi
\ifshort
\caption{An illustration of how a Hamiltonian Cycle in normal form can
  interact with a drawing of $\pe(b)$ for a P-node $b$. Here, the pertinent graphs
  $\pe(c)$ for all children $c$ of $b$ (without the nodes $s$ and $t$ of
  the common reference edge $(s,t)$) are represented by gray ellipses.
  The Hamiltonian cycle is given in blue with dashed segments
  representing path segments outside of $\pe(b)$. The red curves
  represent the subcurves of $N_c$ for every child $c$ of $b$. 
  In this figure
  all but the types of the second and fourth pertinent graph are
  clean. Moreover, the type of the third
  and fifth pertinent graphs are $1$-good and $2$-good, respectively,
  and the types of all other pertinent graphs are bad. }
\fi
\label{fig:shtd}
\end{figure}

Let $\rho$ be weakly compatible. We define the following auxiliary graph
$H(\rho)$. $H(\rho)$ has two vertices $s$ and $t$ and additionally
for every $i$ with $1 \leq i \leq \ell$ and every vertex
$v \in V(\psi)$, $H(\rho)$ has a vertex $v_i$. For convenience,
we also use $s_i$ and $t_i$ to refer to $s$ and $t$, respectively.
Moreover, $H(\rho)$ has the following edges:

\begin{itemize}
\item for every $1 \leq i \leq \ell$ if $M_i=\emptyset$ and
  $S_i=\{s_i,t_i\}$, $H(\rho)$ has a cycle on $s_i$ and $t_i$,
\item for every $1 \leq i \leq \ell$ if $M_i \neq \emptyset$ then for every $e=\{u,v\}\in M_i$, $H(\rho)$ has the edge $\{u_iv_i\}$,
\item for every $1 \leq i < \ell$, $H(\rho)$ contains the edge
  $\{r_i,l_{i+1}\}$ if $r \in \psi_i(R)$ and $l \in \psi_{i+1}(L)$,
\item for every $1 \leq i < \ell$, $H(\rho)$ contains the edge $\{r_i',l_{i+1}'\}$ if $r'
  \in \psi_i(R)$ and $l' \in \psi_{i+1}(L)$.
\end{itemize}

\iflong
\begin{lemma} 
\label{lem:planar}
  Let $\rho$ be weakly compatible. Then, $H(\rho)$ is planar.
\end{lemma}
\begin{proof}
  Because $M_i$ is a non-crossing matching w.r.t.\ the cyclic
  ordering $(s, r, r', t, l', l)$ for every $i \in [1,\ell]$, it holds
  that the graph $H(\rho)$ induced by the vertices in $\{s_i,t_i\} \cup
  V(\psi_i)$ has a planar drawing $D_i$, where $\psi_i(L)$ are placed to the
  east, $s$ is placed in the north, $\psi_i(R)$ is placed in the west,
  and $t$ is placed on the south. Taking the disjoint union of the
  drawings $D_i$ ordered $D_1,\dotsc,D_\ell$ from east to west and
  identifying all $s_i's$ with $s$ and all $t_i$'s with $t$, then
  gives a planar drawing of $H(\rho)$.
\end{proof}
\fi
We say that $\rho$ is \emph{compatible} if it is weakly compatible and
furthermore either $H(\rho)$ is acyclic, or $H(\rho)-(\bigcup_{i=1}^\ell S_i)$ is a single
(Hamiltonian) cycle.

\iflong
\begin{lemma} 
\label{lem:DUpaths}
  Let $\rho$ be compatible such that $H(\rho)$ is acyclic. Then, $H(\rho)$ is the
  disjoint union of paths whose endpoints are in
  $\{s,t,l_1,l_1',r_\ell,r_\ell'\}$. Moreover, no vertex in
  $\{l_1,l_1',r_\ell,r_\ell'\}$ can be an inner vertex of those paths.
\end{lemma}
\begin{proof}
  We first show that the degree of every vertex in $H(\rho)$ is at most
  two. Because of \enref{C2} this clearly holds for the vertices $s$ and
  $t$. Moreover, every vertex in $v \in \{l_i,l_i',r_i,r_i'\}$ for any
  $i$ with $1 \leq i \leq \ell$ has exactly one neighbor
  among $\{l_i,l_i',r_i,r_i',s,t\}$ and at most one neighbor in
  $V(H(\rho))\setminus \{l_i,l_i',r_i,r_i',s,t\}$. Therefore, $H(\rho)$ has
  maximum degree at most two and since $H(\rho)$ is acyclic, $H(\rho)$ is a
  disjoint union of paths. Moreover, the vertices
  $\{l_1,l_1',r_\ell,r_\ell'\}$ have degree exactly one and hence cannot
  be inner vertices of the paths. Finally, since every vertex of $H(\rho)$ apart
  from the vertices $\{s,t,l_1,l_1',r_\ell,r_\ell'\}$ must have degree
  exactly two, only these vertices can act as endpoints of the paths.
\end{proof}
\fi

In the following let $\rho=(X_1,\dotsc,X_\ell)$ be compatible.
              We now define the type $X$ associated
with $\rho$, which we denote by $X(\rho)$, as follows.
If $H(\rho)$ is a single cycle and $\{s,t\} \subseteq
\bigcup_{i=1}^\ell S_i$, then we set
$X(\rho)=(\psi,\emptyset,\{s,t\})$, where
$\psi(L)=\psi(R)=\emptyset$. Otherwise, 
let $\PPP(\rho)$ be the set of paths in $H(\rho)$, which \iflong due
  to Lemma~\ref{lem:DUpaths} \fi \ifshort can be shown to \fi
have their endpoints in $\{s,t,l_1,l_1',r_\ell,r_\ell'\}$.  Then,
we set $X(\rho)=(\psi,M,S)$, where $\psi$, $M$, and $S$ are defined as follows.
$M$ contains the set $\{u,v\}$ for every path in $\PPP(\rho)$ with
endpoints $u$ and $v$; for brevity, we denote $l_1$, $l_1'$, $r_\ell$, $r_\ell'$ as 
 $l$, $l'$, $r$, $r'$, respectively. Moreover, $\psi(L)=V(M)\cap \{l,l'\}$,
$\psi(R)=V(M)\cap \{r,r'\}$, and $S$ contains $s$ ($t$) if
$\sum_{i=1}^\ell\#_s(X_i)=2$ ($\sum_{i=1}^\ell\#_t(X_i)=2$).
\iflong This completes the definition of $X(\rho)$, which can be easily seen
to be a type for $b$ because $G(\rho)$ is planar due to Lemma~\ref{lem:planar}.
\fi

We say that $\rho$ is \emph{realizable} if there is an ordering
$\pi=(c_1,\dotsc,c_\ell)$ of the children in $C$ and an assignment
$\tau : C \rightarrow \XXX$ from children to types with $\tau(c) \in
\RRR(c)$ for every $c \in C$ such that
$\rho=\tau(\pi)=(\tau(c_1),\dotsc,\tau(c_\ell))$.
\iflong 
Below, we prove that if $\rho$ is a compatible and
realizable, then $b$ has type $X(\rho)$.
\begin{lemma} 
\label{lem:seq-to-type}
  Let $\rho$ be compatible and realizable. Then, $b$ has type
  $X(\rho)$.
\end{lemma}
\begin{proof}
  Let $(\pi,\tau)$ with $\pi=(c_1,\dotsc,c_\ell)$ be the ordering and
  assignment that witnesses that $\rho$ is realizable.
  Since $\tau(c)=(\psi_c,M_c,S_c) \in \RRR(c)$ and $\RRR(c)\subseteq \XXX(c)$ (using
  \enref{R1} in the definition of $\RRR(c)$) for every child $c$ of $b$, we
  obtain that there is a set $\PPP_c$ of vertex-disjoint paths
  in the complete graph with vertex set $V(\pe^*(c))$ such that:
  \begin{itemize}
  \item $\PPP_c$ consists of exactly one path $P_e$ between $u$ and $v$
    for every $e=\{u,v\} \in M_c$.
  \item $\SB \INVP(P)\SM P \in \PPP_c\SE$ is a partition of
    $(V(\pe(b))\setminus \{s,t\})\cup S_c$, where
    $\INVP(P)$ denotes the set of inner vertices of the path $P$.
  \item there is a planar drawing $D(b,\tau(c))$ of $\pe^*(b)\cup
    \bigcup_{P \in \PPP_c}P$ with outer-face $f$ such that
    $\seqb(f)=\{s,r,r',t,l',l\}$.
  \end{itemize}
  Moreover, since $\seqb(f)=\{s,r,r',t,l',l\}$, we can (and will) in the
  following assume that the drawing $D(c,\tau(c))$ has: $s$ at its
  north, $t$ at its south, $l$
  and $l'$ at its west with $l$ being north of $l'$ and $r$ and $r'$
  at its east with $r$ to the north of $r'$.

  \sloppypar
  Let $D$ be the planar drawing obtained from the disjoint union of
  the drawings $D(c_1,\tau(c_1)),\dotsc,D(c_\ell,\tau(c_\ell)$ by
  drawing them in the order given by $\rho$ from west to east without
  overlap. To avoid name clashes between vertices, we refer to the
  vertex $v \in \{s,t,l,l',r,r'\}$ belonging to the drawing
  $D(c_i,\tau(c_i))$ inside $D$ as $v_i$.

  Because of the above
  mentioned properties of the drawings $D(c,\tau(c))$, we can now add
  (and draw) the following edges to $D$ without crossings to obtain
  the planar drawing $D'$.
  \begin{itemize}
  \item The edges of the paths $(s_1,\dotsc,s_\ell)$ and
    $(t_1,\dotsc,t_\ell)$,
  \item For every $i$ with $1 \leq i < \ell$, the edges
    $\{r_i,l_{i+1}\}$ and $\{r_i',l_{i+1}'\}$.
  \end{itemize}
  Let $D''$ be the planar drawing obtained from $D'$ after:
  \begin{itemize}
  \item contracting the path $(s_1,\dotsc,s_\ell)$ into the fresh vertex
    $s$,
  \item contracting the path $(t_1,\dotsc,t_\ell)$ into the fresh vertex
    $t$,
  \item For every $i$ with $1 \leq i \leq \ell$:
    \begin{itemize}
    \item if $i \neq \ell$ and $|\psi_i(R)|=1$ contract the edge $\{r_i,r_i'\}$
      into the vertex $r_i$,
    \item if $i \neq \ell$ and $|\psi_i(R)|=0$ remove the vertices
      $r_i$ and $r_i'$,
    \item if $i \neq 1$ and $|\psi_i(L)|=1$ contract the edge $\{l_i,l_i'\}$
      into the vertex $l_i$,
    \item if $i \neq 1$ and $|\psi_i(L)|=0$ remove the vertices
      $l_i$ and $l_i'$.
    \end{itemize}
  \item removing all edges of the form $sl_i$, $l_il_i'$, $l_i't$,
    $sr_i$, $r_ir_i'$, and $r_i't$ for every $i \notin\{1,\ell\}$.
  \end{itemize}

  Let $D_b$ be the planar drawing obtained from $D''$ after:
  \begin{itemize}
  \item contracting all edges incident to any vertex in $\{ l_i, l_i'
    \mid 1 < i \leq \ell \}\cup \{r_i, r_i' \mid 1 \leq i < \ell \}$,
  \item renaming the vertices $l_1$, $l_1'$, $r_\ell$, and $r_\ell'$
    to $l$, $l'$, $r$, and $r'$.
  \end{itemize}
  Let $H_b$ be the planar graph corresponding to $D_b$.

  Let $D_\rho$ be the planar drawing obtained from $D''$ after:
  \begin{itemize}
  \item contracting every path $P \in \PPP_c$ for every $c \in C$ into
    a single edge,
  \item if $|\psi_1(L)|=1$ contract the edge $\{l_1,l_1'\}$
    into the vertex $l_1$,
  \item if $|\psi_1(L)|=0$ remove the vertices $l_1$ and $l_1'$,
  \item if $|\psi_\ell(R)|=1$ contract the edge $\{r_\ell,r_\ell'\}$
    into the vertex $r_\ell$,
  \item if $|\psi_\ell(R)|=0$ remove the vertices $r_\ell$ and $r_\ell'$,    
  \end{itemize}
  Let $H_\rho$ be the planar graph corresponding to $D_\rho$.

  Then, $H_\rho$ is isomorphic to $G(\rho)$ and $D_b$ is a planar
  drawing of $\pe^*(b)\cup
  \bigcup_{P \in \bigcup_{c \in C}\PPP_c}$ that witnesses that $b$ has
  type $X(\rho)$.
\end{proof}

\begin{lemma} 
\label{lem:type-to-seq}
  Let $W=(D,D_H,G_H,H)$ be a witness for $G$ that \emph{respects}
  $\mathcal{B}$ and $\mathcal{T}$. Then, there is a realizable and compatible $\rho$
  such that $X(\rho)=\type_W(b)$.
\end{lemma}
\begin{proof}
  Let $\pi=(c_1,\dotsc,c_\ell)$ be the ordering of the children in $C$
  according to the drawing $D_H$. Moreover, let $\tau : C \rightarrow
  \XXX$ be the assignment defined by setting
  $\tau(c)=\type_W(c)$ for every $c \in C$. Note that because of \enref{R2} in the definition of
  $\RRR(c)$ it holds that $\tau(c) \in \RRR(c)$ for every $c \in
  C$. Finally, let $\rho$ be the sequence
  $(\tau(c_1),\dotsc,\tau(c_\ell))$. Then, $\rho$ is clearly
  realizable. Moreover, $\rho$ is compatible since $W$ respects
  $\mathcal{B}$ and since $H$ is a Hamiltonian cycle.
\end{proof}

From~\Cref{lem:seq-to-type,lem:type-to-seq}, we now obtain the
following corollary.
\fi
\ifshort The following lemma now allows us to focus on finding the set of all types $X$
for which there is a compatible and realizable $\rho$ such that $X=X(\rho)$.
\begin{lemma} \fi
\iflong\begin{corollary}\fi\label{cor:pn-record-char}
  The set $R$ containing every type $X \in \XXX$ such that there
  is a compatible and realizable $\rho$ with $X=X(\rho)$ satisfies the properties
  \enref{R1} and \enref{R2}.
\ifshort\end{lemma}\fi\iflong\end{corollary}\fi
\iflong Therefore, from now onward we can focus on finding the set of all types $X$
for which there is a compatible and realizable $\rho$ such that $X=X(\rho)$.\fi
We will now show that this can be achieved very efficiently because only a
constant number, i.e., at most $8$ types (and their ordering) need to be specified
in order to infer the type of a sequence $\rho$. 
Let $X=(\psi,M,S) \in \XXX$ be a type. We say that $X$ is \emph{dirty}
if $\#_s(X)+\#_t(X)>0$ and otherwise we say that $X$ is \emph{clean}.
We say that $X$ is \emph{$0$-good}, \emph{$1$-good},
and \emph{$2$-good}, if $X$ is clean and additionally
$M=\emptyset$, $M=\{\{l,r\}\}$, and
$M=\{\{l,r\},\{l',r'\}\}$, respectively. We say that $X$ is \emph{good}
if it is $x$-good for some $x \in \{0,1,2\}$ and otherwise we say that $X$ is
\emph{bad}. We denote by $\XXX_G$ and $\XXX_B$ the subset of $\XXX$ consisting
only of the good respectively bad types.
An illustration of these notions is provided in Figure~\ref{fig:shtd}.
              \iflong
\begin{lemma} 
\label{lem:atmostfour}
  Let $\rho=(X_1,\dotsc,X_\ell)$ be compatible. Then, $\rho$ contains
  at most $4$ dirty types and at most $4$ types that are clean and
  bad.
\end{lemma}
\begin{proof}
  Let $\rho=(X_1,\dotsc,X_\ell)$ with $X_i=(\psi_i,M_i,S_i)$
  be compatible.
  The statement that $\rho$ contains at most $4$ dirty types follows
  directly from \enref{C2} in the definition of weak compatibility. It remains
  to show that $\rho$ contains at most $4$ types that are clean and bad.
  First note that if type $X=(\psi,M,S)$ is clean and bad, then
  either $M=\{\{L,L\}\}$, $M=\{\{R,R\}\}$, or $M=\{\{L,L\},\{R,R\}\}$.
  Now suppose for a contradiction that $\rho$
  contains at least $5$ types that are clean and bad. Then, there are
  indices $1 \leq i<j<k \leq \ell$ such that $M_i$, $M_j$,
  and $M_k$ either all contain the pair $\{L,L\}$ or all of them
  contain the pair $\{R,R\}$. Let us assume the former case since the
  argument for the latter case is analogous. Because $\rho$ is
  compatible the path/cycle $P$ in $H(\rho)$ that contains the edge
  $\{l_j,l_j'\}$ must also contain $s$ and $t$. This is because if $P$
  does not contain $s$ and $t$ it can only go to the left until it is
  blocked by the path containing the edge $\{l_i,l_i'\}$.
  The same holds for the path/cycle in $H(\rho)$ that contains the edge
  $\{l_k,l_k'\}$. Therefore, $\{l_j,l_j'\}$ and $\{l_k,l_k'\}$ are
  contained together with $s$ and $t$ on a cycle $C$ in
  $H(\rho)$. Finally, because of \enref{C2} in the definition of weakly
  compatible, we obtain that $C$ does not contain the edge
  $\{l_i,l_i'\}$, which contradicts our assumption that $H(\rho)$ is
  either a single cycle or acyclic.
\end{proof}
The following corollary follows immediately from Lemma~\ref{lem:atmostfour} since
every bad type is either clean or dirty.
\fi
\ifshort\begin{lemma}\fi
\iflong\begin{corollary}\fi\label{cor:atmosteight}
  Let $\rho=(X_1,\dotsc,X_\ell)$ be compatible, then $\rho$
  contains at most $8$ bad types.
\ifshort\end{lemma}\fi\iflong\end{corollary}\fi

\iflong
Moreover, since deciding whether $\rho$ is compatible merely requires us to check that $\rho$ satisfies \enref{C1} and
\enref{C2} and that either $H(\rho)$ is acyclic or
  $H(\rho)-(\bigcup_{i=1}^\ell S_i)$ is a single
  (Hamiltonian) cycle, we observe:

\begin{observation}\label{obs:checkcomp}
  It is possible to decide whether a given $\rho=(X_1,\dotsc,X_\ell)$ is compatible in time $\bigoh(\ell)$.
\end{observation}
\fi

Next, we will show that any compatible sequence contains at
most $8$ bad types and that the type $X(\rho)$ is already determined
by looking only at the sequence of bad types that occur in
$\rho$. This will then allow us to simulate the enumeration of all possible sequences,
by enumerating merely all sequences of at most $8$ bad types.

        We say that a sequence $\rho'$ is an extension of $\rho$ if $\rho$ is
a (not necessarily consecutive) sub-sequence of $\rho'$.
We call a compatible sequence $\rho$ \emph{$(X,i)$-extendable} for
some $X \in \XXX$ and integer $i$, if there is a compatible extension
$\rho'$ of $\rho$
such that $\rho'$ is obtained by adding $i$ elements of type $X$ to
$\rho$ and $X(\rho)=X(\rho')$. We call $\rho$
\emph{$X$-extendable} if $\rho$ is $(X,i)$-extendable for any integer
$i$. We say that $\rho'$ is an \emph{$(X,i)$-extension} of $\rho$ if
$\rho'$ is a compatible sequence obtained after adding $i$ elements of
type $X$ to $\rho$ and $X(\rho)=X(\rho')$.

\begin{lemma}
\label{lem:addgood}
  Let $\rho=(X_1,\dotsc,X_\ell)$ with $X_i=(\psi_i,M_i,S_i)$ and $X
  \in \XXX_G$. Then, $\rho$
  is $(X,1)$-extendable if and only if $\rho$ is $X$-extendable.
  Moreover, deciding whether $\rho$ is $(X,1)$-extendable
  and if so computing an $(X,i)$-extension $\rho'$ of $\rho$
  can be achieved in time $\bigoh(\ell+i)$ for every integer $i$.
\end{lemma}
\iflong \begin{proof}
  The first statement of the lemma follows because if $\rho'$ is a
  compatible extension of $\rho$ containing at least
  one $x$-good type $X_i$, then we can add another $x$-good type
  immediately after or before $X_i$ without violating the
  compatibility. Moreover, we can also delete $X_i$ from $\rho'$
  without violating its compatibility.

  Moreover, it is straightforward to verify that $\rho$
  can be extended by $1$ $x$-good type if and only if either: (1)
  $|\psi_1(L)|=x$, (2) $|\psi_\ell(R)|=x$, or there is an index $i$ with $1 \leq i
  < \ell$ such that $|\psi_i(L)|=|\psi_{i+1}(R)|=x$. This can clearly be tested
  in time $\bigoh(\ell)$ and if the test succeeds, it is also easy to compute an $(X,i)$-extension by
  adding all $i$ elements of type $X$ in one of the possible positions.
\end{proof}\fi
 \begin{lemma}
\label{lem:badsubseq}
  Let $\rho$ be a compatible sequence and let $\rho'$ be the
  sub-sequence of $\rho$ consisting only of the bad types in $\rho$.
  Then, $\rho'$ is compatible and $X(\rho)=X(\rho')$.
\end{lemma}
\iflong \begin{proof}
  The lemma holds because removing any good type $X$ preserves
  compatibility and does not change the type of the sequence; this
  is because neither $s$ nor $t$ are used by $X$ and moreover $\#_L(X)=\#_R(X)$.
\end{proof}\fi

At this point, we are ready to describe the algorithm we will use to compute $\RRR(b)$ (and argue its correctness).
The algorithm first enumerates all
possible compatible sequences $\rho$ of at most $8$ bad types, i.e.,
$\rho=(Y_1,\dotsc,Y_r)$ with $r\leq 8$ and $Y_i \in \XXX_B$ for every
$i$. Note that there are at most $(|\XXX_B|+1)^{8}$ (and therefore
constantly many) such sequences and those can be enumerated in
constant time. Given one such sequence $\rho=(Y_1,\dotsc,Y_r)$, the
algorithm then tests whether the sequence can be realized given the
types available for the children in $C$ as follows. It first uses Lemma~\ref{lem:addgood} to test whether 
$\rho$ allows for adding a $0$-good, $1$-good or $2$-good type in constant time.
Let $A_\rho \subseteq \XXX_G$ be the set of all good types that
can be added to $\rho$ and let $C_\rho$
be the subset of $C$ containing all children $c$ such
that $A_\rho\cap\RRR(c)\neq \emptyset$.

Consider the
following bipartite graph $Q_\rho$ having one vertex $y_i$ for every
$i$ with $1 \leq i \leq r$ representing the type $Y_i$ 
on one side and one vertex $v_c$ for every $c \in C$
representing the child $c$ on the other side of the
bipartition. Moreover, $Q_\rho$ has an edge between $y_i$ and $v_c$ if
$Y_i \in \RRR(c)$. We claim that $\rho$ can be extended to a
compatible and realizable sequence if and only if $Q_\rho$ has a
matching that saturates $\{y_1,\dotsc,y_r\}\cup
\SB v_c \SM c \in C\setminus C_\rho\SE$. This problem can be 
solved using a simple reduction to the well-known maximum flow problem\iflong as shown by the following lemma\fi.
\iflong
\begin{lemma} 
    \label{lem:max-match}
    Let $Q$ be a bipartite graph with partition $\{A,B\}$ and let $V
    \subseteq B$. There is an algorithm that in time
    $\bigoh(|E(Q)||A|)$ decides whether $Q$ has a matching that
    saturates $A\cup V$.
\end{lemma}
\begin{proof}
    We solve the problem using a reduction to the maximum flow
    problem, which can be solved in time $\bigoh(mU)$ for a flow
    network with $m$ edges with every edge having integer capacity at
    most $U$~\cite{ford_fulkerson_1956}.
    Let $N$ be the
    network obtained as follows. The vertices of $N$ are new vertices
    $s$, $t$, $t'$ plus the vertices of $Q$. Moreover, $N$ contains
    the following arcs:
    \begin{itemize}
    \item an arc from $s$ to $a$ with capacity $1$ for every $a \in
      A$,
    \item an arc from $a$ to $b$ with capacity $1$ for every edge
      $\{a,b\} \in E(Q)$,
    \item an arc from $v$ to $t$ with capacity $1$ for every $v \in
      V$,
    \item an arc from $t'$ to $t$ with capacity $|A|-|V|$,
    \item an arc from $b$ to $t'$ with capacity $1$ for every $b \in
      B\setminus V$.
    \end{itemize}
    It is now straightforward to show that $Q$ has a matching that
    saturates $A\cup V$ if and only if $N$ has an integer flow from
    $s$ to $t$ with value $|A|$. Since $U\leq |A|$ and
    $m=\bigoh(|V(Q)|+|E(Q)|)$, we obtain that our problem can be
    decided in time $\bigoh(|E(Q)||A|)$, which shows the
    stated run-time.
                                       \end{proof}
\fi
The following lemma now establishes the correctness (i.e., the soundness and
completeness) of the algorithm.
\begin{lemma}
\label{lem:pn-correctness}
  Let $X \in \XXX$. Then, there is a compatible and realizable
  sequence $\rho$ with $X=X(\rho)$ if and only if there is a compatible
  sequence $\rho=(Y_1,\dotsc,Y_r)$ of bad types with $r\leq 8$ with $X=X(\rho)$
  such that the bipartite graph $H_\rho$ has a matching that saturates
  $\{y_1,\dotsc,y_r\}\cup \SB v_c \SM c \in C\setminus C_\rho\SE$.
\end{lemma}
\iflong \begin{proof}
  Towards showing the forward direction, let $\rho$ be a compatible
  and realizable sequence and let $\rho'=(Y_1,\dotsc,Y_r)$ be the sub-sequence of
  $\rho$ containing only the bad types in $\rho$.
  Because of~\Cref{cor:atmosteight}, it holds that $r \leq 8$ and
  because of~\Cref{lem:badsubseq}, we have that $\rho'$ is compatible.
  It remains to show that $H_\rho$ has a matching that saturates
  $\{y_1,\dotsc,y_r\}\cup (C\setminus C_\rho)$. Let $(\pi,\tau)$ be
  the ordering and assignment that witnesses that $\rho$ is
  realizable. Let $\pi'=(c_1,\dotsc,c_r)$ be the subsequence of $\pi$
  containing only the children $c$ such that $\tau(c)$ is bad; note that
  $\rho'=\tau(\pi)=(\tau(c_1),\dotsc,\tau(c_r))$.
  Then, $M=\{ \{y_i,v_{c_i}\} \mid 1 \leq i \leq r \}$ is a
  matching in $H_\rho$ that saturates $\{y_1,\dotsc,y_r\}\cup (C\setminus
  C_\rho)$.

  Towards showing the reverse direction, let $\rho'=(Y_1,\dotsc,Y_r)$
  be a compatible sequence of bad types with $r \leq 8$ and let 
  $M$ be the matching in $H_\rho$ that saturates $\{y_1,\dotsc,y_r\}\cup (C\setminus
  C_\rho)$. For convenience, we represent $M$ as the bijective
  function $\tau' : C' \rightarrow \{Y_1,\dotsc,Y_r\}$, where
  $C'=V(M)\cap C$ and such that $M=\{\{c,\tau'(c)\} \mid c \in C'\}$.
  Let $\tau : C \rightarrow \XXX$ be the assignment of children to types given by
  $\tau(c)=\tau'(c)$ for every $c \in C'$ and $\tau(c) \in
  (A_\rho\cap \RRR(c))$ for every $c \in C\setminus C'$. Because
  $C\setminus C' \subseteq C_\rho$, it holds that $A_\rho\cap
  \RRR(c)\neq \emptyset$ for every $c  \in C\setminus C'$ and
  therefore it is possible to assign $\tau(c)$. We can now
  use~\Cref{lem:addgood} to obtain a compatible extension $\rho$ of
  $\rho'$ with $X(\rho')=X(\rho)$ such that $\rho$ is obtained from
  $\rho'$ after adding $|\tau^{-1}(X)|$ elements of type $X$ for every
  $X \in A_{\rho'}$. Due to the choice of $\rho$ and $\tau$, there
  now exists an ordering $\pi$ of $C$ such that $\rho=\tau(\pi)$,
  which together with $\tau$ shows that $\rho$ is also realizable.
\end{proof}\fi

\iflong
We are now ready to prove the central lemma of this section.
\begin{proof}[Proof of \Cref{lem:Pnode}]
  We need to show that given $\RRR(c)$ for
  every child $c \in C$, we can compute
  $\RRR(b)$ in time $\bigoh(\ell)$.

  Because of~\Cref{cor:pn-record-char} we can compute $\RRR(b)$ by
  computing all types $X \in \XXX$ such that there is a compatible and realizable
  sequence $\rho$ of types with $X=X(\rho)$. Moreover, because
  of~\Cref{lem:pn-correctness}, a type $X \in \XXX$ has such a
  compatible and realizable sequence $\rho$ if and only if there is a
  compatible sequence $\rho=(Y_1,\dotsc,Y_r)$ of bad types with $r
  \leq 8$ and $X=X(\rho)$ such that the bipartite graph $H_\rho$ has a
  matching that saturates $\{y_1,\dotsc,y_r\} \cup \SB v_c \SM c \in
  C\setminus C_\rho\SE$ and this can be achieved by the following algorithm. 

  The algorithm first enumerates all
  possible compatible sequences $\rho$ of at most $8$ bad types, i.e.,
  $\rho=(Y_1,\dotsc,Y_r)$ with $r\leq 8$ and $Y_i \in \XXX_B$ for every
  $i$. Note that there are at most $(|\XXX_B|+1)^{8}$ (and therefore
  constantly many) sequences of at most $8$ bad types and because
  of~\Cref{obs:checkcomp} checking
  whether such a sequence is compatible can be achieved in constant
  time. Therefore, all such compatible sequences can be enumerated in
  constant time. Given one such sequence $\rho=(Y_1,\dotsc,Y_r)$,
  the algorithm first uses~\Cref{lem:addgood} to compute the set 
  $A_\rho \subseteq \XXX_G$ (i.e., the set of all good types that
  can be added to $\rho$) in constant time. It then computes
  $C_\rho$
  (i.e., the subset of $C$ containing all children $c$ such
  that $A_\rho\cap\RRR(c)\neq \emptyset$) and constructs the bipartite
  graph $H_\ell$ in time $\bigoh(\ell)$. Finally, it uses~\Cref{lem:max-match} to decide whether $H_\rho$ has a matching
  that saturates $\{y_1,\dotsc,y_r\}\cup \SB v_c \SM c \in C \setminus
  C_\rho \SE$. If so, the algorithm correctly adds the type $X(\rho)$ to  $\RRR(b)$ 
  and otherwise the algorithm continues with the next
  sequence $\rho=(Y_1,\dotsc,Y_r)$ with $r\leq 8$ and $Y_i \in \XXX_B$.
  
  As pointed out above, the correctness of the algorithm follows
  from~\Cref{cor:pn-record-char,lem:pn-correctness,lem:addgood}.
  The total run-time of the algorithm is dominated by the time
  required to decide whether $H_\rho$ has a matching saturating 
  $\{y_1,\dotsc,y_r\}\cup \SB v_c \SM c \in C \setminus
  C_\rho \SE$, which because of~\Cref{lem:max-match} can be achieved
  in time $\bigoh(|E(H_\rho)|r)$. Since  $|E(H_\rho)|\leq 8\ell$ and
  $r\leq 8$ this term is equal to $\bigoh(\ell)$, as claimed. 
\end{proof}
\fi

 \subsection{Handling R-nodes and S-nodes}\label{ssec:rsnodes}
    \ifshort
Here, we will show how to compute a set of types satisfying \enref{R1} and
\enref{R2} for every R-node and S-node of $\mathcal{B}$. To achieve this we
will again use a dynamic programming algorithm albeit on a sphere-cut
decomposition of $\sk(b)$ instead of on the \SPQR{}-tree. The aim
of this subsection is therefore to show the following lemma.
\begin{lemma}
\label{lem:Rnode}
  Let $b$ be an R-node or S-node of $\mathcal{B}$ such that $\RRR(c)$
  has already been computed  for every child
  $c$ of $b$ in $\mathcal{B}$. Then, we can compute $\RRR(b)$
  in time $\bigoh((84\sqrt{14})^{\omega}\omega\ell+\ell^3)$,
  where $\omega$ is the branchwidth of the graph $\sk(b)$ and $\ell$ is the number
  of children of $b$ in $\mathcal{B}$.
\end{lemma}
In the following, let $b$ be an R-node or S-node of $\mathcal{B}$ with reference
edge $(s_b,t_b)$ and let $\langle T_b, \lambda_b, \Pi_b \rangle$ be a
sphere-cut decomposition of $\sk(b)$ that is rooted in
$r=\lambda^{-1}_b((s_{b},t_{b}))$. For a weak noose $O \subseteq C(T_b)$, let $\ARC(O)$
be the set of all types of $O$
satisfying the following two natural analogs of \enref{R1} and \enref{R2}, i.e.:
\ifshort
\textbf{(RO1)} if $X \in \ARC(O)$, then $O$ has type $X$, and \textbf{(RO2)} if there is a witness $(D,D_H,G_H,H)$ for $G$ that   \emph{respects} $\mathcal{B}$ such that $\type_W(b,O)=X$\ifshort, where $\type_W(b,O)$ is defined analogously to $\type_W(b)$ for the graph $\pe(b,O)$\fi, then $X \in \ARC(O)$. 
\fi

\iflong
\begin{enumerate}[(RO1)]
\item If $X \in \ARC(O)$, then $O$ has type $X$.
\item If there is a witness $(D,D_H,G_H,H)$ for $G$ that
  \emph{respects} $\mathcal{B}$ such that $\type_W(b,O)=X$\ifshort, where $\type_W(b,O)$ is defined analogously to $\type_W(b)$ for the graph $\pe(b,O)$\fi, then $X \in \ARC(O)$. 
\end{enumerate}
\fi

Our aim is to compute $\ARC(O_{a^r})$ for the arc $a^r$ incident to
the root $r$ of $T_b$. This is achieved by
computing $\ARC(O_a)$ for every inner arc $a$ of $T_b$ via a bottom-up dynamic
programming algorithm along $T_b$; after initially calculating
$\ARC(O_a)$ from $\RRR(c)$ for every leaf-arc $a$ corresponding to the
child $c$ of $b$. 
Employing our framework introduced in
Subsection~\ref{ssec:SPCD-framework}, we only have to show how to
compute $\ARC(O_1\xor{} O_2)$ from $\ARC(O_1)$ and $\ARC(O_2)$ for
any weak nooses $O_1$ and $O_2$.

Let $O_1$ and $O_2$ be two weak nooses having type $X_1=(\psi_1, M_1, S_1)$ 
and type $X_2=(\psi_2, M_2, S_2)$, respectively. We say that $X_1$ and
$X_2$ are \emph{compatible} if
\begin{enumerate}[(1)]
    \item\label{comp1} $O=O_1 \xor O_2$ is a weak noose,
    \item\label{comp2} the inside region of the noose $O$ contains all subcurves in $(O_1\cap O_2)$,
    \item\label{comp3} $\forall {c \in O_1 \cap O_2}$, it holds $\psi_1(c) = \psi_2(c)$,
    \item\label{comp4} for every $u \in \m{O_1 \cap O_2} \setminus  \m{O_1 \xor O_2}$, it holds that $u$ is only in one of following sets: $S_1$, $S_2$ or $V(M_1)\cap V(M_2)$, and
    \item\label{comp5} the multi-graph obtained from the union of $M_1$ and $M_2$
      is acyclic, or is one cycle and $\m{O}
        \subseteq S_1\cup S_2 \cup (V(M_1) \cap V(M_2))$,
    \item\label{comp6} if $X_1$ is the \fulltype{}, then $X_2$ is the \emptytype{}
      and $\m{O_2} \subseteq \m{O_1}$, and vice versa.
    \end{enumerate}

We denote by $X_1 \concat X_2$
the \emph{combined type} $X=(\psi,M, S)$ of $X_1=(\psi_1,M_1, S_1)$ and $X_2=(\psi_2,M_2, S_2)$ for the
weak noose $O=O_1 \xor O_2$ that is defined as follows and also illustrated in
Figure~\ref{fig:typesSPCD}. 
For each $c \in O$, if $ c \in O_1$ then $\psi(c)$ is equal to
$\psi_1(c)$, otherwise $\psi(c)$ is equal to $\psi_2(c)$ 
and the set $S$ is equal to $(S_1 \cup S_2 \cup (V(M_1)\cap V(M_2)))\cap
\m{O}$, i.e., any vertex with degree two w.r.t. $X$ must be in $\m{O}$
and have degree two already w.r.t. $X_1$ or $X_2$,
or it must be in both matchings $M_1$ and $M_2$.
If either $X_1$ or $X_2$ is a \fulltype, then by \enref{6} we get that 
$M_1 = M_2 = M = \emptyset$ and $X_1\concat X_2$ is the \fulltype{}.
If the multi-graph $M_1 \cup M_2$ is one cycle, then by \enref{5} we get that 
$M = \emptyset$ and $X_1\concat X_2$ is the \fulltype{}.
Otherwise, due to \enref{5}, the multi-graph $M_1 \cup
M_2$ is acyclic and corresponds to a set of paths.
Therefore, the matching $M$ is the set
containing the two endpoints for every path in $M_1 \cup M_2$.

\begin{observation}\label{obs:combined_type_time}
  Let $X_1$ and $X_2$ be two types defined on the weak nooses $O_1$
  and $O_2$, respectively. Then, we can check whether $X_1$ and $X_2$
  are compatible and if so compute the type $X_1 \concat X_2$ in time
  $\bigoh(|O_1| + |O_2|)$.
\end{observation}

To show the correctness of our approach it now remains to show that:
(1) if there is a witness $W$ for $G$ that
respects $\BBB$, then for every two weak nooses $O_1$ and $O_2$ it
holds that $\type_W(b,O_1)$ and $\type_W(b,O_2)$ are compatible types
and $\type_W(b,O)=\type_W(b,O_1) \concat \type_W(b,O_2)$ and (2) if $O_1$ and $O_2$
have compatible types $X_1$ and $X_2$, then $O=O_1\xor O_1$ has type
$X_1 \concat X_2$. 
\fi

\begin{figure}

\begin{minipage}[c]{0.5\textwidth}

\begin{tikzpicture}[scale=0.92]
\coordinate (M1) at ($(0,0)$);
\coordinate (M2) at ($(0.65,0)$);
\coordinate (M3) at ($(1.3,0)$);
\coordinate (M4) at ($(1.73,0)$);
\coordinate (M5) at ($(2.17,0)$);
\coordinate (M6) at ($(2.6,0)$);
\coordinate (M7) at ($(3.9,0)$);
\coordinate (M8) at ($(4.55,0)$);
\coordinate (M9) at ($(5.2,0)$);
\coordinate (U1) at ($(0, 2)$);
\coordinate (U2) at ($(1.3, 2.8)$);
\coordinate (U3) at ($(2.6, 3)$);
\coordinate (U4) at ($(3.9, 2.8)$);
\coordinate (U5) at ($(5.2, 2)$);
\coordinate (U6) at ($(5.6,1)$);
\coordinate (D1) at ($(0, -2)$);
\coordinate (D2) at ($(2.6, -3)$);
\coordinate (D3) at ($(3.46, -2.9)$);
\coordinate (D4) at ($(4.32, -2.65)$);
\coordinate (D5) at ($(5.2, -2)$);
\coordinate (D6) at ($(5.6, -1)$);

\draw[red!30!white, very thick] (M1) arc (210:-30:3cm and 2cm) 
(M1) arc (150:390:3cm and 2cm)
(M1)--(M9);
\draw[very thick,green!50!black]
(M2) to[out=45,in=135] (M4)(M7) to[out=45,in=135] (M8) (M5)to[out=135,in=270](U2)  (U4)to[out=270,in=135](M9) (U5)to[out=270,in=180](U6)
;
\draw[very thick,red]
(D1)to[out=0,in=270](M2)(M4) to[out=270,in=270] (M8)(M7) to[out=270,in=270] (M5)(M9) to[out=225,in=180] (D6) (D3) to[out=90,in=90](D4)
;

\draw
(U1) circle (3pt)
(U3) circle (3pt)
 ;

\draw[fill=white]
(M1) circle (3pt)
(U3) circle (3pt)
(D2) circle (3pt)
;

\draw[fill=green!50!black]
(U1) circle (3pt)
(M6) circle (3pt)
;

\draw[fill=red]
(M3) circle (3pt)
(D5) circle (3pt)
;

\draw[fill=black]
(U5) circle (3pt)
(M9) circle (3pt)
(M7) circle (3pt)
(D1) circle (3pt)
;

\draw (M2) node[very thick,cross=4pt] {};
\draw (M4) node[very thick,cross=4pt] {};
\draw (M5) node[very thick,cross=4pt] {};
\draw (M8) node[very thick,cross=4pt] {};
\draw (U2) node[very thick,cross=4pt] {};
\draw (U4) node[very thick,cross=4pt] {};
\draw (U6) node[very thick,cross=4pt] {};
\draw (D3) node[very thick,cross=4pt] {};
\draw (D4) node[very thick,cross=4pt] {};
\draw (D6) node[very thick,cross=4pt] {};

\node at (2.6,1.5) {$O_1$};
\node at (2.6,-1.5) {$O_2$};
\node[left] at (M1) {$u_1$};
\node[above] at (M2) {$x_7$};
\node[below] at (M3) {$u_{9}$};
\node[above] at (M4) {$x_8$};
\node[above] at (M5) {$x_9$};
\node[below] at (M6) {$u_{10}$};
\node[above] at (M7) {$u_{11}$};
\node[above] at (M8) {$x_{10}$};
\node[right] at (M9) {$u_5$};
\node[left] at (U1) {$u_2$};
\node[above] at (U2) {$x_1$};
\node[above ] at (U3) {$u_3$};
\node[above ] at (U4) {$x_2$};
\node[above ] at (U5) {$u_4$};
\node[right] at (U6) {$x_3$};
\node[left] at (D1) {$u_8$};
\node[below] at (D2) {$u_7$};
\node[below] at (D3) {$x_6$};
\node[below] at (D4) {$x_5$};
\node[below] at (D5) {$u_6$};
\node[right] at (D6) {$x_4$};
\end{tikzpicture}
\end{minipage} \begin{minipage}[c]{0.5\textwidth}
\begin{tikzpicture}[scale=0.92]
\coordinate (M1) at ($(0,0)$);
\coordinate (M2) at ($(0.65,0)$);
\coordinate (M3) at ($(1.3,0)$);
\coordinate (M4) at ($(1.73,0)$);
\coordinate (M5) at ($(2.17,0)$);
\coordinate (M6) at ($(2.6,0)$);
\coordinate (M7) at ($(3.9,0)$);
\coordinate (M8) at ($(4.55,0)$);
\coordinate (M9) at ($(5.2,0)$);
\coordinate (U1) at ($(0, 2)$);
\coordinate (U2) at ($(1.3, 2.8)$);
\coordinate (U3) at ($(2.6, 3)$);
\coordinate (U4) at ($(3.9, 2.8)$);
\coordinate (U5) at ($(5.2, 2)$);
\coordinate (U6) at ($(5.6,1)$);
\coordinate (D1) at ($(0, -2)$);
\coordinate (D2) at ($(2.6, -3)$);
\coordinate (D3) at ($(3.46, -2.9)$);
\coordinate (D4) at ($(4.32, -2.65)$);
\coordinate (D5) at ($(5.2, -2)$);
\coordinate (D6) at ($(5.6, -1)$);

\draw[red!30!white, very thick] (M1) arc (210:-30:3cm and 2cm) 
(M1) arc (150:390:3cm and 2cm)
;
\draw[very thick,blue]
 (D1)to[out=0,in=270](U2)  (U4)to[out=270,in=180](D6) (U5)to[out=270,in=180](U6) (D3) to[out=90,in=90] (D4)
;

\draw[fill=white]
(M1) circle (3pt)
(U3) circle (3pt)
(D2) circle (3pt)
;

\draw[fill=blue]
(U1) circle (3pt)
(D5) circle (3pt)
(M9) circle (3pt)
 ;

\draw[fill=black]
(U5) circle (3pt)
(D1) circle (3pt)
;

    \draw (U2) node[very thick,cross=4pt] {};
\draw (U4) node[very thick,cross=4pt] {};
\draw (U6) node[very thick,cross=4pt] {};
\draw (D3) node[very thick,cross=4pt] {};
\draw (D4) node[very thick,cross=4pt] {};
\draw (D6) node[very thick,cross=4pt] {};

\node at (2.6,0) {$O$};
\node[left] at (M1) {$u_1$};
       \node[right] at (M9) {$u_5$};
\node[left] at (U1) {$u_2$};
\node[above] at (U2) {$x_1$};
\node[above ] at (U3) {$u_3$};
\node[above ] at (U4) {$x_2$};
\node[above ] at (U5) {$u_4$};
\node[right] at (U6) {$x_3$};
\node[left] at (D1) {$u_8$};
\node[below] at (D2) {$u_7$};
\node[below] at (D3) {$x_6$};
\node[below] at (D4) {$x_5$};
\node[below] at (D5) {$u_6$};
\node[right] at (D6) {$x_4$};
\end{tikzpicture}

\end{minipage}

\caption{\iflong\textbf{(Left)}\fi An illustration of combining two compatible
  types $X_1=(\psi_1,M_1,S_1)$ and $X_2=(\psi_2,M_2,S_2)$ for two
  weak nooses $O_1$ and $O_2$ into the combined type
  $X=(\psi,M,S)=X_1\concat X_2$ for
  $O=O_1\xor{}O_2$. Vertices of the graph are represented as circles
  and vertices subdividing the nooses, i.e., vertices in
  $V(\psi_1)\cup V(\psi_2)$, are represented as
  crosses. Black vertices are the vertices that are within a matching,
  i.e., the vertices in $V(M_1)\cup V(M_2)$, green (red) vertices are the
  vertices in $S_1$ ($S_2$) and all other vertices of the graph are
  white. \iflong The following holds for the type $X_1$ and $X_2$:
  $V(\psi_1)=\{x_1,x_2,x_3,x_7,x_8,x_9,x_{10}\}$, $V(\psi_2)=\{x_4,x_5,x_6,x_7,x_8,x_9,x_{10}\}$, $M_1=\{\{x_1,x_9\},\{x_2,u_5\},\{u_4,x_3\},\{x_{10},u_{11}\},\{x_8,x_7\}\}$,
    $M_2=\{\{x_7,u_8\},\{x_8,x_{10}\},\{x_9,u_{11}\},\{u_{5},x_{4}\},\{x_5,x_6\}\}$, $S_1=\{u_2,u_{10}\}$, and $S_2=\{u_9,u_6\}$.
  \textbf{(Right)}
  The resulting type $X$ of $O$ for which the following holds: $V(\psi)=\{x_1,\dotsc,x_6\}$,
  $M=\{\{x_1,u_8\},\{x_2,x_4\},\{u_4,x_3\},\{x_5,x_6\}\}$, and
  $S=\{u_2,u_5,u_6\}$.\fi}
\label{fig:typesSPCD}
\end{figure}
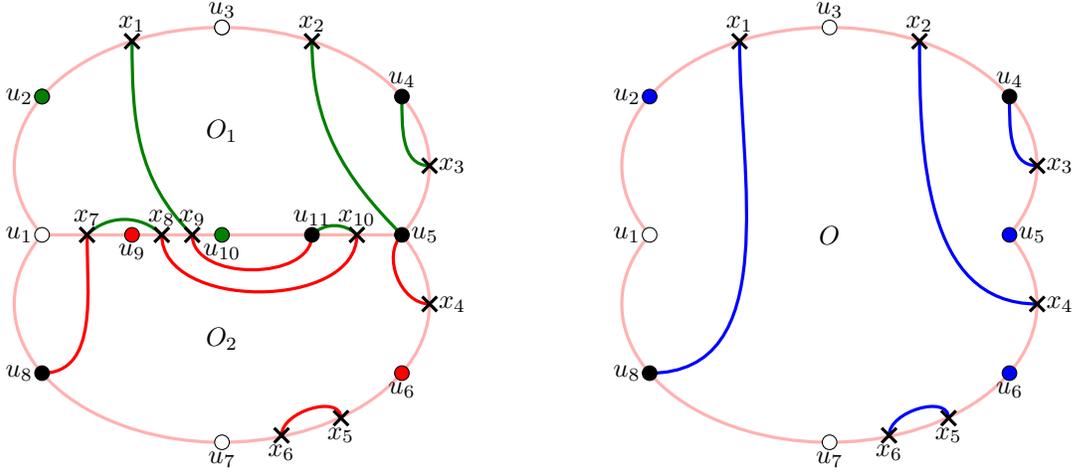

\iflong
Here, we will show how to compute the set of types satisfying \enref{R1} and
\enref{R2} for every R-node and S-node of $\mathcal{B}$. To achieve this we
will again use a dynamic programming algorithm albeit on a sphere-cut
decomposition of $\sk(b)$ instead of on the \SPQR{}-tree. The aim
of this subsection is therefore to show the following lemma.
\begin{lemma}
\label{lem:Rnode}
  Let $b$ be an R-node or S-node of $\mathcal{B}$ such that $\RRR(c)$
  have already been computed  for every child
  $c$ of $b$ in $\mathcal{B}$. Then, we can compute $\RRR(b)$
  in time $\bigoh(315^{\omega}\ell+\ell^3)$,
  where $\omega$ is the branchwidth of the graph $\sk(b)$ and $\ell$ is the number
  of children of $b$ in $\mathcal{B}$.
\end{lemma}
In the following, let $b$ be an R-node or S-node of $\mathcal{B}$ with reference
edge $(s_b,t_b)$ and let $\langle T_b, \lambda_b, \Pi_b \rangle$ be a
sphere-cut decomposition of $\sk(b)$ that is rooted in
$r=\lambda^{-1}((s_{b},t_{b}))$. For a weak noose $O \subseteq C(T_b)$, let $\ARC(O)$
be the set of all types of $O$
satisfying the following two natural analogs of \enref{R1} and (R2), i.e.:
\begin{enumerate}[(RO1)]
\item If $X \in \ARC(O)$, then $O$ has type $X$.
\item If there is a witness $(D,D_H,G_H,H)$ for $G$ that
  \emph{respects} $\mathcal{B}$ such that $\type_W(b,O)=X$, then $X \in \ARC(O)$.
\end{enumerate}
Our aim is to compute $\ARC(O_{a^r})$ for the arc $a^r$ incident to
the root $r$ of $T_b$. We will achieve this by
computing $\ARC(O_a)$ for every arc $a$ of $T_b$ via a bottom-up dynamic
programming algorithm along $T_b$. Note that there is a one-to-one
correspondence between the arcs of $T_b$ that are connected to a leaf
and the children of $b$ in $\mathcal{B}$, i.e., the arc of $T_b$ incident to
leaf $l$ corresponds to the child $c$ of $b$ representing the edge
$\lambda(l)$. We start with two simple
lemmas showing that: (1) We can compute $\ARC(O_a)$ for every
arc of $T_b$ incident to a leaf $l$ of $T_b$ in linear-time from
$\RRR(c)$, where $c$ is the child of $b$ in $\BBB$ corresponding the
edge $\lambda(l)$ and (2). We can compute $\RRR(b)$ from $\ARC(O_{a^r})$ in
linear-time.

\begin{lemma} 
\label{lem:spcut-leaf}
  Let $a$ be an arc of $T_b$ connected to a leaf and let $c$ be the
  corresponding child of $b$ in $\mathcal{B}$.
  Then, $\ARC(O_a)$ can be computed in linear-time from $\RRR(c)$. 
\end{lemma}
\begin{proof}
  First note that $\md{a}=\{s_c,t_c\}$, where $(s_c,t_c)$ is the
  reference edge of $c$ in $\mathcal{B}$. Moreover, because
  of~\Cref{cor:add-nooses-to-drawing}, we can assume that $O_a=N_c$
  since their subcurves connect the same two vertices in the same
  face. Therefore, there is a one-to-one correspondence between the
  types in $\RRR(c)$ and the types in $\ARC(O_a)$. Moreover,
  given a type $X=(\psi, M, S) \in \RRR(c)$, then the corresponding
  type $X'=(\psi',M',S')$ in $\ARC(O_a)$ can be obtained as follows.
  Let $\alpha : \{L,R\} \rightarrow O_a\}$ the bijection such that
  $\alpha(L)=c$ if $c$ is equal to $L_c$ and $\alpha(c)=R$
  otherwise.
  \begin{itemize}
  \item We define $\psi'$ by setting
    $\psi'(\alpha(A))=\emptyset$ if $\psi(A)=\emptyset$,
    $\psi'(\alpha(A))=[x]$ if $|\psi(A)|=1$, and
    $\psi'(\alpha(A))=[x,x']$ if $|\psi(A)|=2$ for every $A \in
    \{L,R\}$,
  \item $M'$ is obtained from $M$ after replacing the vertices in
    $\psi(A)$ with their counterparts in $\psi'(\alpha(A))$ for every
    $A \in \{L,R\}$,
  \item $S'=S$.
  \end{itemize}
  Therefore, we obtain $\ARC(O_a)$ as the set $\SB X' \SM X \in
  \RRR(c)\SE$, which also shows that it can be computed in linear-time
  from $\RRR(c)$.
\end{proof}

\begin{lemma} 
\label{lem:arc_root}
  $\RRR(b)$ can be computed in linear-time from $\ARC(O_{a^r})$. 
\end{lemma}
\begin{proof}
  First note that $\md{a^r}=\{s_b,t_b\}$ and therefore that $O_{a^r}$ consists of two
  subcurves $c=(\{s,t\},f)$ and $c'=(\{s,t\},f')$, where both $f$ and $f'$
  have the reference edge $(s_b,t_b)$ on their border. Therefore,
  every type $X=(\psi,M,S) \in \ARC(O_{a^r})$ can be easily translated
  into two types of $b$ after
  specifying a bijection $\alpha$ between $\{L,R\}$ and
  $\{c,c'\}$. That is given such a bijection $\alpha : \{L,R\}
  \rightarrow \{c,c'\}$, we obtain the type
  $X_\alpha=(\psi_\alpha,M_\alpha,S_\alpha)$ of $b$ corresponding
  to $X$ by setting:
  \begin{itemize}
  \item $\psi_\alpha(L)=\emptyset$ if $\psi(\alpha(L))=\emptyset$,
    $\psi_\alpha(L)=\{l\}$ if $|\psi(\alpha(L))|=1$, and
    $\psi_\alpha(L)=\{l,l'\}$ if $|\psi(\alpha(c))|=2$,
  \item $M_\alpha$ is obtained from $M$ by replacing the first vertex
    in $\psi(\alpha(L))$ ($\psi(\alpha(R))$) with $l$ ($r$) and the second vertex in
    $\psi(\alpha(L))$ ($\psi(\alpha(R))$) with $l'$ ($r'$),
  \item $S_\alpha=S_\alpha$.
  \end{itemize}
  It is now straightforward to verify that $O_{a^r}$ has type $X$ if and
  only if $b$ has type $X_{\alpha_1}$ and $X_{\alpha_2}$ for the two
  possible bijections $\alpha_1$ and $\alpha_2$ between $\{L,R\}$ and
  $\{c,c'\}$. Therefore, it
  holds that $\RRR(b)=\SB X_{\alpha_1},X_{\alpha_2} \SM X \in
  \ARC(O_{a^r}) \SE$, which shows that $\RRR(b)$ can be computed in
  linear-time from $\ARC(O_{a^r})$.
\end{proof}

Given the above Lemmas, it now merely remains to show how to
compute $\ARC(O_{a_P})$ from $\ARC(O_{a_L})$ and $\ARC(O_{a_R})$ for
any inner node of $T_b$ with parent arc $a_P$ and child arcs $a_L$ and
$a_R$. Employing our framework introduced in Subsection~\ref{ssec:SPCD-framework} allows us
to solve a simpler problem instead, i.e., we only have to show how to
compute $\ARC(O_1\xor{} O_2)$ from $\ARC(O_1)$ and $\ARC(O_2)$ for
any weak nooses $O_1$ and $O_2$.

Let $O_1$ and $O_2$ be two weak nooses having type $X_1=(\psi_1, M_1, S_1)$ 
and type $X_2=(\psi_2, M_2, S_2)$, respectively. We say that $X_1$ and $X_2$ are \emph{compatible} if
\begin{enumerate}[(1)]
    \item $O=O_1 \xor O_2$ is a weak noose,
    \item the inside region of the noose $O$ contains all subcurves in $(O_1\cap O_2)$,
    \item $\forall {c \in O_1 \cap O_2}$, it holds $\psi_1(c) = \psi_2(c)$,
    \item for every $u \in \m{O_1 \cap O_2} \setminus  \m{O_1 \xor O_2}$, it holds that $u$ is only in one of following sets: $S_1$, $S_2$ or $V(M_1)\cap V(M_2)$, and
    \item the multi-graph obtained from the union of $M_1$ and $M_2$
      is acyclic, or is one cycle and $\m{O}
        \subseteq S_1\cup S_2 \cup (V(M_1) \cap V(M_2))$,
    \item if $X_1$ is the \fulltype{}, then $X_2$ is the \emptytype{}
      and $\m{O_2} \subseteq \m{O_1}$, and vice versa.
\end{enumerate}
Please also refer to Figure~\ref{fig:typesSPCD} for an illustration of
two compatible types.
Let $X_1=(\psi_1,M_1, S_1)$ and $X_2=(\psi_2,M_2, S_2)$ be two compatible types defined on weak nooses $O_1$ and $O_2$, respectively.

We denote by $X_1 \concat X_2$
the \emph{combined type} $X=(\psi,M, S)$ of $X_1=(\psi_1,M_1, S_1)$ and $X_2=(\psi_2,M_2, S_2)$ for the weak noose $O=O_1 \xor O_2$ that is defined as follows. 
For each $c \in O$, if $ c \in O_1$ then $\psi(c)$ is equal to
$\psi_1(c)$, otherwise $\psi(c)$ is equal to $\psi_2(c)$ 
and the set $S$ is equal to $(S_1 \cup S_2 \cup (V(M_1)\cap V(M_2)))\cap
\m{O}$, i.e., any vertex with degree two w.r.t. $X$ must be in $\m{O}$
and have degree two already w.r.t. $X_1$ or $X_2$,
or it must be in both matchings $M_1$ and $M_2$.
If either $X_1$ or $X_2$ is a \fulltype, then by \enref{6} we get that 
$M_1 = M_2 = M = \emptyset$ and $X_1\concat X_2$ is the \fulltype{}.
If the multi-graph $M_1 \cup M_2$ is one cycle, then by \enref{5} we get that 
$M = \emptyset$ and $X_1\concat X_2$ is the \fulltype{}.
Otherwise, due to \enref{5}, the multi-graph $M_1 \cup
M_2$ is acyclic and corresponds to a set of paths.
Therefore, the matching $M$ is the set
containing the two endpoints for every path in $M_1 \cup M_2$.

\iflong
\begin{observation}\label{obs:combined_type_time}
  Let $X_1$ and $X_2$ be two types defined on the weak nooses $O_1$
  and $O_2$, respectively. Then, we can check whether $X_1$ and $X_2$
  are compatible and if so compute the type $X_1 \concat X_2$ in time
  $\bigoh(|O_1| + |O_2|)$.
\end{observation}

The following two lemmas are crucial for showing the correctness of
our approach. The former shows that if there is a witness $W$ for $G$ that
respects $\BBB$, then for every two weak nooses $O_1$ and $O_2$ it
holds that $\type_W(b,O_1)$ and $\type_W(b,O_2)$ are compatible types
and $\type_W(b,O)=\type_W(b,O_1) \concat \type_W(b,O_2)$. The latter
shows in some sense the reverse direction, i.e., if $O_1$ and $O_2$
have compatible types $X_1$ and $X_2$, then $O=O_1\xor O_1$ has type
$X_1 \concat X_2$.
\begin{lemma} 
\label{lem:concat-for}
  Let $W=(D,D_H,G_H,H)$ be a witness for $G$ that \emph{respects}
  $\BBB$. Let
  $b$ be an R-node or an S-node with sphere-cut decomposition
    $\langle T_b,\lambda_b,\Pi_b \rangle \in \mathcal{T}$. 
  Let $O_1$ and $O_2$ be two weak nooses that are subsets of
  $C(T_b)$ and satisfy properties \enref{1} and \enref{2}, i.e.,
  $O=O_1\xor O_2$ is also a weak noose and the inside region of $O$ contains all
  subcurves in $(O_1\cap O_2)$.
     Then, $X_1=\type_W(b,O_1)$ and $X_2=\type_W(b,O_2)$ are compatible
  types and $\type_W(b,O)=X_1\concat X_2$.
  \end{lemma}
\iflong \begin{proof}
  Let $i \in \{1,2\}$ and let $X_i=(\psi_i,M_i,S_i)$ be the type $\type_W(b,O_i)$. 
  The properties \enref{1} and \enref{2} given in the description of compatible types are a direct consequence of the assumptions of this lemma.
    Recall that $H^{O_i}$ is defined in
    Subsection~\ref{ssec:SPCD-framework} and essentially corresponds to
    the subgraph of $H$ including crossings at the subcurves of $O_i$ inside $O_i$.
  If $H^{O_i}$ contains a cycle then $X_i$ and $X$ are the \fulltype{}s.
  Also, for $j\in\{1,2\} \setminus\{i\}$, 
  $X_j$ must be the \emptytype{} and $\m{O_j}\subseteq V(H^{O_i})$, which satisfy property \enref{6}.
  In this case the properties \enref{3}, \enref{4}  and \enref{5} are satisfied,
  because $V(\psi_i)=V(\psi_j)=M_1=M_2=\emptyset$ and $S_i=\m{O_i}$.

  Otherwise,
  let $\PPP_i$ be a set of all maximal paths in $H^{O_i}$ each of size at least $2$.
  Then, $(\PPP_i,D_H^{O_i})$ witnesses that $O_i$ has type $X_i$. Note that $\bigcup(\PPP_1 \cup \PPP_2)$ is almost equal to $H^O$. In fact, $H^O$ only misses the vertices in $V(\psi_1)\cap V(\psi_2)$. We therefore define $H^O_*$
  as the graph obtained from $H^O$ after subdividing the edges that cross the subcurves in $O_1\cap O_2$. Then, we can assume that $\bigcup(\PPP_1 \cup \PPP_2)$ is equal to $H^O_*$.
  
  The property \enref{3} is simply obtained from the fact that $X_1$ and $X_2$ are obtained from the same Hamiltonian cycle and therefore agree on all subcurves shared between $O_1$ and $O_2$. 

  For each $v$ in $\m{O_1\cap O_2}\setminus \m{O}$, $v$ has degree $2$ in $H^O$, because $v$ is not in $\m{O}$. Since $\bigcup(\PPP_1 \cup \PPP_2)=H^O_*$, it follows that $v$ is in one of the following sets: $\INVP(\PPP_1) \cap \m{O_1}$, $\INVP(\PPP_2)\cap \m{O_2}$ and $\bigl(V(\PPP_1) \setminus \INVP(\PPP_1)\bigr) \cap \bigl(V(\PPP_2) \setminus \INVP(\PPP_2)\bigr)$,  which correspond to sets $S_1$, $S_2$ and $V(M_1)\cap V(M_2)$, respectively. This demonstrates property \enref{4} .  

  Property \enref{5} now follows because
  the matchings $M_1$ and $M_2$ have an edge between the endpoints of every path in $\PPP_1$ and $\PPP_2$,
  so if $H^O_*$ is acyclic then the multi-graph $M_1 \cup M_2$ is also acyclic,
  otherwise $H^O_*$ is a cycle and the multi-graph $M_1 \cup M_2$ is a cycle, and $\m{O} \subseteq V(H^O_*)$.
\end{proof}\fi

\begin{lemma}
\label{lem:concat-rev}
    If $O_1$ and $O_2$ have compatible types $X_1$ and $X_2$, respectively, then $O=O_1 \xor O_2$ has type $X=X_1\concat X_2$.
\end{lemma}
\iflong \begin{proof}
    
    Let $X_1=(\psi_1, M_1, S_1)$, $X_2=(\psi_2, M_2, S_2)$, and $X=(\psi, M, S)$.
    For each $i\in \{1,2\}$, let $(\PPP_i,D_i)$ be the witness that $O_i$ has type $X_i$. Since $X_1$ and $X_2$ are compatible and in particular because of properties \enref{2} and \enref{3}
    from the definition of compatible types,
    the drawing $D=D_1\cup D_2$ is planar. 
    Note that $O$ is a weak noose because of property \enref{1}.
    Consider the graph $H=\bigcup(\PPP_1\cup \PPP_2)$.
    Because of property \enref{4}  all endpoints of the paths in $H$ are in $\m{O}\cup V(\psi)$.
    
    If $H$ is a single cycle then from the property \enref{5} and \enref{6},
    we get that $X_1\concat X_2$ is a \fulltype{} and the witness is $(H,D)$. 
    
    Otherwise, the multi-graph $M_1\cup M_2$ is a disjoint union of paths, due to the property \enref{5}.
    Note that each path in $\PPP_i$ corresponds to an edge in $M_i$, for $i \in \{1,2\}$.
    So $H$ is also a disjoint union of paths and let $\PPP$ be the set of paths of $H$. 
    Then $(\PPP, D)$ is the witness of that $O$ has the type $X$.
\end{proof}\fi
\fi

The following lemma is required to compute the types for a weak noose
$O=O_1\xor{}O_2$ and provides a detailed analysis of the run-time required.
\begin{lemma} 
\label{lem:compatible_triples}
    Let $O$, $O_1$ and $O_2$ be weak nooses such that $O=O_1\xor
    O_2$. There are at most $6(84\sqrt{14})^k$ triples $(X, X_1, X_2)$
    such that $X$, $X_1$ and $X_2$ are types defined on $O$, $O_1$ and $O_2$,
    respectively, $X_1$ is compatible with $X_2$, and $X=X_1\concat
    X_2$. Moreover, all such triples can be
    enumerated in $\bigoh((84\sqrt{14})^kk)$, where $k = max\{|O|,|O_1|,
    |O_2|\}$.
\end{lemma}

\iflong \begin{proof}
    We define the role of a vertex or subcurve in a type as the information stored in the type about that vertex or subcurve.
    First, we will show that for fixed $X$ the role of each vertex $\m{O}\setminus \m{O_1\cap O_2}$ and each subcurve from $O$ in types $X_1$ and $X_2$ remains the same.
    Let $X=(\psi,M,S)$ be a type that can be defined on a weak noose $O$ and $u$ be an arbitrary vertex from $\m{O} \cap \m{O_1}\cap \m{O_2}$.
    For each $i \in \{1,2\}$, 
    let $v_i$ be first vertex after $u$ in clockwise orientation such that
    $v_i \in V(M)\setminus \m{O_1\cap O_2}$ and
    $v_i$ is a vertex from the subcurves from the segment $O \cap O_i$.
    Let $DW^i$ be the Dyck word corresponding
    to the matching $M$ from Observation~\ref{obs:dyck} with starting
    vertex $v_i$ and clockwise orientation and 
    let $DW^{i}_{\star}$ be a prefix of $DW^i$ corresponding to vertices on the
    subcurves from the segment $O \cap O_i$.
    Let $X_1=(\psi_1, M_1, S_1)$ and $X_2=(\psi_2, M_2, S_2)$ be the
    types that can be defined on $O_1$ and $O_2$ respectively, such
    that $X=X_1\concat X_2$.
    Let $DW_i$ be the Dyck word corresponding
    to the matching $M_i$ from Observation~\ref{obs:dyck} with starting
    vertex $v_i$ and clockwise orientation.
   Note that for different pairs $X_1$ and
    $X_2$, the type $X$ remains the same and therefore also $DW^{1}_{\star}$, $DW^{2}_{\star}$ and $S$.
    Moreover, $DW^{i}_{\star}$ is a prefix of $DW_i$ and 
    $S_i\cap( \m{O\cap O_i} \setminus  \m{O_1\cap O_2}) = 
     S \cap (\m{O\cap O_i} \setminus  \m{O_1\cap O_2})$.
    This means that, the role of each vertex $\m{O \cap O_i} \setminus  \m{O_1\cap O_2}$
    and each subcurve $O \cap O_i$ in type $X_i$ remains the same, 
    so the only places where the different pairs of compatible types 
    may differ is in the segment $O_1 \cap O_2$.

    Secondly, we will bound number of different pairs $X_1$ and $X_2$ for fixed $X$.
    By condition \enref{1} from the definition of compatible types, for each $v \in \m{O_1 \cap O_2} \setminus  \m{O}$, there are $2+2\cdot 2 = 6$ different combinations of roles of $v$ in types $X_1, X_2$, i.e.,
    $v \in S_1 \land v \notin S_2 \cup V(M_2)$,
    $v \in S_2 \land v \notin S_1 \cup V(M_1)$ 
    or $v \in V(M_1) \cap V(M_2)$ and $v$ corresponds to either $"["$ or $"]"$ in $DW_1$ and either $"["$ or $"]"$ in $DW_2$.
    Moreover, for each $v \in \m{O_1 \cap O_2}  \cap \m{O}$ there are
    at most $6$ different combination of roles of $v$ in types $X_1$
    and $X_2$, because the role of $v$ in type $X$ is known and this
    can only decrease number of combinations.
    Due to the type definition and condition \enref{3} from the definition
    of compatible types, we obtain that $\bigcup_{c\in O_1\cap
      O_2}\psi_1(c)=\bigcup_{c\in O_1\cap
      O_2}\psi_2(c)=V(\psi_1) \cap V(\psi_2) \subseteq V(M_1) \cap
    V(M_2)$. Therefore, for each $v \in V(\psi_1) \cap V(\psi_2)$,
    there are $2\cdot 2=4$ different combinations of roles of $v$ in
    types $X_1$ and $X_2$, i.e., $v$ corresponds to either $"["$ or
    $"]"$ in $DW_1$ and either $"["$ or $"]"$ in $DW_2$.
    For each subcurve $c \in  O_1\cap O_2$, there are $3$ possible values
    $\{\emptyset, [x], [x, x']\}$ for $\psi_1(c)$, and therefore there
    are $1+4+16=21$ possibilities, i.e., $1$, $4$, and $16$ possibilities
    in case that $\psi_1(c)=\emptyset$, $\psi_1(c)=[x]$, and
    $\psi_1(c)=[x,x']$, respectively, of the role of $c$ in types $X_1$ and $X_2$.
    Furthermore, since $ |\m{O_1\cap O_2}| = |O_1\cap O_2| +1$, there
    are at most $6^{|O_1\cap O_2| +1 }21^{|O_1\cap O_2|} = 6\cdot
    126^{|O_1\cap O_2|}$ different pairs of types $X_1$ and $X_2$ for
    fixed $X$.
    
    There are at most $28^{|O|}$ different types $X$ that can be
    defined on $O$, due to the Lemma~\ref{lem:enum_arc_types}, and
    there are $6\cdot 126^{|O_1\cap O_2|}$ different pairs of types
    $X_1$ and $X_2$ that can be defined on $O_1$ and $O_2$
    respectively, such that $X=X_1\concat X_2$, so there are at most
    $28^{|O|}\cdot 6\cdot 126^{|O_1\cap O_2|}$ different triples
    $(X,X_1, X_2)$.
    Note that $|O_1|+|O_2| - 2|O_1 \cap O_2| = |O|$ therefore there
    are $\bigoh(28^k 126^{\frac{k}{2}}) = \bigoh(84\sqrt{14}^k)$
    different triples $(X,X_1, X_2)$.
    
    In order to generate all valid triples, first we enumerate all
    possible types $X$, of which there are $\bigoh(28^{|O|})$, in
    time $\bigoh(28^{|O|}|O|)$ using Lemma~\ref{lem:enum_arc_types}.
    Then based on type $X$ we fix the role of all vertices in
    $\m{O}\setminus \m{O_1 \cap O_2}$ and all subcurves in $O$ in types
    $X_1$ and $X_2$.
    We can then assign a role to each
    vertex in $\m{O_1 \cap O_2}$ and every subcurve of $O_1 \cap O_2$
    for types $X_1$ and $X_2$ and verify that the
    corresponding words are Dyck words in time $\bigoh(|O_1| + |O_2|)$
    and if so translate it into a type
    description using Observation~\ref{obs:dyck}.
    Lastly we check if $X_1$ and $X_2$ are compatible and if so check
    if $X=X_1 \concat X_2$, in $\bigoh(|O_1| + |O_2|)$ time using
    Observation~\ref{obs:combined_type_time}.
    Therefore the time complexity of this operation is
    $\bigoh(28^{|O|}(|O| + 126^{|O_1\cap O_2|}\cdot(|O_1|+|O_2|))) =
    \bigoh(84\sqrt{14}^kk)$, due to equation $|O_1|+|O_2| - 2|O_1 \cap
    O_2| = |O|$.
\end{proof}\fi

\begin{lemma} 
\label{lem:spcut-inner}
  Let $b$ be an R-node or S-node and
  let $a_P$ be a parent arc with two child arcs $a_L$ and $a_R$ in the
  sphere-cut decomposition $\langle T_b,\lambda_b,\Pi_b\rangle$ of $\sk(b)$.
  We can compute $\ARC(O_{a_P})$ from $\ARC(O_{a_L})$ and $\ARC(O_{a_R})$ in
  $\bigoh(315^{k})$ time, where $k = max(|\md{a_P}|, |\md{a_L}|,$ $
  |\md{a_R}|)$.
\end{lemma}
\begin{proof}
  Note first that $\sk(b)$ is biconnected because $b$ is either an R-node
  or an S-node. Therefore, we can apply~\Cref{lem:triangles}, to obtain a sequence $Q$ of
  at most $3$ $\xor{}$-operations such that:
  \begin{itemize}
  \item $Q$ contains only the weak nooses $O_{a_L}$, $O_{a_R}$ and at
    most two weak nooses $O^1$ and $O^2$ each bounding an edge-less
    graph with three vertices.
  \item Every step of $Q$ produces a weak noose $O$ such that $|O|\leq 1+k$ and $O_{a_P}$ is the weak
    noose produced by $Q$ after the final step.
  \end{itemize}
  Before we can employ $Q$ to compute $\ARC(O_{a_P})$, we first need
  to compute $\ARC(O^i)$ for the at most two weak nooses $O^1$ and
  $O^2$. To do so we employ~\Cref{lem:enum_arc_types} to
  enumerate all possible types $X$ of $O^i$, which because $|\m{O^i}|\leq
  3$ can be achieved in constant time. We then add each of those types
  to $\ARC(O^i)$; this is correct because the noose does not contain
  any edges and therefore allows for every possible type.
  We then compute $\ARC(O_{a_P})$ using $Q$ as follows. For every step
  of $Q$, which given two weak nooses $O_1$ and $O_2$ for which the
  set of types $\ARC(O_1)$ and $\ARC(O_2)$ have already been computed,
  computes the weak
  noose $O=O_1\xor{} O_2$, we do the following to compute $\ARC(O)$.
  Let $k'=\max\{|O|,|O_1|,|O_2|\}$. Using~\Cref{lem:compatible_triples} we
  enumerate all of the at most $6(84\sqrt{14})^{k'}$ triples $(X,X_1,X_2)$ of types defined on $O$, $O_1$,
  $O_2$, respectively, in time $\bigoh((84\sqrt{14})^{k'}k')$. Then, for each
  such triple $(X,X_1,X_2)$, we check (in constant time) whether $X_1
  \in \ARC(O_1)$ and $X_2 \in \ARC(O_2)$ and if so we add $X$ to
  $\ARC(O)$. Because $k'\leq k+1$ and since $Q$ consists of at most
  $3$ steps, we obtain that computing all steps
  of $Q$ and therefore computing the set $\ARC(O_{a_P})$ takes time at
  most $\bigoh((84\sqrt{14})^{k+1}(k+1))=\bigoh(315^{k})$. Finally, the correctness of the
  procedure follows immediately from~\Cref{lem:concat-for,lem:concat-rev}.
\end{proof}

\begin{proof}[Proof of Lemma~\ref{lem:Rnode}]
  We first use~\Cref{lem:comp-spcut} to compute a sphere-cut
  decomposition $\langle T_b, \lambda_b,\Pi_b \rangle$ of $\sk(b)$, whose
  width $\omega$ is equal to the branchwidth of $G$, having 
  at most $\bigoh(|V(\sk(b)|)=\bigoh(\ell)$ nodes in time
  $\bigoh(\ell^3)$. Note that to compute $\langle T_b, \lambda_b,\Pi_b
  \rangle$ we can use any of the (at most) two planar drawings of
  $\sk(b)$ that contain the reference edge $(s_b,t_b)$ in the
  outer-face, since we will take the resulting symmetries into account
  when we compute the set of types; more specifically
  in~\Cref{lem:arc_root} and~\Cref{lem:spcut-leaf}.

  We then compute $\ARC(O_{a^r})$ using a
  bottom-up dynamic programming algorithm on $T_b$. In particular, we
  use~\Cref{lem:spcut-leaf} to compute $\ARC(O_a)$ for all
  arcs in $T_b$ incident to a leaf node of $T_b$ and then we use
  Lemma~\ref{lem:spcut-inner} to compute $\ARC(O_a)$ for any other arc
  $a$ of $T_b$ in a bottom-up manner. Having computed $\ARC(O_{a^r})$, we
  then use Lemma~\ref{lem:arc_root} to obtain $\RRR(b)$ from
  $\ARC(O_{a^r})$. The correctness of the algorithm follows from the
  employed lemmas. To analyze the run-time of the algorithm, we first
  note that we require time at most $\bigoh(\ell^3)$ to compute the
  sphere-cut decomposition $\langle T_b, \lambda_b,\Pi_b
  \rangle$. Moreover, the run-time of the dynamic
  programming algorithm on $\langle T_b, \lambda_b,\Pi_b
  \rangle$ is at most equal to the number of inner nodes of
  $T_b$, i.e., at most $|E(\sk(b))|=\ell+1$, times the time required
  for one application of Lemma~\ref{lem:spcut-inner}, i.e., at most
  $\bigoh(315^{\omega})$, where $\omega$ is the width of
  $T_b$; note that here we use that $k = max\{|\md{a_P}|,
  |\md{a_L}|,|\md{a_R}|\}\leq \omega$. Therefore, we obtain
  $\bigoh(315^{\omega}\ell+\ell^3)$ as the
  total run-time required to compute $\RRR(b)$.
\end{proof}
\fi

 \subsection{Putting Everything Together}\label{ssec:puttogether}
 \ifshort
     Finally, we show how to compute the set of types for every leaf
  (Q-node) $l$ of $\BBB$ in time $\bigoh(1)$; informally, since
  $\pe(b)$ is just an edge $(s,t)$, $\RRR(l)$ contains all types that do not
  allow the Hamiltonian cycle to cross from left to right without
  using either $s$ or $t$. Together
  with Lemma~\ref{lem:Pnode} and~\ref{lem:Rnode}, this then concludes
  the proof of Lemma~\ref{lem:sh-fpt-bw-given}.
  \fi

\iflong
Here, we put everything together and
prove~\Cref{lem:sh-fpt-bw-given}. Before doing so, we first need the
following simple lemma
that allows us to compute the set of types for every leaf node of
$\BBB$ in
constant time.
\begin{lemma} 
\label{lem:Qnode}
  Let $l$ be a leaf-node (and Q-node) of $\BBB$.
  We can compute $\RRR(l)$ in time $\bigoh(1)$.
\end{lemma}
\begin{proof}
  Let $l$ be a leaf-node with reference edge $(s,t)$ of $\BBB$. Then,
  $l$ is also a Q-node with edge $\{s,t\}$ due to the properties of
  \SPQR{}-trees. Let $\psi_{x,y}$ for $x,y \in [0,2]$ be defined by
  setting $\psi_{0,y}(L)=\emptyset$, $\psi_{1,y}(L)=\{l\}$,
  $\psi_{2,y}(L)=\{l,l'\}$,
  $\psi_{x,0}(R)=\emptyset$, $\psi_{x,1}(R)=\{r\}$, and $\psi_{x,2}(R)=\{r,r'\}$.

  $\RRR(l)$ contains the following types:
  \begin{itemize}
  \item Types for $\psi_{0,0}$:
    \begin{itemize}
    \item the type $(\psi_{0,0},\emptyset, \{s,t\})$
      indicating a Hamiltonian cycle on $\pe(l)$,
    \item the type $(\psi_{0,0},\emptyset, \emptyset)$,
    \item the type $(\psi_{0,0},\{\{s,t\}\}, \emptyset)$;
    \end{itemize}
  \item Types for $\psi_{1,0}$ (symmetrically for $\psi_{0,1}$):
    \begin{itemize}
    \item the types $(\psi_{1,0},\{\{l,s\}\}, \emptyset)$ and
      $(\psi_{1,0},\{\{l,s\}\}, \{t\})$,
    \item the types $(\psi_{1,0},\{\{l,t\}\}, \emptyset)$ and $(\psi_{1,0},\{\{l,t\}\}, \{s\})$;
    \end{itemize}
  \item Types for $\psi_{1,1}$:
    \begin{itemize}
    \item for every $S \subseteq \{s,t\}$ the type
      $(\psi_{1,1},\{\{l,r\}\}, S)$,
    \item the types $(\phi_{1,1},\{\{l,s\},\{t,r\}\},\emptyset)$ and $(\phi_{1,1},\{\{l,t\},\{s,r\}\},\emptyset)$;
    \end{itemize}
  \item Types for $\psi_{2,0}$ (symmetrically for $\psi_{0,2}$):
    \begin{itemize}
    \item for every $S \subseteq \{s,t\}$, the type $(\psi_{2,0},\{\{l,l'\}\}, S)$,
    \item the type $(\psi_{2,0},\{\{s,t\},\{l,l'\}\}, \emptyset)$,
    \item the type $(\psi_{2,0},\{\{l,s\},\{l',t\}\}, \emptyset)$,
    \end{itemize}
  \item Types for $\psi_{2,1}$ (symmetrically for $\psi_{1,2}$):
    \begin{itemize}
    \item the type $(\psi_{2,1},\{\{l,r\},\{l',t\}\}, \{s\})$,
    \item for every $S \in \{\emptyset, \{t\}\}$, the type
      $(\psi_{2,1},\{\{l,l'\},\{s,r\}\}, S)$,
    \item for every $S \in \{\emptyset, \{s\}\}$, the type
      $(\psi_{2,1},\{\{l,l'\},\{t,r\}\}, S)$,
    \end{itemize}
  \item Types for $\psi_{2,2}$:
    \begin{itemize}
    \item the type
      $(\psi_{2,2},\{\{l,l'\},\{s,r\},\{t,r'\}\},\emptyset)$,
    \item the type $(\psi_{2,2},\{\{r,r'\},\{l,s\},\{l',t\}\},\emptyset)$,            
    \item the type
      $(\psi_{2,2},\{\{l,l'\},\{r,r'\},\{s,t\}\},\emptyset)$,
    \item for every $S \subseteq \{s,t\}$, the type
      $(\psi_{2,2},\{\{l,l'\},\{r,r'\}\}, S)$,
    \item the type $(\psi_{2,2},\{\{l,r\},\{l',r'\}\}, \{s,t\})$,
    \end{itemize}
  \end{itemize}
  Note that $\RRR(l)$ can be computed in constant time and actually contains all types of $\pe(l)$ and
  therefore also satisfies \enref{R1} and \enref{R2}.
\end{proof}

We are now ready to show~\Cref{lem:sh-fpt-bw-given}.
\lemshfpt*

\begin{proof}   We start by showing how to compute the set of types $\RRR(b)$
  for every node $b$ of the \SPQR{}-tree
  $\BBB$, which we will achieve using a bottom-up dynamic programming algorithm
  along $\BBB$. As stated in Section~\ref{sec:pre}, we assume that
  $\BBB$ is rooted at some Q-node with edge $e$, whose child $b_r$ has
  $e$ as its reference edge. Starting at the leaves of $\BBB$, we
  use~\Cref{lem:Qnode} to compute $\RRR(l)$ for every leaf node $l$ of
  $\BBB$ in constant time. We then iteratively consider the inner
  nodes $b$ for which $\RRR(c)$ for all children $c$ of $b$ in $\BBB$ have
  already been computed. Let $b$ be a node of $\mathcal{B}$ with
  $\ell$ children. If $b$ is an R-node or an S-node, we
  use~\Cref{lem:Rnode} to compute $\RRR(b)$ in time
  $\bigoh(315^{\omega}\ell+\ell^3)$, where $\omega$ is the branchwidth
  of $\sk(b)$. Otherwise $b$ is a P-node and we use~\Cref{lem:Pnode}
  to compute $\RRR(b)$ in time $\bigoh(\ell))$.
  By applying the above procedure exhaustively, we obtain the set
  $\RRR(b)$ of types for all nodes apart
  from the root node $r$ of $\BBB$; this is because $r$ is a Q-node
  which is not a leaf of $\BBB$. Let $b_r$ be the unique child of $r$
  in $\BBB$ and let $e=(s,t)$ be the reference edge of $b_r$ (which is
  also the reference edge of $r$, because $r$ is a Q-node). Since
  $b_r$ is not the root of $\BBB$, we have computed the set
  $\RRR(b_r)$ of types for $b_r$. We now claim that $G$ is
  subhamiltonian if and only if $(\psi_\emptyset,\emptyset, \{s,t\})
  \in \RRR(b_r)$, where
  $\psi_\emptyset(L)=\psi_\emptyset(R)=\emptyset$. Towards showing the
  forward direction of the claim suppose that $G$ is
  subhamiltonian. It then follows from~\Cref{lem:nicedrawing} that $G$
  has a witness $(D,D_H,G_H,H)$ that respects $\BBB$.
  Consequently, we
  obtain from \enref{R2} that
  $\type_W(b_r) \in \RRR(b_r)$. Therefore, if
  $\type_W(b_r)=(\psi_\emptyset,\emptyset,\{s,t\})$, then we are done.
  Otherwise, consider first the case that $H$ contains
  the edge $\{s,t\}$ of $G$. In this case we can replace the edge in $H$ by
  adding a new edge between $s$ and $t$, which we can draw arbitrary
  close to the original edge in $G$ between $s$ and $t$. Therefore, we can
  assume that $H$ does not contain the edge of $G$ between $s$ and
  $t$. But then, we can obtain a new witness $W'=(D,D_H',G_H,H)$ that
  respects $\BBB$ such that
  $\type_{W'}(b_r)=(\psi_\emptyset,\emptyset,\{s,t\})$ by changing the
  drawing $D_H$ of $H$ into the new drawing $D_H'$
  such that $H$ touches the noose $N_{b_r}$ only at $s$ and
  $t$; this can be achieved by replacing every subcurve in $D_H$ of
  $H$ outside of $N_{b_r}$ with a curve inside $N_{b_r}$ drawn
  arbitrarily close to $N_{b_r}$. 
  Towards showing the reverse direction, suppose that
  $(\psi,\emptyset,\{s,t\}) \in \RRR(b_r)$. By the definition of a
  type, it follows that there is a Hamiltonian cycle $H$ for
  $G$ that can be drawn together with $G\setminus \{e\}$ entirely
  within the noose $N_b$. But then, $H$ can also be drawn together
  with $G$ and therefore shows that $G$ is subhamiltonian, as
  required.

  The run-time of the algorithm is at most the number of nodes of
  $\BBB$, which because of~\Cref{lem:computeSPQR} is at most
  $\bigoh(|E(G)|)=\bigoh(|V(G)|)$ (because $G$ is planar), times the maximum
  time required by the dynamic programming procedure at every node of
  $\BBB$. Since the latter is dominated by the time required for
  R-nodes, i.e., $\bigoh(315^{\omega}+\ell^3)$ due
  to~\Cref{lem:Rnode}, where $\omega$ is the branchwidth of $\sk(b)$
  and $\ell$ is the number of children of the node in $\BBB$, we
  obtain $\bigoh((315^{\omega}|V(G)|+|V(G)|^3))$ as the total
  run-time of the algorithm.
\end{proof}
\fi

\section{An Algorithm Using the Feedback Edge Number}\label{sec:btfen}

In this section, we establish the following theorem:
 
\begin{theorem}
  \label{thm:fen}
  \bt{} is fixed-parameter tractable when parameterized by the feedback edge number of the input graph.
\end{theorem}

\ifshort
The result is achieved by separately handling two cases: one where the targeted number of pages is greater than $2$, or where it is precisely $2$. Both cases are handled by a kernelization procedure, and in both cases it is easy to show that pendant vertices can be safely removed. At this point, the target graph consists of a tree plus $k$ edges, whereas the only part that may remain large in this tree are paths of degree-$2$ vertices. In the former case, we obtain a non-trivial proof that allows us to reduce the maximum length of such a path to length that is bounded by an exponential function of the feedback edge number. In the latter case (which is equivalent to solving \SH), the reduction step is easier and we in fact obtain a linear kernel for the problem:

  \begin{theorem}
  \label{thm:fenhamkernel}
  \SH{} parameterized by the feedback edge number $k$ admits a kernel 
  with at most $12k-8$ vertices and at most $14k-9$ edges.
\end{theorem}

Moreover, by combining Theorem~\ref{thm:fenhamkernel} with the subexponential algorithm of Corollary~\ref{cor:main}, we can slightly strengthen our main result as follows.
\fi

 \iflong
   To obtain the result, we distinguish whether the bound on the number of pages is $2$, or more. We begin with the latter case.

\subsection{The Case with More than Two Pages}
The remainder of this section is devoted to a proof of
Theorem~\ref{thm:fenbig} (stated below), which is based on providing an (exponentially sized) kernel for the problem. 
We begin by introducing a few section-specific definitions.

\begin{theorem}\label{thm:fenbig}
  When restricted to inputs $(G,k)$ such that $k\geq 3$, \bt~is fixed-parameter tractable parameterized by the feedback edge number.
\end{theorem}

\subparagraph{Notation and Definitions.}
 Let $G$ be an $n$-vertex graph and $L=(\prec,\sigma)$ be a $k$-page embedding of $G$.
For the purposes of this section, it will be useful to think of the linear order $\prec$ on $V(G)=\{v_1,\ldots,v_n\}$ as a set of $n$ points on the real line such that $v_1< v_2<\ldots< v_n$. With this interpretation in mind, we can define a \emph{region} as the open interval between two consecutive elements of $V(G)$. 
 
Let $U\subseteq V(G)$ be a subset of vertices of $G$. We say that a path $P$ of $G$ is {\it maximal proper for $U$} if $P$ (1) is not only a single edge, (2) every internal vertex $v$ of $P$ satisfies $\textup{deg}_G(v)=2$ and $v\not\in U$, and (3) $P$ is maximal (with respect to containment) among all paths satisfying properties (1) and (2).

Let $L=(\prec,\sigma)$ be a book embedding, we say that edge $e'=u'v'$ is {\it nested} under the edge $e=uv$ if $\sigma(e)=\sigma(e')$ and $(u',v')\subset (u,v)$, that is either
$u\prec u'\prec v'\prec v$ or $u=u'\prec v'\prec v$ or $u\prec u'\prec v'=v$.
   Moreover, we say that an edge {\it touches} a region if one of its endpoints belongs to the interior of that region, and that a path {\it touches} a region $t$ times if $t$ of its edges touches that region. 
We use $\sigma(v)=\{\sigma(e)~|~e\in E(G)\wedge v\in e\}$ to denote the set of all the page numbers of the edges incident to $v$.

We say that a graph $G$ and a set of paths $\mathcal{P}$ are {\it near-disjoint} if the vertices that are in common among $G$ and the paths of $\mathcal{P}$ are exactly their endpoints.
Let $G$ be a graph and $\mathcal{P}$ be a set of paths such that $G$ and $\mathcal{P}$ are near-disjoint: we say that a set of paths $\mathcal{P}'$ {\it serves as} $\mathcal{P}$ if $G$ and $\mathcal{P}'$ are near-disjoint, $|\mathcal{P}'|=|\mathcal{P}|$, and there is a path in $\mathcal{P}$ with $u$ and $v$ as endpoints if and only if there is a path in $\mathcal{P}'$ with $u$ and $v$ as endpoints.
If a graph $G$ and a set of paths $\mathcal{P}$ are near-disjoint, we denote the graph obtained by inserting $P$ into $G$ as $G\curlyvee \mathcal{P}$.

 \newcommand{\Gr}{G_{\geq 2}}
\newcommand{\Ll}{T_1}
\newcommand{\Ii}{T_{\geq 3}}
\subparagraph{The Branching Step.}
 In the first part of the algorithm, we apply brute-force branching and a simple preprocessing rule. We begin by exhaustively removing all pendant vertices in the graph.
 
  \begin{lemma}
\label{lem:isol}
Let $G$ be a graph and $v$ be a vertex of degree at most one in
$G$. Then $(G,k)$ and $(G-v,k)$ are equivalent instances of \bt.
\end{lemma}
\iflong
\begin{proof}
We show that $(G,k)$ is a yes-instance of \bt~if and only if $(G-v,k)$ is a yes-instance of the same problem.
The forward direction follows directly from the fact YES instances are preserved when considering subgraphs: if $H$ is a subgraph of $G$, a $k$-page book embedding of $H$ can be obtained from a $k$-page book embedding of $G$ by removing $G-H$. 

Consider now a $k$-page book embedding $L$ for $G-v$. If $v$ is an isolated vertex of $G$, a $k$-page book embedding for $G$ can be obtained from $L$ by inserting $v$ in any region of $L$. Suppose $v$ is a leaf of $G$ and let $u_v$ be the unique neighbor of $v$ in $G$. A $k$-page book embedding $L'$ for $G$ can be obtained from $L$ by inserting $v$ in any of the two regions of $L$ adjacent to $u_v$ and by setting $\sigma(vu_v)=1$. Note that $L'$ is indeed a $k$-page book embedding because the edge $vu_v$ neither nests nor intersects any other edge (on any page).
\end{proof}
\fi

Let $G$ be a graph and let $\Gr$ be the graph obtained from $G$ by removing, exhaustively, every vertex of degree at most one. Note that every vertex of $\Gr$ has at least two neighbors.

\begin{corollary}[of Lemma~\ref{lem:isol}]\label{cor:redu}
Then instances $(G,k)$ and $(\Gr,k)$ are equivalent instances of \bt.
\end{corollary}

Let $F$ be a minimum feedback edge set of $\Gr$; note that $|F|=\fen(G)=\fen$, since none of the vertices or edges that are present in $G$ but not in $\Gr$, that is no element of $(V(G)\setminus V(\Gr))\cup (E(G)\setminus E(\Gr))$, is part of a cycle of $G$.
We define $G_F$ to be the tree $\Gr-F$. We denote with $\Ll$ and $\Ii$ the set of leaves and of vertices of degree at least 3 of $G_F$, respectively. The elements of $\Ii$ are also usually called {\it branching vertices}. 
 
Let $V_F$ be the set of all the vertices in $G_F$ that are adjacent to an edge of $F$ in $\Gr$, and note that $\Ll\subseteq V_F$.
 Let $B_F=\Ii\cup V_F$ and let $\mathcal{P}_F$ be the set of all maximal proper paths of $G_F$ for $B_F$.

\newcommand{\fun}{f(\fen)}
\newcommand{\gl}{2^{|\mathcal{P}|}|\mathcal{P}|}
\newcommand{\sh}{(|B|+1)\gl}
\newcommand{\lo}{\mathcal{P}_{\leq}}

  \begin{lemma}
\label{lem:count}
{\bf (a)} $|\Ii|\leq |\Ll|\leq |V_F|\leq 2\fen$ and {\bf (b)} $|\mathcal{P}_F|\leq |B_F|\leq 4\fen$.
\end{lemma}
\iflong
\begin{proof}
Let us first prove that $|\Ii|\leq |\Ll|$. It is immediate to note that in any tree the number of branching vertices is at most the number of leaves. 
Recall that $\Ll\subseteq V_F$. Finally recall is $V_F$ is defined as the set of all the vertices in $\Gr$ that are adjacent to an edge of $F$ and so the number of vertices in $V_F$ is at most twice the size of $F$. This completes the proof of {\bf (a)}.

By recalling that $B_F=\Ii\cup V_F$, we obtain that $|B_F|\leq 4\fen$.
Now let us consider the auxiliary graph $G'_F$ that is obtained from $G_F$ by 
contracting to one edge every path of $\mathcal{P}_F$.

Note that $G'_F$ is a tree and there are two 1-to-1 correspondences:  one between the vertices of $G'_F$ and $B_F$, and one between the edges of $G'_F$ and $\mathcal{P}_F$.
It follows that 
 $|\mathcal{P}_F|=|E(G'_F)|\leq |V(G'_F)| -1\leq |B_F|$. This proves {\bf (b)}.
\end{proof}
\fi

\iflong
We remark that, in view of Fact~\ref{fact:comp-fes}, all of these objects are efficiently computable.

\begin{observation}
\label{obs:gf}
It is possible to compute $\Gr$, a minimum feedback edges set $F$ of $\Gr$, $G_F$, $\Ll$, $\Ii$, $V_F$, $B_F$ and $\mathcal{P}_F$ in polynomial time.
\end{observation}

Thanks to Observation~\ref{obs:gf}, we are able to compute $B_F$, $\mathcal{P}$ and, by Lemma~\ref{lem:count}, all these sets have bounded size. 
\fi

At this point, the only issue to obtain a kernel of the desired size is that $\mathcal{P}_F$ might contain paths of unbounded length. 
 
\subparagraph{Long Path Insertion.}
 As our first step towards dealing with the long paths that remain in our instance, we show that yes-instances are preserved if we extend a path that is near-disjoint with the rest of the graph\ifshort\ (see Figure~\ref{fig:longer})\fi.

  \begin{lemma}
\label{lem:longer}
Let $G$ be a graph and $P$ be a path of length at least 2 such that $G$ and $\{P\}$ are near-disjoint. If $G\curlyvee P$ admits a $k$-page book embedding, then $G\curlyvee P'$ admits a $k$-page book embedding, where $P'$ is obtained from $P$ by subdividing an arbitrary edge of $P$ once.
\end{lemma}
\iflong
\begin{proof}
Consider a $k$-page book embedding $L$ for $G\curlyvee P$ and let $e=uv$ of be an edge of $P$. Now we want to show that we can always subdivide $e$ once. 
First, suppose there is no edge that is nested under $e$. In this case, we can delete $e$, add a vertex $w$ such that $u\prec w\prec v$, and insert the edges $uw$ and $wv$.

Since $P$ is of length at least two, one of the two endpoints of $e$ is an internal vertex of $P$, say it is $u$ (the case for $v$ can be proven symmetrically).
Consider any edge $e'=u'v'$ that is nested under $e$ and has minimum $u'$. If $u'=u$, that is $e$ and $e'$ are two consecutive edges of $P$, then let $R$ be the region of $L$ that has $u$ as the right endpoint. 
We delete $e$, add a vertex $w$ in the region $R$ (and so we have $w\prec u\prec v$). See Figure~\ref{fig:longer} (left) for this replacement.
If $u'\neq u$, then let $R$ be the region $(u,u')$ of $L$. We delete $e$, add a vertex $w$ in the region $R$ (and so we have $u\prec w\prec u'\prec v$). See Figure~\ref{fig:longer} (right) for this replacement.
For every of these cases, we set $\sigma(uw)=\sigma(wv)=\sigma(e)$.
\end{proof} 
\fi

\begin{figure}[h]
\small
\begin{minipage}{0.5\linewidth}
\centering
\begin{tikzpicture}[scale=0.7] 
\draw (-3,0) to[out=40,in=140] (3,0) (-3,0) to[out=30,in=150] (1,0); \draw[very thick,dotted] (-4,0) to[out=50,in=130] (-3,0) (-4,0) to[out=50,in=130] (3,0);\draw[fill=black] (-3,0) circle (2pt)(1,0) circle (2pt)(3,0) circle (2pt)(-4,0) circle (2pt);\node[below] at (-3,-0.2) {$u=u'$}; \node[below] at (1,-0.2) {$v'$}; \node[below] at (3,-0.2) {$v$};\node[below] at (0,1) {$e$};\node[below] at (-1,0.5) {$e'$};\node[below] at (-4,-0.2) {$w$};
\end{tikzpicture}
\end{minipage} \begin{minipage}{0.5\linewidth}
\centering
\begin{tikzpicture}[scale=0.7] 
\draw 
(-4,0.5) to[out=0, in=150] (-3,0)
(-3,0) to[out=40,in=140] (3,0) 
(-1,0) to[out=30,in=150] (2,0); 
\draw[very thick,dotted] (-3,0) to[out=30,in=150] (-2,0)(-2,0) to[out=30,in=150] (3,0);\draw[fill=black] (-3,0) circle (2pt)(-1,0) circle (2pt)(3,0) circle (2pt)(2,0) circle (2pt)(-2,0) circle (2pt);\node[below] at (-3,-0.2) {$u$}; \node[below] at (-1,-0.2) {$u'$}; \node[below] at (3,-0.2) {$v$};\node[above] at (0,1.2) {$e$};\node[below] at (0.5,0.5) {$e'$};\node[below] at (-2,-0.2) {$w$};\node[below] at (2,-0.2) {$v'$};
\end{tikzpicture}
\end{minipage}
\caption{The subdivision of the edge $e=uv$ when $u$ is a vertex of degree 2 while containing the edge $e'=u'v'$ by Lemma~\ref{lem:longer}:
when $e'$ is incident to $u$ (left) and when $e'$ is not incident to $u$ (right).}\label{fig:longer}
\end{figure}
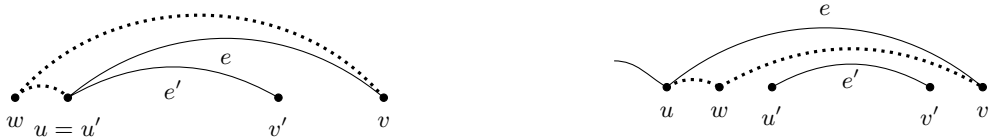

By exhaustively applying Lemma~\ref{lem:longer}, we can 
transform paths of length at least two into arbitrarily long proper paths while preserving yes-instances. To obtain our kernel, we will however need to shorten sufficiently long paths while preserving yes-instances. The following lemma allows us to handle this in case where all the considered paths are sufficiently long.
 
  \begin{lemma}
\label{lem:fix}
 Let $k\geq 3$, $G$ be a graph and $\mathcal{P}$ be a set of paths each of length more than $(|V(G)|+1)\gl$ such that $G$ and $\mathcal{P}$ are near-disjoint.
Then, $G\curlyvee \mathcal{P}$ admits a $k$-page book embedding if and only if 
there exists a set of paths $\mathcal{P}'$ each of length at most $(|V(G)|+1)\gl$ that serves as $\mathcal{P}$ such that $G\curlyvee \mathcal{P}'$ also admits a $k$-page book embedding.
\end{lemma}
\iflong
\begin{proof}
We start with the forward direction. 
Consider a $k$-page book embedding $L$ for $G\curlyvee \mathcal{P}$ and let $L_G$ be the restriction of $L$ to $G$.
First we establish that, given any region $R$ of $L_G$, in $L$
we can replace $\mathcal{P}$ with a set of paths $\mathcal{P}'$ that serves as $\mathcal{P}$ such that every path of $\mathcal{P}'$ touches the region $R$ most $\gl$ times.

If every path of $\mathcal{P}$ touches $R$ at most $\gl$ times, then the statement is true for $\mathcal{P}'=\mathcal{P}$. Suppose otherwise, and let $\mathcal{R}$ be the set of all paths in $\mathcal{P}$ that touch the region $R$. 

For every $P\in\mathcal{R}$, 
say $P$ has vertex set $\{u_1,\ldots,u_t\}$ for some $t\geq 2$ with edges $u_iu_{i+1}$ for every $i\in [t-1]$,
and let $Z_P$ be the set of all vertices of $P$ in the region $R$ that have either minimum or maximum label. Clearly, since $P\in \mathcal{R}$, the set $Z_P$ must have either one or two elements.
If $Z_P=\{a\}$, we define $u_P$ and $v_P$ both to be equal to $a$.
If $Z_P=\{a,b\}$ with $a\prec b$, we define $u_P=a$ and $z_P=b$.

The idea now is to delete the subpath $P_{uv}$ of $P$ between $u_P$ and $v_P$ and create a path $S_P$ such that $S_P$ {\bf (1)} has bounded length, {\bf (2)} serves as $P_{uv}$ and {\bf (3)} is completely contained in $R$. Intuitively, the length of the path $S_P$ will be at most the number of vertices between $u_P$ and $v_P$: we will show that we will always be able to ``jump'' at least one vertex with an edge. For the argument to work, the replacement of the path segments $P_{uv}$ with $S_P$ will be carried out using a recursive strategy where the next path to be replaced is selected based on a specific condition.
  
We initiate by having $\mathcal{R}$ contain all the paths that touch region $R$, as introduced earlier. Moreover, we set $\mathcal{R}'=\emptyset$ to be the empty set, and we will use $\mathcal{R}'$ to store paths which have already been processed. 
 Inductively, until $\mathcal{R}= \emptyset$, we consider a path $P\in \mathcal{R}$ such that for every other path $P'\in \mathcal{R}$ we have $(u_{P'},v_{P'})\not\subseteq (u_P,v_P)$, \ie, the interval between $u_{P'}$ and  $v_{P'}$ does not fully contain the corresponding interval for any other path in $\mathcal{R}$. 
Let $M_P$ be the set defined as follows 
$$M_P=(\bigcup_{P''\in \mathcal{R}} \{u_{P''},v_{P''}\}\cup \bigcup_{P'\in \mathcal{R}'}V(P'))\cap (u_P,v_P)$$
The set $M_P$ represents the set of all vertices, that are present at this stage, in the interval $(u_P,v_P)$.

If $M_P=\emptyset$ then we either do not do anything if $u_P=v_P$ or add the edge $u_Pv_P$ and set $\sigma(u_Pv_P)=1$. 
Suppose $M_P\neq\emptyset$ and let $M_P=\{v_0,\ldots,v_t\}$ for some $t\geq 0$ and assume $v_0\prec \ldots \prec v_t$.
We add a path $S_P$ in the following way: if $t=0$, we add the edge $u_Pv_P$ and set $\sigma(u_Pv_P)=\nin ([h]\setminus \sigma(v_0))$.
If $t\geq 1$, $S_P$ has $t$ internal vertices $\{u_1,\ldots, u_t\}$ such that $v_0\prec u_1 \prec v_1 \prec \ldots \prec u_t \prec v_t$. 
We add the edges $u_Pu_1$, $u_tv_P$ and $u_iu_{i+1}$, for every $i\in [t-1]$. 
Finally we set $\sigma(u_Pu_1)=\nin ([h]\setminus \sigma(v_0))$, 
$\sigma(u_tv_P)=\nin ([h]\setminus \sigma(v_t))$
and $\sigma(u_iu_{i+1})=\nin ([h]\setminus \sigma(v_i))$ for every $i\in [t-1]$.

Now we show that the way we assigned pages to the edges of $S_P$ do not any create edge crossings.
Since $h\geq 3$ and every vertex in $M_P$ has at most two incident edges (thus resulting in $|\sigma(v_i)|\leq 2$), the sets $[h]\setminus \sigma(v_0)$, $[h]\setminus \sigma(v_t)$ and $[h]\setminus \sigma(v_i)$ for every $i\in [t-1]$ are not empty:
since the edge of $S_P$ used to {\it jump} $v_i$ can not be assigned to any of the pages in $\sigma(v_i)$, we can always assign a page in $[h]$ to every edge of $S_P$ in such a way no edge crossing is created.
We can finally define $P'$ as the path obtained from $P$ by replacing $P_{uv}$ with $S_P$.

Now we are left to show how many vertices are necessary to construct $S_P$ at each step: recall that this number is equal to $|M_P|-1$.
Let $i$ be an integer such that $0\leq i\leq |\mathcal{R}|\leq |\mathcal{P}|$. 
We aim to evaluate $T(i+1)$, that is the maximum number of internal vertices of the $(i+1)$-th subpath $S_P$.
Let $N_P=(\bigcup_{P''\in \mathcal{R}} \{u_{P''},v_{P''}\})\cap (u_P,v_P)$ and $O_P=(\bigcup_{P'\in \mathcal{R}'}V(P'))\cap (u_P,v_P)$ and note that $M_P=N_P\cup O_P$.

\begin{figure}
\centering
\begin{tikzpicture}[scale=0.9] 
\draw[very thick] (-6,0)--(6,0);\draw[dashed,very thick,red](-6,0) to[out=50,in=130] (0,0)(-4,0) to[out=50,in=130] (-2,0)(2,0) to[out=50,in=130] (6,0);\draw[dashed,very thick,blue](-5,0) to[out=50,in=130] (2,0);\draw[dashed,very thick,green!60!black](-2,0) to[out=50,in=130] (5,0)(0,0) to[out=50,in=130] (4,0);\draw[very thick,blue](-3,0) to[out=50,in=130] (-1,0)(-1,0) to[out=50,in=130] (1,0)(3,0) to[out=50,in=180] (6,1);\draw[very thick,green!60!black](1,0) to[out=50,in=130] (3,0)(-6,1) to[out=0,in=130] (-3,0);\draw[fill=white] (0,0)(-6,0) circle [radius=2pt](-5,0) circle [radius=2pt](-4,0) circle [radius=2pt](-2,0) circle [radius=2pt](0,0) circle [radius=2pt](2,0) circle [radius=2pt](4,0) circle [radius=2pt](5,0) circle [radius=2pt](6,0) circle [radius=2pt];\draw[fill=black](-3,0) circle [radius=3pt](-1,0) circle [radius=3pt](1,0) circle [radius=3pt](3,0) circle [radius=3pt];
\node[below] at (-3,-0.2) {$u_P$}; \node[below] at (3,-0.2) {$v_P$};
\node[below] at (-2,-0.2) {$v_0$};
\node[below] at (0,-0.2) {$v_1$};
\node[below] at (2,-0.2) {$v_2$};
\node[below] at (-1,-0.2) {$u_1$};
\node[below] at (1,-0.2) {$u_2$};
\end{tikzpicture}
\caption{An example where black vertices and full edges are part of $P'$, the path obtained from $P$ by replacing the subpath $P_{uv}$ between $u_P$ and $v_P$ with $S_P$.
Note that edges with different colors belong to different pages of the book embedding.}\label{fig:st}
\end{figure}
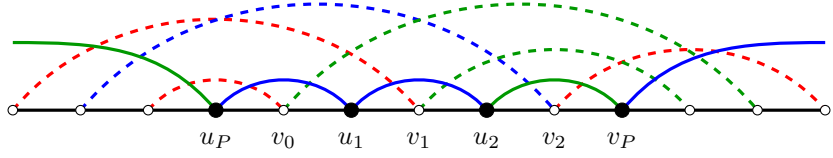

Let us first analyze the case $i=0$, that is $T(1)$. Since $\mathcal{R}'$ is empty, $S_P$ creates $|N_P|-1$ new vertices and so we have $T(1)= |N_P|-1\leq |\mathcal{R}|\leq |\mathcal{P}|$.

Let us consider the case $i\geq 1$. The set $N_P$ contains at most one vertex per path of $\mathcal{R}$: this is ensured by the choice of $P$. For this reason, we have $|N_P|\leq |\mathcal{R}|\leq |\mathcal{P}|$.
The set $O_P$ might contain every vertex in $R$, both endpoint and internal, of every path of $\mathcal{R}'$. For this reason we have $|O_P|\leq 2|\mathcal{R}'|+\sum_{j=1}^{i} T(j)$.
Now we have that
\[
\begin{array}{cclc}
T(i+1) & \leq & |\mathcal{R}|+2|\mathcal{R}'|+\sum_{j=1}^{i} T(j) \\
& \leq & 2|\mathcal{P}|+\sum_{j=1}^{i} T(j) \\
& \leq & T(i)+2|\mathcal{P}|+\sum_{j=1}^{i-1} T(j) \\
& \leq & 2T(i) \\
& \leq & 2^iT(1) \\
& \leq & 2^i|\mathcal{P}|
\end{array}
\]
To summarize we have found out that $T(i+1)\leq 2^i|\mathcal{P}|$ for every $0\leq i\leq |\mathcal{R}|\leq |\mathcal{P}|$.
In particular, we have found a set $\mathcal{P}'$ that serves as $\mathcal{P}$, where every path in $\mathcal{P}'$ touches the region $R$ at most $\gl$ times.
At this point we set $\mathcal{R}:=\mathcal{R}\setminus \{P\}$ and $\mathcal{R}':=\mathcal{R}'\cup \{P'\}$ and, if possible, select another path of $\mathcal{R}$.

Note that by the construction described in the paragraphs above, applying this procedure to a region $R$ does not increase the number of times paths of $\mathcal{P}$ touch regions that are not $R$. This means we can apply this procedure on one region at a time until the claim holds for every region and thus satisfying the statement.
This provides a $k$-page book embedding of $G\curlyvee \mathcal{P}'$ with the desired property.

Let us consider the backwards direction.
Suppose there exists a set $\mathcal{P}'$ that serve as $\mathcal{P}$, where every path in $\mathcal{P}'$ has length at most $(|V(G)|+1)\gl$ and $G\curlyvee \mathcal{P}'$ admits a $k$-page book embedding.
Since each path of $\mathcal{P}$ is of length more than $(|V(G)|+1)\gl$, we apply Lemma~\ref{lem:longer} the appropriate number of times on the paths of $\mathcal{P'}$ so that the lengths of the paths coincide with the ones of $\mathcal{P}$ and obtain that $G\curlyvee \mathcal{P}$ admits a $k$-page book embedding.
\end{proof}
\fi

The following result is a direct consequence of combining Lemma~\ref{lem:longer} and~\ref{lem:fix}.

  \begin{corollary}
\label{cor:exa}
Let $k\geq 3$, $G$ be a graph and $\mathcal{P}$ be a set of paths each of length more than $(|V(G)|+1)\gl$ such that $G$ and $\mathcal{P}$ are near-disjoint.
Then, $G\curlyvee \mathcal{P}$ admits a $k$-page book embedding if and only if 
$G\curlyvee \mathcal{P}'$ admits a $k$-page book embedding where $\mathcal{P}'$ is a set of paths each of length exactly $(|V(G)|+1)\gl$ that serves as $\mathcal{P}$.
\end{corollary}
\iflong
\begin{proof}
Let us start with the forward direction. Suppose that $G\curlyvee \mathcal{P}$ admits a $k$-page book embedding. 
We apply the forward direction of Lemma~\ref{lem:fix} and that $G\curlyvee \mathcal{P}'$ admits a $k$-page book embedding where every path of $\mathcal{P}'$ has length at most $(|V(G)|+1)\gl$.
Now by Lemma~\ref{lem:longer}, we obtain obtain that $G\curlyvee \mathcal{P}''$ admits a $k$-page book embedding where every path of $\mathcal{P}''$ has length exactly $(|V(G)|+1)\gl$.

The reverse direction follows directly from Lemma~\ref{lem:longer}.
\end{proof}
\fi

\subparagraph{Putting Everything Together.}
At this point, we have all the ingredients needed to establish the fixed-parameter tractability of \bt\ with respect to the feedback edge number.
 
\ifshort
\begin{proof}[Proof Sketch for Theorem~\ref{thm:fenbig}]
Let $(G,k)$ be an input of \bt. We construct $\Gr$ by exhaustively deleting vertices of degree $1$ to obtain an equivalent instance of \bt, as per Lemma~\ref{lem:isol}. 

At this point we start a loop that repeats at most $|\mathcal{P}_F|\leq 4\fen$ times: at each iteration either we obtain a kernel of the desired size and the algorithm ends, or the size of $\mathcal{P}$ is reduced by at least one. 
If $\lo=\emptyset$ holds, that is, all the paths in $\mathcal{P}$ have length more than $\sh$, Lemma~\ref{lem:fix} can be applied to obtain a kernel of the desired size and the algorithm ends.
If $\lo\neq \emptyset$ holds; denote with $V_{\leq}$ the set of vertices in paths of $\lo$.
Note that $|V_{\leq}|\leq \sh|\lo|\leq (|B|+1)2^{|\mathcal{P}|}|\mathcal{P}|^2$.
In this case, We update the sets $B$ and $\mathcal{P}$  by setting the former to $B\cup V_{\leq}$ and the latter to $\mathcal{P}\setminus \lo$, and enter the next iteration of the loop.

To conclude the proof, it suffices to provide an upper bound to the size of a graph created in this way.
\end{proof}
\fi

\iflong \begin{proof}[Proof of Theorem~\ref{thm:fenbig}]
Let $(G,k)$ be an input of \bt. We establish fixed-parameter tractability by constructing a problem kernel in polynomial time.
By Observation~\ref{obs:gf}, we compute $\Gr$, $F$, $G_F$ and $B_F$.
We set $B:=B_F$ and let $\mathcal{P}:=\mathcal{P}_F$ be the set all maximal proper paths of $G_F$ for $B$. Note that, given $B$, the set $\mathcal{P}_F$ can be computed using Observation~\ref{obs:gf}.
Let $\lo$ be the set of all the paths of $\mathcal{P}$ of length at most $\sh$.

At this point we start a loop that repeats at most $|\mathcal{P}|\leq 4\fen$ times: at each iteration either we obtain a kernel of the desired size and the algorithm ends, or the size of $\mathcal{P}$ is reduced by at least one. 
If $\lo=\emptyset$ holds, that is, all the paths in $\mathcal{P}$ have length more than $\sh$, Lemma~\ref{lem:fix} can be applied to obtain a kernel of the desired size and the algorithm ends.
If $\lo\neq \emptyset$ holds; denote with $V_{\leq}$ the set of vertices in paths of $\lo$.
Note that $|V_{\leq}|\leq \sh|\lo|\leq (|B|+1)2^{|\mathcal{P}|}|\mathcal{P}|^2$.
In this case, We update the sets $B$ and $\mathcal{P}$  by setting the former to $B\cup V_{\leq}$ and the latter to $\mathcal{P}\setminus \lo$, and enter the next iteration of the loop.

To conclude the proof, it suffices to provide an upper bound to the size of a graph created in this way. The largest graph obtained by this construction results from there being $|\mathcal{P}|$ iterations, whereas in each iteration there is only a single path in $\lo$ and this path is of maximum length.
 Let $S(i+1)$ be an upper bound on the number of vertices that have been added to $B$ after the $i$-th step of the recursion.
For $i=0$, we have that $S(1)=|B_F|$.

Let us consider the case $i\geq 1$. Together with the vertices that were present after the $(i-1)$-th step, that is $S(i)$, we also have to consider the vertices of a unique path having maximum length allowed at this step, that is $2^{|\mathcal{P}|-i}(|\mathcal{P}|-i)(S(i)+1)$. Now we have that: 
\[
\begin{array}{cclc}
S(i+1) & \leq & S(i)+2^{|\mathcal{P}|-i}|(\mathcal{P}|-i)(S(i)+1) \\
& \leq & S(i)+2^{|\mathcal{P}|-i}(|\mathcal{P}|-i)S(i)+2^{|\mathcal{P}|-i}(|\mathcal{P}|-i) \\
& \leq & 3*2^{|\mathcal{P}|-i}(|\mathcal{P}|-i)S(i) 
  \end{array}
\]

Hence, the total size of the obtained kernel can be upper-bounded by $(3*2^{|\mathcal{P}|-i}(|\mathcal{P}|-i))^{|\mathcal{P}|}\cdot S(1)$, which is at most $2^{\bigoh(\fen(G)^2)}$.
          \end{proof}\fi

\subsection{An FPT-algorithm for \SH{} using the Feedback Edge Number}
\label{sec:subhamfen}

In this section, we provide a linear kernel for \SH{} parameterized by
the feedback edge number. The main idea is to reduce the size of the
tree $G-F$, where $F$ is minimum feedback edge set of the input
graph $G$. To do so we need the following simple corollary and lemma
that allow us to bound the number of leaves and vertices of degree at
most two in $G$. We have already seen in Lemma~\ref{lem:isol} that we can
remove vertices of degree at most one. The next lemma allows us to bound the number of vertices of degree two.
  \begin{lemma}
\label{lem:sh-fes-degtwo}
  Let $G$ be a graph and $P$ be a path of length $4$ in $G$ such that
  all inner vertices of $P$ have degree two. Let $G'$ be a graph
  obtained from $G$ by contracting any edge on $P$.
  Then, $G$ is subhamiltonian if and only if so is $G'$.
\end{lemma}
\begin{proof}
  Let's assume that $G'$ is subhamiltonian with witness $(G_H,H)$
  and $P'=[v_{b},v_1,v_2,v_{e}]$ is a path $P$ after
  contracting. Due to the fact that $deg_{G}(v_1) = deg_{G}(v_2) =
  2$ and $2 \leq deg_{G_H}(v_1) ,deg_{G_H}(v_2) \leq 4$, there exist a
  face $f_H$ in drawing $G'_H$ such that $(v_1,v_2) \in E(f_H)$ and
  $E(f_H) \cap E(H) \neq \emptyset$.  From
  Observation~\ref{obs:subham}~\enref{2} and~\enref{1} we obtain that $G$ is subhamiltonian.
  
  Let's assume that $G$ is subhamiltonian with witness $(G_H,H)$. Let $P$ be a path $[v_b, v_1, v_2, v_3, v_e]$ and $D$ be a drawing of $G$ that respects $H$. 
  All inner vertices from $V(P)$ have degree two, which implies that there exists a face $f$ of $D$ such that $V(P)\subseteq V(f)$. 
  From Lemma~\ref{lem:edge} applied to $v_bv_e$ and a face $f$, we get new witness $(G_{H'},{H}')$ and $H'$ crosses $v_bv_e$ at most in two points.
  There cannot be three edges which have one ending in $v_1$, $v_2$ and $v_3$ and crosses $v_bv_e$, so at least one pair of varieties $(v_b, v_2)$, $(v_1, v_3)$ or $(v_2, v_e)$ are in the same face together, so from Observation~\ref{obs:face} we can connect them. From Observation~\ref{obs:subham}~\enref{1} we obtain that $G'$ is subhamiltonian.
\end{proof}

We are now ready to provide our kernel for \SH{}.
  \begin{theorem}
  \label{thm:fenhamkernel}
  \SH{} parameterized by the feedback edge number $k$ admits a kernel 
  with at most $12k-8$ vertices and at most $14k-9$ edges.
\end{theorem}
\begin{proof}
  Let $G$ be a connected graph, \ie, the given instance of \SH{}.
  We first use Lemma~\ref{lem:isol} to ensure that $G$
  has no leaves. We now compute a minimum feedback edge set $F
  \subseteq E(G)$ for $G$ using Fact~\ref{fact:comp-fes}. Let $T$ be the tree
  $G-F$. Note that $T$ has at most $2|F|$ many leaves, since every
  leaf of $T$ must be adjacent to an edge in $F$. This also implies
  that $T$ has at most $|L|-2$ vertices of degree at least
  $3$, where $L$ is the set of all leaves of $T$. Therefore, it only
  remains to obtain an upper bound on the
  vertices having degree exactly two in $T$. We say that a path $P$ in
  $T$ is
  \emph{proper} if it is an inclusion-wise maximal path in $T$ having
  only inner vertices of degree two in $G$.
  Because of Lemma~\ref{lem:sh-fes-degtwo}, we can assume that any
  proper path in $T$ has length at most three. Also note that every vertex having degree
  two in $T$ is an inner vertex of such a maximal path. Moreover,
  since every proper path must have both of its endpoints in
  $V(F)\cup B$, where $B$ is the set of all vertices having degree at
  least three in $T$, the number of distinct proper paths in $T$ is equal
  to the number of edges in a tree with $|V(F)\cup B|$ many
  vertices. Therefore, the number of proper path in $T$ is at most
  $|V(F)\cup B|-1\leq 2|F|+2|F|-2-1=4|F|-3$. Since every proper path
  contains at most two vertices of degree two in $T$, we obtain that $T$ contains
  at most $2(4|F|-3)=8|F|-6$ vertices of degree two. Altogether, $T$
  contains at most $2|F|+2|F|-2+8|F|-6=12|F|-8$ vertices and at most
  $12|F|-9$ edges. Therefore, $G$ has at most $12|F|-8$ vertices and
  at most $14|F|-9$ edges. Finally, the time required to obtain the
  kernel is at most $\bigoh(|V(G)|+|E(G)|)$.
\end{proof}

Theorem~\ref{thm:fen} now follows directly from Theorems~\ref{thm:fenhamkernel} and~\ref{thm:fenbig}.
   Moreover, by combining Theorem~\ref{thm:fenhamkernel} with the subexponential algorithm of Corollary~\ref{cor:main}, we can slightly strengthen our main result as follows.
\fi

\begin{corollary}
\SH{} can be solved in time $2^{\bigoh(\sqrt{k})}\cdot n^{\bigoh(1)}$, where $k$ is the feedback edge number of the input graph.
\end{corollary}

\section{Concluding Remarks}

While our main algorithmic result settles the complexity of computing 2-page book embeddings under the exponential time hypothesis, many questions remain when one aims at computing $k$-page book embeddings for a fixed $k$ greater than $2$. To the best of our knowledge, even the existence of a single-exponential algorithm for this problem is open.

In terms of the problem's parameterized complexity, it is natural to ask whether one can obtain a generalization of Theorem~\ref{the:sh-fpt-tw} for computing $k$-page book embeddings when $k>2$. In fact, it is entirely open whether computing, e.g., 4-page book embeddings is even in \XP\ when parameterized by the treewidth. In this sense, our positive result for the feedback edge number can be seen as a natural step on the way towards finally settling the structural boundaries of tractability for computing page-optimal book embeddings.

\section*{Acknowledgements}
Robert Ganian was supported by Project No. Y1329 of the Austrian
Science Fund (FWF) and Project No. ICT22-029 of the Vienna Science
Foundation (WWTF). Sebastian Ordyniak was supported by the Engineering
and Physical Sciences Research Council (EPSRC) (Project EP/V00252X/1).

%\bibliographystyle{plainurl}
%\bibliography{GD-ref}

\end{document}